\documentclass{article}

\usepackage{graphicx} 
\usepackage[left=2cm,right=2cm,top=2cm,bottom=2cm]{geometry} 
\usepackage[x11names]{xcolor} 
\usepackage{enumitem}
\usepackage{pifont}  
\usepackage{comment} 
\usepackage{float} 
\usepackage{multicol} 

\usepackage{authblk}

\usepackage{soul}

\usepackage{mathtools} 
\usepackage{esvect} 
\usepackage{stackrel,amssymb} 
\usepackage{amsmath} 
\usepackage{amsthm} 
\usepackage{mathrsfs} 
\usepackage{cancel} 
\usepackage{mathabx} 
\usepackage{stackengine} 
\usepackage{tablefootnote}
\usepackage{caption}

\usepackage{pgfplots}
\pgfplotsset{compat=1.18}

\usepackage{slashbox} 
\usepackage{diagbox} 
\usepackage{multirow} 

\usepackage{graphicx}      
\usepackage{eso-pic}       
\graphicspath{ {./images/} }

 \usepackage[
     backend=biber,
     style=numeric-comp,
     sorting=none,
   ]{biblatex}
\usepackage{ifthen}
\usepackage{stackengine} 
\usepackage{scalerel}

\usepackage{listings} 

\definecolor{codegreen}{rgb}{0,0.6,0}
\definecolor{codegray}{rgb}{0.5,0.5,0.5}
\definecolor{codepurple}{rgb}{0.58,0,0.82}
\definecolor{backcolour}{rgb}{0.95,0.95,0.92}

\lstdefinestyle{mystyle}{
    backgroundcolor=\color{backcolour},   
    commentstyle=\color{codegreen},
    keywordstyle=\color{magenta},
    numberstyle=\tiny\color{codegray},
    stringstyle=\color{codepurple},
    basicstyle=\ttfamily\footnotesize,
    breakatwhitespace=false,         
    breaklines=true,                 
    captionpos=b,                    
    keepspaces=true,                 
    numbers=left,                    
    numbersep=5pt,                  
    showspaces=false,                
    showstringspaces=false,
    showtabs=false,                  
    tabsize=2
}

\usepackage{hyperref}
\hypersetup{
    colorlinks=true,
    linkcolor=black,
    citecolor=black,      
    urlcolor=black,
    pdftitle={Multipartite entanglement spread},
    pdfpagemode=FullScreen,
}

\usepackage{lscape}
\usepackage{graphicx}

\makeatletter 

\newcommand*\Neginternal[3]{\mathpalette\Neg@{{#1}{#2}{#3}}}
\newcommand*\Neg@[2]{\Neg@@{#1}#2}
\newcommand*\Neg@@[4]{%
  \mathrel{\ooalign{%
    $\m@th#1#4$\cr
    \hidewidth$\m@th#3{#1}\mkern\muexpr#2*2$\hidewidth\cr
  }}%
}

\newcommand*\negslash[1]{\m@th#1\not\mathrel{\phantom{=}}}
\newcommand*\snegslash[1]{\rotatebox[origin=c]{60}{$\m@th#1-$}}
\newcommand*\ssnegslash[1]{\rotatebox[origin=c]{60}{$\m@th#1{\dabar@}\mkern-7mu{\dabar@}$}}
\newcommand*\sssnegslash[1]{\rotatebox[origin=c]{60}{$\m@th#1\dabar@$}}
\makeatother  

\newcommand{\lat}[1]{\textsl{#1}}
\newcommand{\apriori}{\lat{a~priori} }
\newcommand{\ie}{\lat{i.e.}~}

\newcommand{\eg}{\lat{e.g.}~}
 
\renewcommand{\rm}[1]{\text{#1}}

\newcommand{\pa}[1]{\left( #1 \right)}
\newcommand{\pac}[1]{\left[ #1 \right]}
\newcommand{\paa}[1]{\left\{ #1 \right\}}

\DeclareMathOperator{\Tr}{Tr}          
\DeclareMathOperator{\tr}{Tr}          

\DeclareMathOperator{\id}{id}

\renewcommand{\tilde}[1]{\widetilde{#1}}

\DeclareMathOperator{\flip}{\ast}

\DeclareMathOperator{\op}{\circ}

\newcommand{\ot}{\otimes}

\DeclareMathOperator{\spec}{Spec}

\newcommand{\bra}[1]{\left\langle #1 \right\vert} 			
\newcommand{\ket}[1]{\left\vert #1 \right\rangle} 			
\newcommand{\mean}[1]{\left\langle #1 \right\rangle}

\newcommand{\abs}[1]{\left\lvert #1 \right\rvert}

\renewcommand{\bar}[1]{\overline{#1}}

\newcommand\eqdef{\mathrel{\overset{\makebox[0pt]{\mbox{\normalfont\tiny\sffamily def}}}{=}}}

\newcommand{\bb}[1]{\mathbb{#1}}

\renewcommand{\H}{\mathcal{H}}

\newcommand{\MD}[1]{\cal{M}}

\newcommand{\K}{\mathrm{K}}

\newcommand{\RM}{\mathrm{RM}}

\newcommand{\PT}{\mathrm{PT}}
\newcommand{\GHZ}{\mathrm{GHZ}}

\theoremstyle{plain}
\newtheorem{theo}{Theorem}[section]
\newtheorem{prop}[theo]{Proposition}
\newtheorem{lem}[theo]{Lemma}
\newtheorem{rem}[theo]{Remark}
\newtheorem{ex}[theo]{Example}
\newtheorem{defi}[theo]{Definition}
\newtheorem{cor}[theo]{Corollary}

\renewcommand{\cal}[1]{\mathcal{#1}}

\definecolor{cadmiumred}{rgb}{0.89, 0.0, 0.13}
\newboolean{ifprintcolor}	

\ifthenelse{\boolean{ifprintcolor}}{}{}

\newcommand{\LU}{\mathsf{LU}}
\newcommand{\LO}{\mathsf{LO}}
\newcommand{\LOCC}{\mathsf{LOCC}}

\newcommand{\I}{%
  \textup{\uppercase\expandafter{\romannumeral 1}}%
}
\newcommand{\II}{%
  \textup{\uppercase\expandafter{\romannumeral 2}}%
}

\newcommand{\simLU}{\underset{\LU}{\sim}}
\newcommand{\toLO}{\underset{\LO}{\longrightarrow}}

\DeclareMathOperator{\Span}{Span}

\DeclareMathOperator{\conn}{c}

\newcommand\restr[2]{{
  \left.\kern-\nulldelimiterspace 
  #1 
  \vphantom{\big|} 
  \right|_{#2} 
  }}

\newcommand{\cG}{\mathcal{G}}
\newcommand{\pos}{\mathcal{P}}
\newcommand{\dom}{\triangleleft}
\newcommand{\Mel}{\mathrm{M}}
\newcommand{\SpreadOne}{\mathsf{P}}
\newcommand{\SpreadTwo}{\mathsf{Q}}
\newcommand{\Seq}{\mathsf{Spec}}
\newcommand{\Comp}{\mathsf{Comp}}
\newcommand{\SpecFun}{\mathsf{SpecFun}}

\usepackage[normalem]{ulem}

\title{\bf{Multipartite entanglement spreads}}
\author[1]{Sylvain Carrozza\thanks{sylvain.carrozza@ube.fr}}
\author[1,2]{Johann Chevrier\thanks{johann.chevrier@ube.fr}}
\author[2]{Luca Lionni\thanks{luca.lionni@ens-lyon.fr}}
\affil[1]{\normalsize{Universit\'{e} Bourgogne Europe, CNRS, IMB UMR 5584, F-21000 Dijon, France}}
\affil[2]{\normalsize{CNRS, ENS de Lyon, LPENSL, UMR5672, 69342 Lyon cedex 07, France}}
\date{\today}

\begin{document}

\maketitle

\begin{abstract}
    \noindent The entanglement structure of a bipartite pure state is characterized by its entanglement spectrum, or equivalently, by the family of R\'{e}nyi-$k$ entanglement entropies with integer $k$. Taking the difference of two such entropies defines a quantity known as an \emph{entanglement spread}, that is monotonous under local operations, and that vanishes if and only if the entanglement spectrum of the state is flat. We introduce a multipartite generalization of this notion which, for a fixed number $D\geq 3$ of parties, results in a collection of maps indexed by pairs of $D$-edge-colored graphs (obeying some condition). Each member of this collection takes the form of a difference of two multipartite R\'{e}nyi entanglement entropies, each associated to a local unitary polynomial invariant (known in this context as a \emph{trace-invariant}), and is shown to be monotonous under local operations. We then argue that a previously introduced family of so-called \emph{hypergraph-tensor} (HT) states can be understood as a multipartite counterpart to the family of flat bipartite states. Indeed, we first prove that, when $D=3$, a tripartite pure state is HT if and only if it minimizes a particular (infinite) family of entanglement spreads. Second, we embed the set of HT states into the much larger family of \emph{spectral hypergraph-tensor} (SHT) states, defined by a new Ansatz we introduce. We then prove that HT states can be uniquely characterized as the minimizers of some fixed entanglement spread among SHT states, for arbitrary $D\geq 3$. We also investigate SHT states in their own right: in particular, we characterize their orbits under local unitary transformations. Finally, along the way, the asymptotic large-$N$ expectation values of a number of entanglement spreads are computed in the Haar-random state of local dimension $N$. 
\end{abstract}

\setcounter{tocdepth}{2}
\tableofcontents

\section{Introduction}

This paper is part of a series of works by the same authors dedicated to the study of multipartite entanglement \cite{Carrozza:2026qcf, Factorization2026}. Our approach is based on the \emph{local unitary} framework, whereby inequivalent multipartite entanglement structures are mathematically identified to distinct orbits of the group of \emph{local unitary} ($\LU$) transformations. Such orbits can be separated by a certain family of $\LU$-invariant polynomials known as \emph{trace-invariants} \cite{Bonzom:2012hw,gurau_random_2017,Collins2025}, which take center stage in our work. In the absence of a sufficiently universal concept of singular value for higher order tensors, evaluations of such trace-invariants on multipartite pure states can indeed play a similar role as the bipartite notion of \emph{entanglement spectrum} (which can be defined as the collection of non-zero singular values associated with the matrix representation of a bipartite pure state). More precisely, one can understand trace-invariant evaluations as multipartite generalizations of power sums of the entanglement spectrum. 

This perspective is quite rewarding: not only can one label each orbit of $\LU$-equivalent pure states by the collection of these trace-invariant evaluations, they can also be used to devise simple tests of separability in multipartite state spaces, and to characterize order relations in resource theories of entanglement (such as those based on the \emph{Local Operations} ($\LO$) \cite{Schmid:2020pgv, PhysRevA.53.2046} and \emph{Local Operations and Classical Communication} ($\LOCC$) \cite{PhysRevA.59.1070, PhysRevLett.83.436, PhysRevLett.83.3566, PhysRevLett.127.150503, Chitambar2014} frameworks). We refer the interested reader to \cite{Carrozza:2026qcf} for \eg: a characterization of separable $D$-partite pure states by means of so-called \emph{genuinely $D$-partite} trace-invariants (see Sec.~2.5, Thm.~2.18), a complete characterization of the $\LO$ preorder on $D$-partite pure states (see Sec.~2.7, Prop.~2.28), as well as a partial characterization of the $\LOCC$ preorder (see Sec.~5). 

While not in itself new (see \cite{PhysRevA.58.1833,AnthonySudbery_2001,Makhlin2002Nonlocal,10.1063/5.0028856,Vrana:2011ehx,Vrana:2011gaq,Turner2017} for a selection of relevant earlier works, and \cite{Dong:2021clv,Penington:2022dhr,SciPostPhys.20.3.086,Gadde:2022cqi,Gadde:2023zzj,Hu:2026bhg,Szalay:2026yli,Mallick:2026bji} for more recent developments that are complementary and/or closely related to ours), this approach to entanglement classification based on polynomial invariants becomes particularly effective with the help of a combinatorial toolbox allowing one to seamlessly represent and manipulate trace-invariants. This toolbox originates from a correspondence between combinatorial objects known as \emph{$D$-colored graphs} and triangulated topological (pseudo)manifolds of dimension $D$. First observed and taken advantage of in topology \cite{Crystallisation,Ferri1986,Gagliardi1981Regular,Gagliardi2,Ferri}, it was later leveraged in the context of the tensor models approach to random geometry and quantum gravity \cite{Gurau:2010ba, Bonzom:2011zz,Gurau:2011xp,Bonzom:2012hw,Gurau:2013cbh,gurau_random_2017,lionni_colored_2018,Gurau:2024nzv}. It emerged from this literature that $D$-colored graphs are also in one-to-one correspondence with tensor-invariants, which parametrize the Boltzmann weights allowed by the symmetries of particular random tensor ensembles \cite{Bonzom:2012hw,gurau_random_2017}. More broadly, $D$-colored graphs label interaction terms of certain types of tensorial field theories \cite{BenGeloun:2011rc,BenGeloun:2012pu} and group field theories \cite{Oriti_2012,Carrozza:2012uv,Samary:2012bw,Carrozza:2013wda,carrozza_tensorial_2014}, which can be thought of as quantum field theories based on an exotic notion of locality (that is precisely defined in terms of the combinatorial structure of $D$-colored graphs). 

Our work exploits the one-to-one correspondence between trace-invariants and $D$-colored graphs, together with the trove of combinatorial insights obtained in the random tensor literature, to investigate natural questions in the theory of multipartite entanglement. As is clear from our previous contributions \cite{Carrozza:2026qcf, Factorization2026}, and will be confirmed again in the present paper, the graph encoding of trace-invariants plays an essential role in our approach: it allows to effectively represent and reason about any such invariant, and most importantly, it greatly facilitates the \emph{design} of trace-invariants respecting certain desired properties. 

\medskip

In \cite{Carrozza:2026qcf}, an infinite family of $\LO$-monotones was introduced and investigated. Thought of as multipartite generalizations of R\'{e}nyi-$k$ entanglement entropies $\{S_\rm{R}^{(k)}\}_{k \geq 2}$ (with integer $k$), these monotones are labeled by $D$-colored graphs and support in this sense a much richer information landscape than in the bipartite setting (see also \cite{SciPostPhys.20.3.086,Gadde:2022cqi,Gadde:2023zzj} for a similar approach to ours). In the present work, we generalize this construction further with the introduction of a family of $\LO$-monotones labeled by \emph{pairs} of $D$-colored graphs. In the bipartite setting, this extended family contains certain differences of R\'{e}nyi entanglement entropies, $\{\mathsf{D}_{m,n} \eqdef S_\rm{R}^{(m)}- S_\rm{R}^{(n)}\}_{2\leq m <n}$ (with integer $m,n$), which are known in the literature as \emph{entanglement spreads} \cite{Schmid:2020pgv,Harrow:2010oah,Anshu:2020vyv}. These quantities, which define non-negative $\LO$-monotones, provide quantitative measures of the \emph{spread} of the entanglement spectrum associated with a bipartite pure state (hence their name). In particular, the entanglement spectrum of a bipartite pure state $\ket\psi$ is \emph{flat} (meaning that it consists in a unique non-zero singular value) if and only if its entanglement spread vanishes:
\begin{equation}\label{eq:bipartite_prop_1}
    \ket\psi \; \text{has a flat entanglement spectrum} \quad \Longleftrightarrow \quad \mathsf{D}_{m,n}(\ket\psi) =0\,.
\end{equation}
Another interesting property of the bipartite entanglement spread $\mathsf{D}_{m,n}$ is that it has a finite asymptotic expectation value in the Haar-random state of local dimension $N$. Namely, if $\ket{\varphi_N}$ denotes a random pure state chosen uniformly at random in the bipartite state space $(\mathbb{C}^N)^{\otimes 2}$ (or, more rigorously, in the underlying projective Hilbert space), the random variable $\mathsf{D}_{m,n} (\ket{\varphi_N})$ concentrates around a typical value
\begin{equation}\label{eq:bipartite_prop_2}
    \langle \mathsf{D}_{m,n} (\ket{\varphi_N}) \rangle \underset{N\to +\infty}{=} O(1)
\end{equation}
at large $N$. One can interpret this as saying that a random uniform state is indistinguishable from a state with flat entanglement spectrum at leading order in $N$ (and this is of course a consequence of Page's seminal observation \cite{Page:1993df} that a random uniform state is indistinguishable from a maximally entangled state at leading order in $N$). The main objective of this paper is to understand if and how these properties of bipartite entanglement spreads generalize to the multipartite setting. We will find out that: 1) $\LO$-monotonicity is always guaranteed by construction; 2) even though it is not universally valid, there is a sense in which the observation of \eqref{eq:bipartite_prop_1} can be generalized to particular families of mutipartite entanglement spreads; 3) property \eqref{eq:bipartite_prop_2} is even less universally valid in the multipartite context, but we will show how to compute the Haar expectation value of most of the multipartite entanglement spreads of interest to us (such calculations can be quite hard in general, as generalized R\'{e}nyi entropies may fail to concentrate around their mean in the multipartite setting; see \cite{Factorization2026}).  

\paragraph{Summary of our main results, and outline of the paper.} In {\bf{Sec.~\ref{sec:General}}}, we define multipartite entanglement spreads in general, and derive some of their basic properties. Subsection \ref{subsec:background} introduces our main notations regarding trace-invariants, and recalls some useful results from our previous works \cite{Carrozza:2026qcf, Factorization2026}; a reader already familiar with this background may safely skip all or parts of this discussion. The core results of the section are to be found in Subsec.~\ref{subsec:def_spreads}, where a family of $D$-partite entanglement spreads $\{R_{G,H}\}$ is introduced. Indexed by a pair $(G,H)$ of $D$-colored graphs obeying some condition (see Def.~\ref{def:entanglement_spread}), those $\LU$-invariant functions are shown in Prop.~\ref{prop:LO_monotones_more} to be non-increasing under $\LO$ transformations; they thus define $\LO$-monotones for $D$-partite pure states. We also discuss how the bipartite family of entanglement spreads $\{\mathsf{D}_{m,n} \}_{2\leq m <n}$ can be recovered in our framework, and introduce our first examples of spreads with $D\geq 3$. 

In {\bf{Sec.~\ref{sec:Examples}}}, we introduce and discuss more thoroughly a number of infinite families of entanglement spreads that will prove relevant in later parts of the paper. Subsection \ref{sec:spreads_binary_op} is focused on entanglement spreads that can be efficiently described in terms of two binary operations on colored graphs, which were previously defined and studied in \cite{Carrozza:2026qcf}: the \emph{vertex-contraction} and \emph{flip} operations. Those two operations can be used to associate an entanglement spread to any graph $G$ with a marked vertex or edge (see Prop.~\ref{prop:Spread_contractionOp} and Prop.~\ref{prop:ES_flip}). Any such spread vanishes on HT states (see Props.~\ref{prop:vertex_spread_HT} and \ref{prop:flip_spread_HT}), and more broadly, its minimizers can be characterized by simple graphical equations (see Eqs.~\eqref{eq:CS_minimizers} and \eqref{eq:CS_minimizers_flip}). Remarkably, its expectation value in the Haar-random state can also be evaluated explicitly at large $N$ (see Props.~\ref{prop:Haar_HoH}, Cor.~\ref{cor:Haar_flip-1}, and Cor.~\ref{cor:Haar_flip-2}), and is not always vanishing at leading order. Subsec.~\ref{subse:tripartite_spreads} introduces two infinite families of tripartite entanglement spreads, denoted $\{ \SpreadOne_{n}^{(c)} \}_{n \in \mathbb{N}^*, 1\leq c\leq 3}$ and $\{ \SpreadTwo_{m,n,p,q}^{(c)}\}_{m,n,p,q\in \mathbb{N}, 1\leq c\leq 3}$ (see Defs.~\ref{def:Pn} and \ref{def:Qmnpq}), which were designed so that their minimizers would obey particularly relevant graphical equations. The first family leads to graphical equations that can be interpreted as implementing local moves known as \emph{separating $1$-dipole} contractions or insertions (see Prop.~\ref{prop:graph_eq_move1}). Likewise, the minimizers of the second family can be analyzed in terms of graphical moves that insert or remove pairs of \emph{non-separating $1$-dipoles} along a common face (see Prop.~\ref{prop:move2}). Again, HT states are always among the minimizers of those entanglement spreads (see Props.~\ref{prop:SpreadOne_alpha_states} and \ref{prop:SpreadTwo_alpha_states}). Regarding random states, the large-$N$ Haar expectation value of any $\SpreadOne_n^{(c)}$ can be straightforwardly evaluated (see Prop.~\ref{prop:SpreadOne_Haar}), but not the same can be said of a generic spread of the form $\SpreadTwo_{m,n,p,q}^{(c)}$ (however, see Eq.~\eqref{eq:typical_value_Q0000} for an explicit evaluation of the typical value of $Q_{0,0,0,0}^{(c)}$ in the Haar-random state). 

{\bf{Sec.~\ref{sec:HTstates}}} develops alternative characterizations of tripartite HT states as minimizers of entanglement spreads. This question is motivated by the observation that, when $D=2$, an arbitrary bipartite pure state $\ket\psi$ is ($\LU$-equivalent to) an HT state if and only if it has a flat entanglement spectrum. By Eq.~\eqref{eq:bipartite_prop_1}, it follows that $\ket{\psi}$ is an HT state if and only if it minimizes any given entanglement spread $\mathsf{D}_{m,n}$ (for some fixed integers $m,n$ obeying $2\leq m <n$), if and only if it minimizes the whole family of entanglement spreads $\{\mathsf{D}_{m,n}\}_{2\leq m < n}$. Can this observation be generalized to $D=3$ (and, eventually, to arbitrary $D$)? Or, said differently, to which extent can one interpret the set of HT states as a legitimate multipartite analog to the family of flat bipartite states? To address these questions, we start out by reviewing key aspects of the $D=2$ story in a way that makes minimal use of the notion of singular value, and is thus susceptible to be generalized to $D=3$ (see subsection \ref{sec:bipartite_charac}). Such a generalization is obtained in subsection~\ref{sec:tripartite_HT}, where we actually provide two distinct (but related) characterizations of tripartite HT states. To begin with, Thm.~\ref{thm:first_charac} establishes that a tripartite state $\ket\psi$ is HT if and only if: (a) it is a minimizer of the whole family of spreads $\{\SpreadOne_n^{(c)}\}_{n \in \mathbb{N}^* , 1\leq c \leq 3}$, and (b) there exists some $N \in \mathbb{N}^*$ such that, for any \emph{maximally single-trace} $3$-colored graph $G$ with $2k$ vertices, $\tr_G (\ket\psi)= N^{1-k}$. The notion of maximally single-trace invariant was first formalized in \cite{Ferrari:2017jgw,,Valette:2019nzp} and is recalled in Sec.~\ref{subsec:background}. In $D=2$, there is exactly one maximally single-trace graph for any given value of $k$, and the associated invariant computes the $k$-purity of the reduced density matrix associated to either one of the two subsystems. As a result, one can check that, when $D=2$, condition (b) above is in fact equivalent to the vanishing of the family of entanglement spreads $\{D_{m,n}\}_{2\leq m <n}$, and thus, it characterizes on its own the set of bipartite HT states. By contrast, the two conditions (a) and (b) above are independent when $D=3$, and jointly characterize the set of tripartite HT states. We can thus interpret the result of Thm.~\ref{thm:first_charac} as a characterization based on two conditions that are equivalent when $D=2$ but become independent when $D=3$. Secondly, we propose another characterization in Thm.~\ref{thm:second_charac}, where condition (b) above is replaced by the following: (b') $\ket\psi$ is a minimizer of the whole family of entanglement spreads $\{ \SpreadTwo_{m,n,p,q}^{(c)}\}_{m,n,p,q\in \mathbb{N}, 1\leq c\leq 3}$. This results in a complete characterization of tripartite HT states as minimizers of entanglement spreads. In contrast to the bipartite case, where evaluating a single entanglement spread is sufficient to assess whether a bipartite state is flat or not, a countably infinite number of entanglement spreads needed to be invoked in this characterization. (Moreover, it is not true anymore that all the $(G,H)$-spreads defined in Def.~\ref{def:entanglement_spread} are minimized by HT states when $D\geq 3$.)

In {\bf{Sec.~\ref{sec:SHTstates}}}, we investigate again the question of the characterization of HT states as minimizers of entanglement spreads, in a context which is both more general and more specific than the previously described theorems. Indeed, while the results of that section apply to arbitrary $D\geq 3$, they are formulated within a restricted space of (inequivalent $\LU$ orbits of) $D$-partite states. Namely, we introduce the set of \emph{spectral hyper-graph states} (SHT) in Def.~\ref{def:SHT_states}, which are labeled by a \emph{spectral function} $\pmb{p}$. This object associates a \emph{spectrum} $\pmb{p}(B)$ to any subset of parties $B \subset\{1, \ldots , D\}$ with $\abs{B}\geq 2$. When the spectrum of $B$ is flat, the $|B|$ parties in $B$ share a $\GHZ$ state, and otherwise, they share a deformed $\GHZ$ state parametrized by the non-flat spectrum $\pmb{p}(B)$. Henceforth, the SHT associated to $\pmb{p}$ is an HT state if and only if $\pmb{p}(B)$ is flat for any $B$. The SHT Ansatz is interesting because it allows for an evaluation of any trace-invariant in closed form; see Eq.~\eqref{eq:Tr_SHT}. Moreover, it is easy to see that \emph{any} bipartite state is ($\LU$-equivalent) to a SHT state (see Lem.~\ref{lem:D2LU_eq_1}). While this is not true anymore of $D$-partite states with $D\geq 3$ (\eg the $\ket{W}$ state is not an SHT state, see Sec.~\ref{sec:LU_SHT}), the SHT Ansatz provides a significant generalization of the notion of HT state,  that we see as worth investigating in its own right. We prove in Thm.~\ref{th:LUSHT} that the spectral function $\pmb{p}$ labels the $\LU$-orbits of SHT states \ie that two SHT states are $\LU$-equivalent if and only if they have the same spectral function. We then introduce a new family of entanglement spreads $\{\mathsf{D}_G^{\mathrm{(MST)}}\}$, that is generated by the vertex-contraction operation on \emph{maximally single-trace} graphs $G$ with more than $2$ vertices (see Prop.~\ref{prop:def_Delta_G}). We finally prove in Prop.~\ref{prop:ES_SHT} that an SHT state $\ket{\psi}$ is an HT state if and only if it minimizes \emph{one} of the entanglement spreads from $\{\mathsf{D}_G^{\mathrm{(MST)}}\}$ above (if and only if it minimizes \emph{all} of them). Furthermore, in that case, we can find an integer $N \in \mathbb{N}^*$ such that $\tr_G \left( \ket\psi \right)= N^{1-k}$ for any maximally-single trace $G$ with $2k$ vertices. Prop.~\ref{prop:ES_SHT} is thus completely analogous to the characterization of flat bipartite states among all bipartite pure states we previously mentioned (where a maximally single-trace invariant with $2k$ vertices simply computes the $k$-purity of the reduced density matrix associated with any one of the two subsystems).

\section{Generalities about multipartite entanglement spreads} \label{sec:General}

\subsection{Background and notations}\label{subsec:background}

In the following, $D$ denotes the number of parties or subsystems in an entanglement scenario, and is therefore an integer larger or equal to $2$. Given two finite-dimensional $D$-partite pure states $\ket\psi$ and $\ket\phi$, we will write $\ket\psi \simLU \ket\phi$ whenever $\ket\psi$ and $\ket\phi$ are local-unitary equivalent, and $\ket\psi \toLO \ket\phi$ whenever $\ket\phi$ can be obtained from $\ket\psi$ by application of a $\LO$ transformation (\ie there is an arrow from $\ket\psi$ to $\ket\phi$ in the $\LO$ preorder on pure states).

\paragraph{Trace-invariants.} Given a $D$-partite pure state $\ket{\psi} \in \H= \H_1 \otimes \cdots \ot \H_D$ and local bases $\paa{\ket{i_c}}_{1 \leq i_c \leq N_c}$ with $N_c = \dim(\H_c) \in \bb{N}^*$, we let the tensor components $\psi_{i_1 \dots i_D}$ be the coefficient appearing in the decomposition of $\ket{\psi}$ in the product basis, namely:
\begin{equation}
    \ket{\psi} = \sum_{i_1=1}^{N_1} \cdots \sum_{i_D=1}^{N_D} \psi_{i_1 \dots i_D} \ket{i_1} \ot \cdots \ot \ket{i_D}\,.
\end{equation}
\textit{Trace-invariants} are peculiar $\LU$-invariant homogeneous polynomial functions of the tensor components and their complex conjugates. These invariants are obtained by index summation of $k\in \mathbb{N}^*$ copies of the tensor components $\psi_{i_1 \dots i_D}$ with $k$ copies of their complex conjugate $\bar \psi_{i_1 \dots i_D}$. In order for the result to be $\LU$-invariant, an index in position $c \in \paa{1,\dots,D}$ of $\psi = \paa{\psi_{i_1 \dots i_D}}$ is always summed with an index of $\bar \psi = \paa{\bar \psi_{i_1 \dots i_D}}$ in the same position. We refer to $k$ as the \emph{degree} of any invariant obtained in this way. The inner product $\displaystyle\langle \psi \vert \psi \rangle= \sum_{i_1=1}^{N_1} \cdots \sum_{i_D=1}^{N_D} \abs{\psi_{i_1\ldots i_D}}^2$ is the only trace-invariant of degree $1$, and we will assume throughout the paper that it always takes value $1$: this simply amounts to representing any pure state by a \emph{unit} vector $\ket{\psi} \in S(\H)$, where $S(\H)$ denotes the unit sphere in $\H$. This will sometimes be assumed implicitly in the rest of the text.

Throughout the paper, we will draw on \textit{bipartite regular $D$-edge-colored graphs} to study trace-invariants since these two sets are in one-to-one correspondence, as we will see later. A bipartite $D$-edge-colored graph $G$ (or \textit{$D$-colored graph} for simplicity) is described by two sets: a set of vertices and a set of edges. The set of vertices is partitioned into two disjoint subsets -- one for the white vertices and the other for the black vertices -- of the same size $k(G) \in \bb{N}^*$. As for the edge set, it is partitioned into $D$ disjoint subsets, one for each \textit{color} $c \in \paa{1,\dots,D}$. 
For any color $c$ and any vertex $v$, there is exactly one edge of color $c$ incident to $v$. In particular, the graph $G$ is $D$-regular, and its color-$c$ edge set has size $k(G)$. In addition, any edge is required to connect a black to a white vertex. Said otherwise, the graph $G$ is bipartite.

The correspondence between trace-invariants and colored graphs is given as follows: associate a copy of $\psi$ to every white vertex of $G$ as well as a copy of $\bar{\psi}$ to every black vertex of $G$. Index contractions are then encoded by colored edges, namely: for every edge of color $c \in \paa{1,\dots,D}$, if that edge is connecting a white vertex associated with $\psi_{i_1 \dots i_D}$ to a black vertex associated with $\bar \psi_{j_1 \dots j_D}$, the indices $i_c$ and $j_c$ are then forced to be equal and summed from $1$ to $N_c$. We will denote by $\Tr_G(\ket{\psi})$ the evaluation of a trace-invariant described by the colored graph $G$ on the pure quantum state $\ket{\psi} \in \H$. As a trivial example, note that the inner product $\langle \psi \vert \psi \rangle$, \ie the unique trace-invariant of degree $1$, is represented by the unique $D$-colored graph $G$ with $k(G)=1$ (which consists in two vertices, one white and one black, connected by $D$ edges of distinct colors); given our normalization convention, we always have $\tr_G(\ket\psi) = 1$ for this graph. More precisely, we have
\begin{equation}
    G = \vcenter{\hbox{\includegraphics[scale=1]{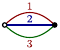}}} \quad \Longrightarrow \quad \Tr_{G}(\ket{\psi}) = \sum_{{\color{Red4}i_1} = 1}^{N_1} \sum_{{\color{Blue4}i_2} = 1}^{N_2} \sum_{{\color{Green4}i_3} = 1}^{N_3} \psi_{{\color{Red4}i_1} {\color{Blue4}i_2} {\color{Green4}i_3}} \bar{\psi}_{{\color{Red4}i_1} {\color{Blue4}i_2} {\color{Green4}i_3}} = \langle \psi \vert \psi \rangle = 1\,.
\end{equation}
The next equation provides a non-trivial example of a trace-invariant of degree three:
\begin{equation}
    G = \vcenter{\hbox{\includegraphics[scale=.6]{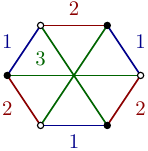}}} \quad \Longrightarrow \quad \Tr_{G}(\ket{\psi}) = \sum_{{\color{Red4}a} = 1}^{N_1} \cdots \sum_{{\color{Green4}i} = 1}^{N_3} \psi_{{\color{Blue4}a} {\color{Red4}d}  {\color{Green4}g}} \bar{\psi}_{{\color{Blue4}a} {\color{Red4}e} {\color{Green4}i}} \psi_{{\color{Blue4}b} {\color{Red4}e}  {\color{Green4}h}} \bar{\psi}_{{\color{Blue4}b} {\color{Red4}f} {\color{Green4}g}}  \psi_{{\color{Blue4}c} {\color{Red4}f}  {\color{Green4}i}} \bar{\psi}_{{\color{Blue4}c} {\color{Red4}d} {\color{Green4}h}} \,.
\end{equation}

Before we move on, we recall two properties of trace-invariants that will be particularly useful later. First, for any pure state $\ket{\psi} \in S(\H)$ and any $D$-colored graph $G$, we have the bound $\abs{\Tr_G(\ket{\psi})} \leq 1$ (see Refs.~\cite{Gadde:2025csh,Carrozza:2026qcf}). Secondly, if $\ket{\varphi} \in S(\H)$ is a second pure state, we have $\Tr_G(\ket{\psi} \ot \ket{\varphi}) = \Tr_G(\ket{\psi}) \Tr_G(\ket{\varphi})$.

\paragraph{Notations regarding colored graphs.} Across the paper, we let $\cG_D$ (resp.~$\cG_D^{\conn}$) be the set of $D$-colored graphs (resp.~connected $D$-colored graphs). What is more, if $G \in \cG_D$, we define $\cG_{D+1}(G)$ to be the set of $(D+1)$-colored graphs obtained by adding edges of \textit{color $0$} to $G$ in every possible way. Usually, we will denote an element of $\cG_{D+1}(G)$ by the use of a hat, \ie $\widehat G \in \cG_{D+1}(G)$.

For $G \in \cG_D$, we let $\Comp(G)$ denote the set of connected components of $G$, and $\kappa(G) \eqdef \abs{\Comp(G)}$ its number of connected components. Furthermore, if $B \subset\paa{1,\dots,D}$ with $\abs{B} \geq 2$, we let $G \vert_B$ be the $\abs{B}$-colored graph obtained by restricting $G$ to its subset of edges whose colors appear in the subset $B$. If $G \in \cG_D$, we let $\bar G$ be the graph obtained by flipping the color of any vertex in $G$ (\ie turning any white vertex into a black one, and conversely).

For a given $G \in \cG_D$ and two distinct integers $c,c' \in \paa{1,\dots,D}$, we define \textit{the number of faces of color} $(c,c')$ as the number of connected components of $G\vert_{\paa{c,c'}}$, \ie: $F_{cc'}(G) \eqdef \kappa(G\vert_{\paa{c,c'}})$. With this definition, we can introduce a number of combinatorial quantities. First, for a graph $G \in \cG_D$ and a graph $\widehat G \in \cG_{D+1}(G)$, we define the number of faces of $G$ and the number of faces of color $0$ of $\widehat G$ respectively as:
\begin{equation}
    F(G) \eqdef \sum_{1 \leq c < c' \leq D} F_{cc'}(G) \qquad \rm{ and } \qquad F_0(\widehat G) \eqdef \sum_{c=1}^D F_{0c}(\widehat G) \,.
\end{equation}
With this data, we can further define the so-called \textit{Gurau degree} by the following relation
\begin{equation}
    \omega(G) \eqdef (D-1) \kappa(G) + \frac{(D-1)(D-2)}{2} k(G) - F(G) \,,
\end{equation}
as well as the quantity $\Delta_0(\widehat G)$, defined as
\begin{equation}
    \Delta_0(\widehat G) \eqdef \frac{D(D-1)}{4} k(G) + \frac 1 2 F(G) - \frac{D-1}{2} F_0(\widehat G)
    \,. 
\end{equation}
This last quantity is closely related to the \emph{degree of compatibility} $\Delta(G)$ of $G$:
\begin{equation}
    \Delta(G) \eqdef \min_{\widehat G \in \cG_{D+1}(G)} \Delta_0(\widehat G) 
    \,. 
\end{equation}
Let $\pa{\ket{\varphi_N}}_{N \in \bb{N}^*}$ be a sequence of $D$-partite Haar random states of local dimension $N$. Thanks to Wick's theorem (see Refs.~\cite{Nechita2007,Collins2025} or App.~A of Ref.~\cite{Carrozza:2026qcf} for more details), the asymptotic behaviour of the expectation value $\mean{\Tr_G(\ket{\varphi_N})}$ at large $N$ can be related to the Gurau degree and the degree of compatibility. More precisely, one has: 
\begin{equation}
    \mean{\Tr_G(\ket{\varphi_N})} \underset{N \to \infty}{=} \mu_G N^{\kappa(G) - k(G) - \frac{\omega(G)}{D-1} - \frac{2 \Delta(G)}{D-1}}\pa{1 + O\pa{\frac 1 N}}\,,
\end{equation}
with $\mu_G \in \bb{N}^*$ a degeneracy factor that counts the number of Wick contractions contributing to the leading order.

Let $G\in\cG_D$ and $v$ a vertex of $G$ (resp.~$e$ an edge of $G$). We denote by $G^{(v)}$ (resp.~$G^{(e)}$) the graph obtained from $G$ by removing the vertex $v$ (resp.~cutting the edge $e$ open). The result is an \emph{open} $D$-colored graph, in which half-edges are allowed: $G^{(v)}$ is made up of $D$ half-edges, $k(G)-1$ vertices, and $D$ fewer edges than $G$; $G^{(e)}$ is made up of $2$ half-edges, $k(G)$ vertices, and $1$ fewer edge than $G$. We will use the notation $\Tr_{G^{(v)}}(\ket{\psi})$ (resp.~$\Tr_{G^{(e)}}(\ket{\psi})$) to refer to the tensor (resp.~matrix) obtained by the same encoding of contraction rules we previously described, namely: the complete edges of $G^{(v)}$ (resp. $G^{(e)}$) encode contraction patterns (or traces), while the number of half-edges of each type determine the nature of the tensor obtained once all such contractions have been implemented. As a result, $\tr_{G^{(v)}}(\ket{\psi})$ defines an element of $\H_1 \otimes \cdots \ot \H_D$ or its dual (depending on the color of the vertex $v$), while $\tr_{G^{(e)}}(\ket{\psi})$ can be understood as an element of $\H_c \otimes \H_c^*$ (in other words, it can be identified to a linear map on $\H_c$, or to its matrix representation).

To end the introduction of notations regarding colored graphs, we introduce some well-known invariants from either the quantum information or random tensor literature. 
\begin{itemize}
    \item Let $k \in \bb{N}^*$. To begin with, we let $C_k$ denotes the unique connected $2$-colored graph with $k$ vertices. The set $\paa{C_k}_{k\in \bb{N}^*}$ is known as the set of \textit{cyclic graphs}. When $k>1$, the graph $C_k$ is associated with the so-called \textit{Rényi-$k$ entanglement entropy}, defined as:
    \begin{equation}\label{eq:Renyi-k-C-k}
        S_\rm{R}^{(k)}(\ket{\psi}) \eqdef \frac{1}{1-k} \ln \Tr_{C_k}(\ket{\psi}) \,.
    \end{equation}
    In the multipartite context, cyclic graphs are associated with a bipartition of the Hilbert space $\H$ of the form $\H_B \ot \H_{\bar{B}}$ where $B \subset \paa{1,\dots,D}$ and $\H_B = \bigotimes_{c\in B} \H_c$. For any $k\in \mathbb{N}^*$, there is then a unique connected $D$-colored graph $C_k(B)$ with $k$ vertices, and such that $\kappa(C_k(B)|_B)= \kappa(C_k(B)|_{\bar{B}})=k$ (visually, this means that, edges with colors from the same block $B$ or $\bar{B}$ can be separated into $k$ distinct groups, in such a way that any pair of edges from the same group are ``parallel'' \ie they connect the same pair of vertices). Furthermore, the trace-invariant associated with $C_k(B)$ determines the R\'{e}nyi-$k$ entanglement entropy across the bipartition $\H_B \otimes \H_{\bar B}$, as in Eq.~\eqref{eq:Renyi-k-C-k}. Some examples of cyclic graphs are represented in Fig.~\ref{fig:Cyclic}.
    
    \begin{figure}[ht]
        \centering
        \includegraphics[width = .7\textwidth]{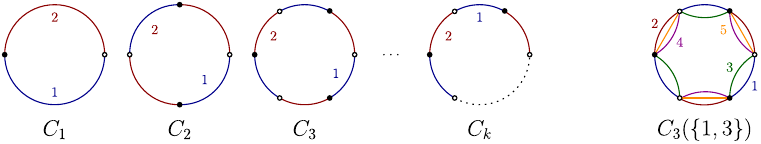}
        \caption{Left: examples of cyclic graphs for $D=2$. Right: example of a cyclic graph for $D=5$, and associated with the bipartition $\H_{\paa{1,3}} \ot \H_{\paa{2,4,5}}$.}
        \label{fig:Cyclic}
    \end{figure}
    
    \item In what follows, the $D$-colored graphs satisfying the equation $\omega(G) = 0$ are known as the \textit{melonic graphs} (see \eg Refs.~\cite{Bonzom:2012hw,gurau_random_2017}). They are constructed recursively by repeated insertions of $2$-vertex subgraphs along edges, as illustrated in Fig.~\ref{fig:melonic_insertion}.
    
    \begin{figure}[ht]
        \centering
        \includegraphics[width = .7\textwidth]{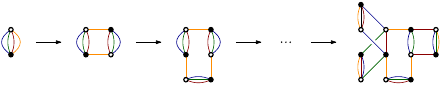}
        \caption{Illustration of the recursive construction of melonic graphs. The $2$-vertex graph is presented on the left.}
        \label{fig:melonic_insertion}
    \end{figure}
    
    \item In the tripartite setting, \ie $\H = \H_1 \ot \H_2 \ot \H_3$, there are two permutations of the indices of the reduced density matrix $\rho_{12} \eqdef \Tr_{\H_3}(\ket \psi \bra \psi)$ that are particularly studied in the field of quantum information (see \eg Refs.~\cite{HORODECKI19961,PhysRevLett.77.1413,Horodecki:2002jnu,Aubrun:2012ydp} and references therein). For $c = 1,2$ and $1 \leq i_c \leq N_c$, those permutations are known as the \textit{partial transpose} and the \textit{realignment}. They are defined respectively as:
    \begin{equation}\label{eq:PartialTransp_Realignm}
        \pa{\rho_{12}^{T_2}}_{i_1 i_2 ; j_1 j_2} = \pa{\rho_{12}}_{i_1 j_2 ; j_1 i_2} \qquad \rm{and} \qquad \pa{\rho_{12}^{R}}_{i_1 i_2 ; j_1 j_2} = \pa{\rho_{12}}_{i_1 j_1 ; i_2 j_2} \,,
    \end{equation}
    and represented in Fig.~\ref{fig:PT_RM_rho}.
    
    \begin{figure}[ht]
        \centering
        \includegraphics[width = .85\textwidth]{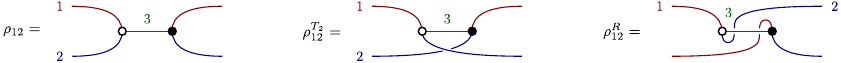}
        \caption{Illustration of the partial transpose and the realignment.}
        \label{fig:PT_RM_rho}
    \end{figure}
    
    We will be interested in the \textit{moments of the partial transpose} and the \textit{realignment moments} respectively associated with the $3$-colored graphs $\PT_k^{(3)}$ and $\RM_{2n}^{(3)}$ with $k,n \geq 1$. The moments are defined as: 
    \begin{equation}
        \Tr_{\PT_k^{(3)}}(\ket{\psi}) \eqdef \Tr\pac{\pa{\rho_{12}^{T_2}}^k} \qquad \rm{and} \qquad \Tr_{\RM_{2n}^{(3)}}(\ket{\psi}) \eqdef \Tr\pac{\pa{\rho_{12}^{R} (\rho_{12}^R)^\dagger}^n} \,.
    \end{equation}
    Examples of these moments are illustrated in Fig.~\ref{fig:PT_RM}. When the color $3$ is replaced by any color $c$ in $\paa{1,2,3}$ one has $\PT_k^{(c)}$ and $\RM_{2n}^{(c)}$. We observe that $\PT_3$ and $\RM_{4}$ are independent of the choice of color. \\
    In App.~\ref{A:proofTechLem}, we will make use of the moments of the partial transpose of $\rho_{ij} = \Tr_{\H_{\bar{\paa{i,j}}}}(\ket{\psi} \bra{\psi})$. These moments will then be associated with the colored graphs $\PT_k^{(\bar{\paa{i,j}})}$. 
    
    \begin{figure}[ht]
        \centering
        \includegraphics[width = .8\textwidth]{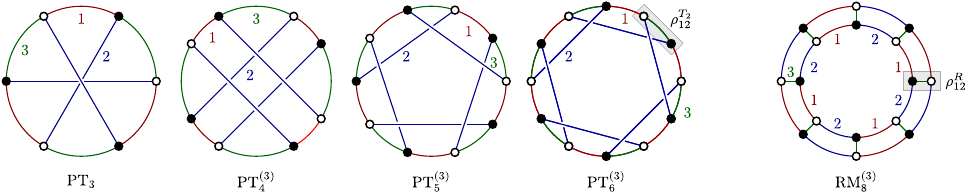}
        \caption{Left: four examples of moments of the partial transpose of $\rho_{12}$. Right: an example of realignment moment of $\rho_{12}$.}
        \label{fig:PT_RM}
    \end{figure}
    
    \item The set of graphs $G$ satisfying $\kappa(G\vert_{\paa{c,c'}}) = 1$ for every pair of distinct integers $c,c' \in \paa{1,\dots,D}$ (which implies $\kappa(G\vert_{B}) = 1$ for any $B \subset\{1, \ldots , B\}$ with $\abs{B}\geq 2$) is known as the set of \textit{maximally single-trace graphs} and will be denoted by  $\cG_D^{(\rm{MST})}$. Three examples of maximally single-trace graphs are provided in Fig.~\ref{fig:MST}.
    
    \begin{figure}[ht]
        \centering
        \includegraphics[width = .5\textwidth]{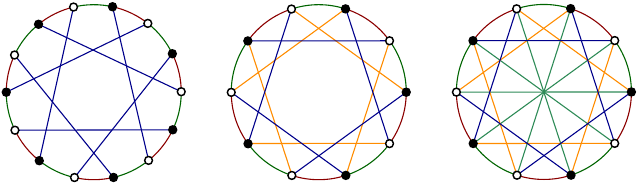}
        \caption{From left to right: examples of graphs from $\cG_3^{(\rm{MST})}$, $\cG_4^{(\rm{MST})}$ and $\cG_5^{(\rm{MST})}$ respectively.}
        \label{fig:MST}
    \end{figure}
\end{itemize}

\paragraph{Hypergraph-tensor states.} By taking tensor products of so-called \textit{Greenberger–Horne–Zeilinger} (GHZ) states \cite{Greenberger:1989tfe}, one can generate an infinite family of multipartite entangled states known as \textit{hypergraph-tensor states} (see \eg Refs.~\cite{Buhrman:2016tif,Vrana:2016edr,Christandl:2018cfb,Christandl:2019zrq,brand2026bilinearcomplexityworksbreaks} and references therein). For notational convenience, we let $\H_{\mathrm{loc}}$ denote a separable Hilbert space of infinite dimension, and $(\ket{i})_{i \in \mathbb{N}^*}$ a fixed orthonormal basis of $\H_{\mathrm{loc}}$. We will define $D$-partite \textit{hypergraph-tensor} states (and, more generally, \textit{spectral hypergraph-tensor} states in Sec.~\ref{sec:SHTstates}), as elements of $D$-partite Hilbert spaces of the form $\H_1\otimes \cdots \otimes \H_D$, where for any $c \in \mathbb{N}^*$, $\H_c = (\H_{\mathrm{loc}})^{\otimes n_c}$ with $n_c \in \mathbb{N}^*$. The particular values of $n_1 , \ldots , n_D$ will be left implicit for each such state $\ket\psi$, as they will be straightforwardly recoverable from our notations. With these notations at hand, the family of \emph{hypergraph-tensor} (HT) states is constructed as follows. Let $2 \leq b \leq D$, $B = \paa{c_1,\dots,c_b} \subset \paa{1,\dots,D}$, and $\alpha(B) \in \bb{N}^*$. We define $\ket{\GHZ}_{B,\alpha(B)}$ to be a GHZ state of dimension $\alpha(B)$ with support on the subsystems labeled by $B$. In other words, we have:
\begin{equation}
    \ket{\GHZ}_{B,\alpha(B)} \eqdef \frac{1}{\sqrt{\alpha(B)}} \sum_{i = 1}^{\alpha(B)} \ket{i}_{c_1} \ot \cdots \ot \ket{i}_{c_b}\,,
\end{equation}
where, for any $c \in \paa{1,\dots,D}$, the subscript $c$ on $\ket{i}_c$ indicates that this basis vector belongs to one of the copies of $\H_{\mathrm{loc}}$ constituting $\H_c$. Then, given a \textit{weight function} $\alpha  : \{B \subset\{1,\dots,D\}\,\vert \, \abs{B}\geq 2\}\to\mathbb{N}^*$, the $D$-partite HT state $\ket{\psi_\alpha}$ is defined as:\footnote{This definition makes it clear that any such $\ket{\psi_\alpha}$ is in fact ($\LU$-equivalent to) a \emph{finite-dimensional} $D$-partite state: indeed, one can always find some $N \in \mathbb{N}^*$ such that $\ket\psi \in (\H_{\mathrm{loc}, N})^{\otimes n_1} \otimes \cdots \otimes (\H_{\mathrm{loc}, N})^{\otimes n_D}$, where $\H_{\mathrm{loc}, N}\eqdef \Span \{\ket{i}\}_{1 \leq i \leq N}$ is finite-dimensional. Thus, invoking the infinite-dimensional local Hilbert space $\H_{\mathrm{loc}}$ and its preferred basis is only a matter of notational convenience, as anticipated.}
\begin{equation}
    \ket{\psi_\alpha} \eqdef \bigotimes_{B \in \paa{1,\dots,D}, \, \abs{B} \geq 2} \ket{\GHZ}_{B,\alpha(B)} \,.
\end{equation}
More broadly, any $D$-partite state that is $\LU$-equivalent to $\ket{\psi_\alpha}$ will be referred to as an HT state of weight function $\alpha$. In the remainder of this paper, we let $W_D$ be the set of weight functions of $D$-partite hypergraph-tensor states, namely
\begin{equation}
    W_D \eqdef \paa{ \alpha  : \{B \subset\{1,\dots,D\}\,\vert \, \abs{B}\geq 2\}\to\mathbb{N}^* }
\end{equation}

\paragraph{Technical lemmas.} For the reader's convenience, we also gather in this section some of the technical lemmas we will make use of in the main body of the text, which are either novel or adapted from our earlier work (see Ref.~\cite{Carrozza:2026qcf}).

First, in order to prove the $\LO$-monotonicity of the multipartite entanglement spreads we will introduce, we will invoke the following proposition and its corollary.

\begin{prop}[\cite{Carrozza:2026qcf}, Prop.~2.28]\label{prop:charac_LO}
    Let $\ket{\psi}$ and $\ket{\phi}$ be two finite-dimensional $D$-partite pure states. We have $\ket{\psi} \toLO \ket{\phi}$ if and only if there exists a finite-dimensional $D$-partite pure state $\ket{\eta}$ such that:\footnote{In Eq.~\eqref{eq:equivalence_LO}, $\cG_D^{\conn}$ can be equivalently replaced by $\cG_D$.}
    \begin{equation}\label{eq:equivalence_LO}
        \forall G \in \cG_D^{\conn}\,, \qquad \tr_G(\ket{\psi}) = \tr_G(\ket{\phi}) \tr_G(\ket{\eta})\,. 
    \end{equation}
\end{prop}

\begin{cor}[\cite{Carrozza:2026qcf}, Cor.~2.33]
    \label{cor:LO_monotones}
     Let $\ket{\psi}$ and $\ket{\phi}$ be two finite-dimensional $D$-partite pure states. If $\ket{\psi} \toLO \ket{\phi}$, then
    \begin{equation}
        \forall G \in \cG_D\,, \qquad |\tr_G(\ket{\psi})| \leq |\tr_G(\ket{\phi})|\,.
    \end{equation}
\end{cor}

Then, let us recall the definition of \textit{the large-$N$ factorization} and the \textit{power law} criteria already introduced in Sec.~6.3.1 of Ref.~\cite{Carrozza:2026qcf}, which are relevant to study entanglement spreads of random states in an asymptotic large-$N$ regime. We focus here on Haar-random states of local dimension $N$, but the same criteria can be applied to other distributions (see Sec.~6.3.1 of \cite{Carrozza:2026qcf} for more broadly applicable definitions).

\begin{defi}[\cite{Carrozza:2026qcf}, Defs.~6.6 and 6.15]\label{def:fact_power}
    Let $D \geq 2$, $G \in \cG_D$, and let $\varphi=\pa{\ket{\varphi_N}}_{N \in \bb{N}^*}$ denote a sequence of Haar-random states of local dimension $N$.\footnote{Concretely, for every $N \in \mathbb{N}^*$, $\ket{\varphi_N}$ can be defined as the unique random $D$-partite pure state on $\left(\mathbb{C}^N\right)^{\otimes D}$ whose distribution is invariant under the group of unitary transformations of $\left(\mathbb{C}^N\right)^{\otimes D}$.} 
    \begin{enumerate}
        \item We say that $\Tr_G$ (or $G$) satisfies  \emph{the large-$N$ factorization criterion} if:\footnote{In Ref.~\cite{Carrozza:2026qcf}, we introduced a more general definition of this criterion, that can be applied to arbitrary sequences of random pure states. Since we will only apply this criterion to $\varphi$ in the present paper, we have simplified the definition a little bit.}
        \begin{equation}
            \mean{\Tr_{G}(\ket{\varphi_N}) \Tr_{\bar G}(\ket{\varphi_N})} \underset{N \to \infty}{\sim} \mean{\Tr_{G}(\ket{\varphi_N})} \mean{ \Tr_{\bar G}(\ket{\varphi_N})}
        \end{equation}
        \item We say that $\Tr_G$ (or $G$) \emph{follows a power law} if there exists $K>0$ and $r \geq 0$ such that: for any $N \in \mathbb{N}^*$, and any $D$-partite pure state $\ket\psi \in S\left(\left( \mathbb{C}^N\right)^{\otimes D}\right)$ (here, the normalization condition $\langle\psi \vert  \psi \rangle=1$ is crucial), 
        \begin{equation}
            \abs{\Tr_G(\ket{\psi})} \geq K N^{-r}\,.
        \end{equation}
    \end{enumerate}
\end{defi}
\begin{rem}
    The second definition of Def.~\ref{def:fact_power} is equivalent to the following. $\tr_G$ obeys a power law if there exists $K > 0$ and $r\geq 0$ such that: for any finite-dimensional $D$-partite state space $\H = \H_1 \otimes \cdots \otimes \H_D$, and any normalized pure state $\ket\psi \in S(\H)$,
    \begin{equation}
        \abs{\tr_G\left( \ket\psi \right)} \geq K \left(\min_{1 \leq c \leq D} \dim\left( \H_c \right)\right)^{- r}\,.
    \end{equation}
This formulation is manifestly stable under isometric embedding of any local Hilbert space $\H_c$ into a higher-dimensional Hilbert space $\tilde\H_c$, and as a result, one may check that it is indeed equivalent to the definition recalled from \cite{Carrozza:2026qcf} in Def.~\ref{def:fact_power}.
\end{rem}

Furthermore, it was proven in Ref.~\cite{Carrozza:2026qcf} that, taken together, the large-$N$ factorization and the power law criteria are sufficient conditions enabling the exchange of the logarithm and the expectation value in the asymptotic large-$N$ evaluation of $\mean{\ln \abs{\Tr_G(\ket{\varphi_N})}}$. Concretely, the next proposition was shown in our earlier work.

\begin{prop}[\cite{Carrozza:2026qcf}, Prop.~6.18] \label{prop:mean_log}
    Let $D \geq 2$, $G \in \cG_D$, and let $\varphi= \pa{\ket{\phi_N}}_{N \in \bb{N}^*}$ denote a sequence of Haar-random states of local dimension $N$. If 1) $G$ satisfies the large-$N$ factorization criterion and 2) $G$ follows a power law, then:\footnote{Note that $\mean{\Tr_G(\ket{\varphi_N})}$ is always strictly positive in the Haar-random state, as a consequence of Wick's theorem (see Ref.~\cite{Nechita2007} or App.~A of Ref.~\cite{Carrozza:2026qcf} for more details).}
    \begin{equation}
        \mean{\ln \abs{\Tr_G(\ket{\varphi_N})}} \underset{N \to \infty}{=} \ln \mean{\Tr_G(\ket{\varphi_N})} + o(1) \,.
    \end{equation}
\end{prop}

The following lemma is new; it provides a simple but useful way of establishing the power law property of large families of graphs.

\begin{lem}\label{lem:power-law}
   Let $D \geq 2$, $G\in \cG_D$, $n \in \mathbb{N}^*$ and $H_1, \ldots , H_n \in \cG_D$. Suppose that $G$ obeys a power law, and that there exist $r_1, \ldots , r_n \in \mathbb{R}_+^*$ such that: for any finite-dimensional $D$-partite pure state $\ket\psi$,
    \begin{equation}
        \abs{\Tr_G(\ket{\psi})} \leq \prod_{i=1}^N \abs{\Tr_{H_i}(\ket{\psi})}^{r_i}\,.
    \end{equation}
Then: for any $i \in \{1,\ldots , D\}$, $\tr_{H_i}$ obeys a power law.
\end{lem}
\begin{proof}
    Since $G$ follows a power law, there exists $(K,r) \in \bb{R_+^*} \times \bb{R}_+$ such that $\abs{\Tr_G(\ket{\psi})} \geq K N ^{-r}$ for any (normalized) $D$-partite pure state $\ket{\psi}\in \left(\mathbb{C}^{N}\right)^{\otimes D}$. But we also know that the state $\ket\psi$ obeys $\abs{\tr_{H}\left( \ket\psi \right)}\leq 1$ for any $H \in \cG_D$. As a result, we find that, for any $i \in \{1, \ldots , n\}$,     \begin{equation}
        K N ^{-r} \leq \abs{\Tr_G(\ket{\psi})} \leq \abs{\Tr_{H_i}(\ket{\psi})}^{r_i} \,,
    \end{equation}
    hence $\abs{\Tr_{H_i}(\ket{\psi})} \geq K^{1/r_i} N^{-r/r_i}$, implying that $\tr_{H_i}$ obeys a power law.
\end{proof}

Finally, inspired by the proof of Thm.~5.3. of Ref.~\cite{Carrozza:2026qcf}, we prove the following lemma in App.~\ref{A:proofTechLem}.
\begin{lem}\label{lem:alpha_integer_valued}
    Let $\ket\psi$ a finite-dimensional $D$-partite pure state. Suppose there exists a function $\alpha : \{B \subset\{1,\dots,D\}\,\vert \, \abs{B}\geq 2\}\to \mathbb{R}_+^*$ such that:
    \begin{equation}
        \forall G \in \cG_D^{\conn}\,, \quad \tr_G(\ket{\psi})= \prod_{B \subset\{1,\dots,D\}\,, \,\abs{B}\geq 2} \alpha(B)^{\kappa(G\vert_B) - k(G)}\,. 
    \end{equation}
    Then, the function $\alpha$ necessarily takes value in $\mathbb{N}^*$, and identifying it to a weight function $\alpha  : \{B \subset\{1,\dots,D\}\,\vert \, \abs{B}\geq 2\}\to\mathbb{N}^*$, we have:
    \begin{equation}
        \ket\psi \simLU \ket{\psi_\alpha}\,.
    \end{equation}
\end{lem}

\begin{proof}
    See App.~\ref{A:proofTechLem}.
\end{proof}

\subsection{Definition and \texorpdfstring{$\LO$}{LO}-monotonicity of entanglement spreads}\label{subsec:def_spreads}

With these prerequisites at hand, we can now introduce the notion of entanglement spread that is the main focus of the present work. 
\begin{defi}\label{def:entanglement_spread}
    Let $G, H \in \cG_D$ be two $D$-colored graphs. We will say that $H$ \emph{dominates} $G$ -- noted $G \dom H$ -- if: for any finite-dimensional $D$-partite pure state $\ket\psi$,
\begin{equation}\label{def:dom}
    |\tr_{G}(\ket{\psi})|\leq |\tr_{H}(\ket{\psi})|\,.
\end{equation}     
If $G \dom H$, for any finite-dimensional $D$-partite pure state $\ket\psi$, we define $R_{G , H}(\ket{\psi}) \in [0, +\infty]$ as
\begin{equation}
    R_{G , H}(\ket{\psi})\eqdef \begin{cases} \ln\abs{\tr_{H}(\ket{\psi})} - \ln\abs{\tr_{G}(\ket{\psi})} &\; \mathrm{if}\; \tr_{G}(\ket{\psi}) \neq 0 \; (\mathrm{and}\;\mathrm{thus}\; \tr_{H}(\ket{\psi}) \neq 0) \\ 
    +\infty &\; \mathrm{otherwise}\end{cases}
\end{equation}
We will refer to the resulting function $R_{G,H}$ as $(G,H)$\emph{-spread}. More generally, any function that is a constant multiple of some $(G,H)$-spread, for some positive and non-zero constant, will be referred to as an \emph{entanglement spread}.
\end{defi}

\begin{ex}\label{ex:R_G}
A first class of examples of entanglement spreads is obtained by fixing $H$ to be the trivial melon graph (\ie the unique $D$-colored graph with two vertices), while keeping $G$ arbitrary. One then has $\tr_H(\ket\psi) = 1$ for any finite-dimensional $D$-partite pure state $\ket\psi$, implying that the condition \eqref{def:dom} is verified (since $\abs{\tr_G (\cdot)} \leq 1$ holds for any graph $G$). The resulting $(G,H)$-spread was denoted by $R_G$ in our previous work \cite{Carrozza:2026qcf}, and interpreted as a multipartite generalization of the notion of entanglement R\'{e}nyi entropy. 
\end{ex}

\begin{rem}
    The definition we adopted in Def.~\ref{def:entanglement_spread} does not necessarily lead to a continuous $R_{G,H}$ in case $\tr_G$ has zeroes. Indeed, we cannot \apriori exclude the existence of \eg a one-parameter continuous family of states $\{\ket{\psi_t}\}_{t \in ]-1,1[}$ such that: a) $\tr_G(\ket{\psi_t}) \neq 0$ for any $t \neq 0$, b) $\tr_G(\ket{\psi_0}) = 0$ and c) $\underset{t \to 0}{\lim}\, \frac{\abs{\tr_H (\ket{\psi_t})}}{\abs{\tr_G (\ket{\psi_t})}} = \ell \in [1, + \infty[$. In such a situation, we would have $\underset{t \to 0}{\lim} \, R_{G,H}(\ket{\psi_t}) = \ln(\ell) \neq +\infty = R_{G,H}(\ket{\psi_0})$, hence $R_{G,H}$ would not be continuous. There exists other potentially pathological scenarios of this type. Whether or not it is possible to prevent all such failures of continuity by suitably refining Def.~\ref{def:entanglement_spread} is a corny question that hinges on the respective geometries of the projective subvarieties defined by the two  polynomial equations $\tr_G = 0$ and $\tr_H = 0$, as well as their relations. We leave this question open for future work.
\end{rem}

By virtue of Cor.~\ref{cor:LO_monotones}, it was observed in \cite{Carrozza:2026qcf} that $R_G$ is non-increasing under any $\LO$ transformation between pure states, and thus defines a $\LO$-monotone on pure states (which justifies its interpretation as a multipartite entanglement entropy). The main fundamental observation the present work relies on is that this property continues to hold for an arbitrary $(G,H)$-spread, as a consequence of our assumption that $G \dom H$. 
\begin{prop}\label{prop:LO_monotones_more}
Let $G, H \in \cG_D$ be two $D$-colored graphs such that $G \dom H$.
If $\ket{\psi}$ and $\ket{\phi}$ are two finite-dimensional $D$-partite pure states such that $\ket{\psi}\toLO \ket{\phi}$, then:
\begin{equation}
    R_{G , H}(\ket{\psi}) \geq R_{G , H}(\ket{\phi})\,.   
\end{equation}
In other words, $R_{G,H}$ defines a $\LO$-monotone on the space of finite-dimensional $D$-partite pure states.
\end{prop}
\begin{proof}
   Let $\ket{\psi}$ and $\ket{\phi}$ be two finite-dimensional $D$-partite pure states such that $\ket{\psi}\toLO \ket{\phi}$. If $R_{G, H}(\ket{\psi})= +\infty$, then we have trivially $R_{G , H}(\ket{\psi}) \geq R_{G , H}(\ket{\phi})$ since $R_{G , H}(\ket{\phi}) \in [0 , +\infty]$. Suppose then that $R_{G , H}(\ket{\psi}) < + \infty$, and therefore that both $\tr_{G}(\ket{\psi})$ and $\tr_{H}(\ket{\psi})$ are non-vanishing. According to Cor.~\ref{cor:LO_monotones}, so are $\tr_{G}(\ket{\phi})$ and $\tr_{H}(\ket{\phi})$. It follows that
   \begin{equation}
       R_{G, H}(\ket{\psi})- R_{G, H}(\ket{\phi}) = \ln\left( \frac{|\tr_{H}(\ket{\psi})|}{|\tr_{G}(\ket{\psi})|} \frac{|\tr_{G}(\ket{\phi})|}{|\tr_{H}(\ket{\phi})|}  \right)
   \end{equation}
Furthermore, by Prop.~\ref{prop:charac_LO}, we can find $\ket{\eta}$ such that we have $\tr_{G}(\ket{\psi}) = \tr_{G}(\ket{\phi}) \tr_{G}(\ket{\eta})$ and $\tr_{H}(\ket{\psi}) = \tr_{H}(\ket{\phi}) \tr_{H}(\ket{\eta})$. In particular, both $\tr_{G}(\ket{\eta})$ and $\tr_{H}(\ket{\eta})$ must be non-vanishing. As a result, we find that
\begin{equation}
       R_{G, H}(\ket{\psi})- R_{G, H}(\ket{\phi}) = \ln\left( \frac{|\tr_{H}(\ket{\eta})|}{|\tr_{G}(\ket{\eta})|}  \right)\,,
   \end{equation}
and the right-hand-side is non-negative given our assumption that $\vert \tr_{G} (\cdot)\vert \leq \vert \tr_{H} (\cdot)\vert$.
\end{proof}

\begin{ex}\label{ex:entanglement_spread}
    In the bipartite case ($D=2$), we can fix $m\,,n\in \bb{N}$ with $2 \leq m < n$. Define $G$ to be the graph with $(n-1)$ connected components, each consisting of a cycle graph $C_m$. Likewise, define $H$ to be the graph with $(m-1)$ connected components, each consisting of a cycle graph $C_n$. We can verify that $G$ and $H$ obey the assumption of Prop.~\ref{prop:LO_monotones_more}. Indeed, if the bipartite pure state $\ket{\psi}$ has entanglement spectrum $\pmb{p}_\psi =(p_1 , \ldots , p_N)$,\footnote{Here, the $p_i$ denote the non-zero eigenvalues of the reduced density matrix on either subsystem, which are positive and sum to $1$. By convention, they are usually ordered in such a way that $p_1 \geq p_2 \geq \cdots \geq p_N > 0$.} then
    \begin{equation}
        \tr_{G}\left( \ket{\psi}\right) = \left( \sum_{i=1}^N (p_i)^m \right)^{n -1} \quad \mathrm{and} \quad \tr_{H}\left( \ket{\psi}\right) = \left( \sum_{i=1}^N (p_i)^n \right)^{m -1}\,.
    \end{equation} Defining $p=\frac{n -1}{m-1}$ and $q=\frac{n -1}{n - m}$, which are both larger or equal to $1$ and obey $\frac{1}{p}+\frac{1}{q}=1$, we can invoke H\"{o}lder's inequality to write
    \begin{align}
        \sum_{i=1}^N (p_i)^m = \sum_{i=1}^N (p_i)^{m-1/q} (p_i)^{1/q} \leq \left(\sum_{i=1}^N \left((p_i)^{m-1/q}\right)^p\right)^{1/p} \left(\sum_{i=1}^N \left((p_i)^{1/q}\right)^q\right)^{1/q} = \left(\sum_{i=1}^N (p_i)^{n}\right)^{\frac{m - 1}{n -1}}\,,
    \end{align}
which proves that $\tr_{G}(\ket{\psi})\leq \tr_{H}(\ket{\psi})$. According to Prop.~\ref{prop:LO_monotones_more}, $R_{G, H}$ defines a $\LO$-monotone, and a direct computation shows that
\begin{equation}\label{eq:ent_spread}
    R_{G, H}(\ket{\psi}) = (1-m)(1-n) \left( S_\rm{R}^{(m)}(\ket{\psi})- S_\rm{R}^{(n)}(\ket{\psi})\right)\,. 
\end{equation}
The quantity $\mathsf{D}_{m,n}(\cdot)\eqdef S_R^{(m)}(\cdot)- S_R^{(n)}(\cdot)$ has already been introduced as a $\LO$-monotone in \cite{Schmid:2020pgv}, as well as in \cite{Harrow:2010oah} in the limiting case $(m = 0, n = \infty)$, where it was referred to as \emph{entanglement spread}. It provides a quantitative measure of the ``spread'' of the probability distribution $\pmb{p}_\psi$, hence its name.
\end{ex}
The previous family of bipartite examples justifies in part our choice to refer to $R_{G,H}$ as a ``spread''. In the multipartite setting, even though we cannot in general characterize the $\LU$-orbit of a multipartite pure state $\ket\psi$ by a probability distribution $\pmb{p}_\psi$, we will see that there is still a sense in which our choice of terminology can be justified (see Thms.~\ref{thm:first_charac}, \ref{thm:second_charac} and Prop.~\ref{prop:ES_SHT}). 

Prop.~\ref{prop:LO_monotones_more} allows for a number of generalizations of the previous example to $D\geq 3$, based on the same recipe: 1) start from a graph $G$ with a combinatorial structure that makes it adapted to the use of an H\"older-type inequality; 2) deduce from that inequality that $|\tr_{G}(\cdot)|$ is upper-bounded by some $|\tr_{H}(\cdot)|$; 3) compute the $\LO$-monotone $R_{G , H}$. Let us provide two simple examples with $D=3$.
\begin{ex}\label{ex:monotone_K33_melon} 
Let $c \in \{1,2,3\}$. We take for $G$ the complete bipartite graph, usually denoted by $\PT_3$ in this context,\footnote{Indeed, $\tr_{\PT_3}(\ket\psi)$ computes the third moment of a \emph{partially transposed} (hence $\PT$) reduced density matrix associated with a pure state $\ket\psi$.} and for $H$ a particular melon graph with $6$ vertices, which we denote $\Mel_{3}^{(c)}$:
\begin{equation}
G = \PT_3 = \vcenter{\hbox{\includegraphics[scale=.9]{images/K_33}}}\,,\qquad H = \Mel_{3}^{(c)} = \vcenter{\hbox{\includegraphics[scale=.9]{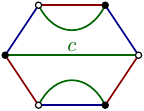}}}\,.
\end{equation}
For any finite-dimensional $3$-partite pure state $\ket{\psi}$, we can view $\tr_{\PT_3} (\ket{\psi})$ as an inner product between two tensors, and apply the Cauchy-Schwarz inequality to infer that (see Fig.~\ref{fig:graphical_proof_K33} for a graphical proof):
\begin{equation}
    \vert \tr_{\PT_3}\left(\ket{\psi}\right)\vert^2 \leq \left( \tr_{\Mel_{3}^{(c)}}\left( \ket{\psi}\right) \right)^2\,.
\end{equation}
Since $\tr_{\Mel_{3}^{(c)}}(\ket{\psi})> 0$ (it is the norm squared of some non-zero tensor),\footnote{The obvious axial symmetry of $\Mel_{3}^{(c)}$ justifies the interpretation as a norm squared. To see that this norm cannot vanish, one can invoke the general argument put forward in the proof of Prop.~\ref{prop:graph_eq_move1}.} this implies the looked for relation
\begin{equation}
    \vert \tr_{\PT_3}\left(\ket{\psi}\right)\vert \leq |\tr_{\Mel_{3}^{(c)}}\left( \ket{\psi}\right)| = \tr_{\Mel_{3}^{(c)}}\left( \ket{\psi}\right)\,.
\end{equation}
It follows that the quantity $R_{\PT_3, \Mel_{3}^{(c)}}$ defined by 
\begin{equation}
    R_{\PT_3, \Mel_{3}^{(c)}} (\ket{\psi}) = \ln \tr_{\Mel_{3}^{(c)}}(\ket{\psi}) - \ln |\tr_{\PT_3}(\ket{\psi})|
\end{equation}
when $|\tr_{\PT_3}(\ket{\psi})|\neq 0$ and $+\infty$ otherwise, is a $\LO$-monotone.
\end{ex}
\begin{figure}
    \centering
    \includegraphics[scale=.8]{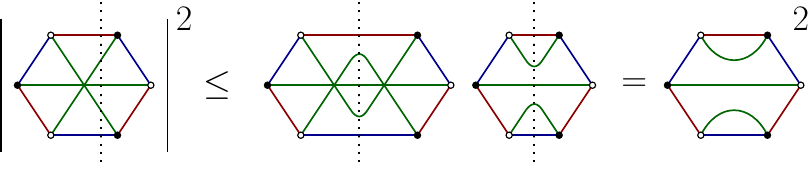}
    \caption{Graphical proof of the bound used in Ex.~\ref{ex:monotone_K33_melon}: the dotted lines are a visual guide indicating to which vectors the Cauchy-Schwarz inequality is being applied.}
    \label{fig:graphical_proof_K33}
\end{figure}

\begin{ex} By analogy with Example~\ref{ex:entanglement_spread}, let us fix again two integers $m, n$ such that $2 \leq m < n$, and introduce one possible tripartite generalization of the notion of entanglement spread. We let $G$ (resp. $H$) denote the graph with $(n-1)$ (resp. $(m-1)$) connected components, each consisting of a graph $\PT_{2m}^{(3)}$ (resp. $\PT_{2n}^{(3)}$). If $\rho_{12}$ denotes the reduced density matrix on $\H_1 \otimes \H_2$ of a finite-dimensional $3$-partite pure state $\ket\psi \in S(\H_1 \otimes \H_2 \otimes \H_3)$, we then have:
\begin{equation}
    \tr_{G} (\ket\psi) = \tr\left( (\rho_{12}^{T_2})^{2m}\right)^{(n-1)}=\left( \sum_{i=1}^N (\lambda_i)^m \right)^{n - 1}\quad \mathrm{and} \quad \tr_{H} (\ket\psi) = \tr\left( (\rho_{12}^{T_2})^{2n}\right)^{(m-1)} = \left( \sum_{i=1}^N (\lambda_i)^n \right)^{m - 1}\,,
\end{equation}
where $(\lambda_1 , \ldots , \lambda_N)$ denote the eigenvalues of the Hermitian operator $(\rho_{12}^{T_2})^2$. Since $(\rho_{12}^{T_2})^2= \rho_{12}^{T_2} \left( \rho_{12}^{T_2} \right)^\dagger$ is manifestly positive non-zero,\footnote{To check that $(\rho_{12}^{T_2})^2\neq0$, one may compute its trace, which reduces to $\tr_{\vcenter{\hbox{\includegraphics[scale=.3]{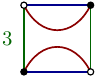}}}}(\ket\psi)$. This quantity is positive and can be lower-bounded (thanks to a Cauchy-Schwarz inequality) by $\frac{1}{N_3}\tr_{\vcenter{\hbox{\includegraphics[scale=.3]{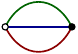}}}}(\ket\psi)=\frac{1}{N_3}>0$, where $N_3 = \dim (\H_3)$.} the eigenvalues $\{ \lambda_i \}$ are all positive (and not all vanishing). We can therefore apply H\"older's inequality as in Example~\ref{ex:entanglement_spread}, to obtain: 
\begin{equation}
0 <   \tr_{G} (\ket\psi) \leq \tr_{H} (\ket\psi)\,. 
\end{equation}
By Proposition \ref{prop:LO_monotones_more}, it follows that $R_{G , H}$ defines a $\LO$ monotone. Equivalently, the quantity
\begin{equation}\label{eq:3partite_spread}
    \mathsf{D}^{\PT^{(3)}}_{m, n} (\ket\psi) \eqdef S_{\PT^{(3)}}^{(m)} (\ket\psi) - S_{\PT^{(3)}}^{(n)}(\ket\psi)\,,
\end{equation}
where, for any $k\in\mathbb{N}$, 
\begin{equation}
    S_{\PT^{(3)}}^{(k)}(\ket\psi)\eqdef \frac{1}{1-k} \ln \tr_{\PT_{2k}^{(3)}}(\ket\psi)\,,
\end{equation}
defines a $\LO$-monotone. Equation \eqref{eq:3partite_spread} thus defines a tripartite generalization of the notion of entanglement spread from \eqref{eq:ent_spread}. Finally, we note that $\mathsf{D}^{\PT^{(3)}}_{m, n} (\ket\psi)$ vanishes if and only if $(\rho_{12}^{T_2})^2$ has a single non-zero eigenvalue $\lambda >0$, in which case $\spec(\rho_{12}^{T_2}) \setminus \{0\}\subset \{\sqrt{\lambda} , - \sqrt{\lambda}\}$ (and the multiplicities of each non-zero eigenvalue are constrained by the normalization condition $\tr(\rho_{12}^{T_2}) =1$).
\end{ex}
It is clear that the previous two constructions can be generalized to other families of graphs. Indeed, any graph can be viewed as an inner product (often times in many inequivalent ways), and one can therefore apply the Cauchy-Schwarz inequality to it (as in Ex.~\ref{ex:monotone_K33_melon}). Moreover, many invariants can be written in the form $\tr(M^m)$ with $M$ a positive matrix constructed out of $\ket{\psi}$, and we can apply H\"{o}lder's inequality to such graphs, as in Ex.~\ref{ex:entanglement_spread}. For instance, one can straightforwardly apply the latter construction to so-called \emph{realignment moments} and \emph{reflected entropies} (see \eg Refs.~\cite{Akers:2021pvd,Akers:2022zxr,Akers:2024pgq}), and thereby obtain other notions of ``tripartite entanglement spread". We will take advantage of these strategies in Sec.~\ref{sec:Examples} to introduce infinite families of entanglement spreads that will prove of particular interest in the remainder of the paper. 

While multipartite entanglement spreads are $\LO$-monotone by construction, it is natural to wonder how they behave in the context of other resource theories of entanglement, such as $\LOCC$. This question can be readily addressed with the tools developed in the present paper. However, in the interest of conciseness, the detailed analysis of the $\LOCC$ properties of multipartite entanglement spreads is deferred to a future publication (which will also include a similar analysis of other families of entropic functionals).

\section{Examples of multipartite entanglement spreads}\label{sec:Examples}

In this section, we introduce a number of infinite families of entanglement spreads that we view as particularly natural, and investigate their basic properties. Some of these families will prove of concrete relevance for the remainder of the paper. For each example of entanglement spread, we will be interested in: 1) characterizing the states for which the entanglement spread vanishes by means of graphical equations; 2) evaluating the entanglement spread of HT states; 3) evaluating the large-$N$ expectation value of the entanglement spread of the Haar-random state $\ket{\varphi_N}$. We will sometimes limit ourselves to providing partial answers to these questions. 

\subsection{Entanglement spreads from binary operations on \texorpdfstring{$D$}{D}-colored graphs}\label{sec:spreads_binary_op}

\subsubsection{Vertex-contraction operation} 

The first family of multipartite entanglement spreads we will be considering is constructed by means of binary operations on $D$-colored graphs. Using the terminology of Ref.~\cite{Carrozza:2026qcf}, we will concentrate on two such operations: the so-called \textit{disjoint union} and \textit{vertex-contraction} of graphs.

\begin{defi}
    Let $D \geq 2$ be an integer, and $G_1, G_2 \in \cG_D$ be two $D$-colored graphs. 

    The \emph{disjoint union} (or simply the \emph{union}) of $G_1$ and $G_2$, noted $G_1 \sqcup G_2$, is the $D$-colored graph whose connected components are precisely those of $G_1$ and $G_2$. In particular, one has that for any subset $B \subset \paa{1,\dots,D}$ with $\abs{B} \geq 2$: 
    \begin{equation}
        \kappa(G_1 \sqcup G_2\vert_B) = \kappa(G_1\vert_B) + \kappa(G_2\vert_B) \quad \rm{ and } \quad k(G_1 \sqcup G_2) = k(G_1) + k(G_2)\,.
    \end{equation}

    Let $v_{G_1}$ be a vertex of $G_1$ and $v_{G_2}$ a vertex of $G_2$, assumed to be of distinct colors (for definiteness, we may assume that $v_{G_1}$ is white, while $v_{G_2}$ is black). The \emph{vertex-contraction} of $G_1$ and $G_2$ along the pair $(v_{G_1}, v_{G_2})$, denoted by $G_1 \op_{(v_{G_1},v_{G_2})} G_2$, is obtained as follows: given the graphs $G_1^{(v_{G_1})}$ and $G_2^{(v_{G_2})}$ with $D$-colored open legs obtained respectively by removing the vertex $v_{G_1}$ in $G_1$ and the vertex $v_{G_2}$ in $G_2$, $G_1 \op_{(v_{G_1},v_{G_2})} G_2$ is the unique $D$-colored graph obtained by gluing the open legs of $G_1^{(v_{G_1})}$ and $G_2^{(v_{G_2})}$ while respecting the colored structure. This construction is illustrated in Fig.~\ref{fig:op}. Moreover, we observe that for any subset $B \subset \paa{1,\dots,D}$ with $\abs{B} \geq 2$:
    \begin{equation}
        \kappa(G_1 \op_{(v_{G_1},v_{G_2})} G_2\vert_B) = \kappa(G_1\vert_B) + \kappa(G_2\vert_B) - 1 \quad \rm{ and } \quad k(G_1 \op_{(v_{G_1},v_{G_2})} G_2) = k(G_1) + k(G_2) - 1\,.
    \end{equation}
    Whenever the information relative to the choice of pair of vertices $(v_{G_1}, v_{G_2})$ is not relevant to a particular argument, we will allow ourselves to keep this choice implicit and write $G_1 \op G_2$ instead of $G_1 \op_{(v_{G_1},v_{G_2})} G_2$.
    \begin{figure}[ht]
		\centering
		\includegraphics[height = 1.7cm]{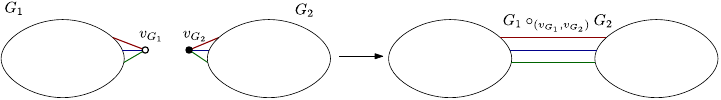}
		\caption{Left: the graphs $G_1$ and $G_2$ with the vertices $v_{G_1}$ and $v_{G_2}$ highlighted. Right: the result of the operation $G_1 \op_{(v_{G_1},v_{G_2})} G_2$.}
		\label{fig:op}
	\end{figure}
\end{defi}

Given a $D$-colored graph $H$ along with a marked vertex $v$, we will denote by $\bar{H}$ the graph obtained from $H$ by exchanging all the white vertices with black vertices and vice versa. The marked vertex $v$ (say white) of $H$ is mapped by this operation to the marked vertex $\bar v$ (thus black) of $\bar H$. With this convention in place, we can construct the vertex-contraction $H \op_{(v ,\bar v)} \bar H$ from the data of $H$ together with its marked vertex $v$. From the concrete illustration provided in Fig.~\ref{fig:mirror_op}, we may name this operation a \textit{mirror vertex-contraction of $H$ with respect to $v$}.
\begin{figure}
    \centering
    \includegraphics[height = 3.5cm]{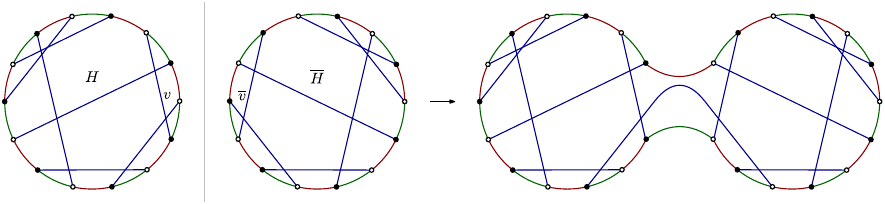}
    \caption{Left: for a given graph $H$ with a marked vertex $v$ we represented the associated graph $\bar H$ with the vertex $\bar v$ and we highlighted the mirror construction with the vertical gray line. Right: example of mirror vertex-contraction $H \op_{(v ,\bar v)} \bar H$ obtained from the previous choice of $H$ and $v$.}
    \label{fig:mirror_op}
\end{figure}

Colored graphs constructed by the use of the mirror vertex-contraction operation satisfy an interesting property as presented in the following lemma.

\begin{lem} \label{lem:dom_HopH}
    For any $H \in \cG_D$ and $v$ a marked vertex of $H$, $H \sqcup \bar H$ is dominated by $H \circ_{(v , \bar v)}\bar H$, \ie 
    \begin{equation}
        H \sqcup \bar H \dom H \circ_{(v , \bar v)}\bar H\,. 
    \end{equation}
\end{lem}
\begin{proof}
    The proof relies mainly on the Cauchy-Schwarz inequality, which we represent graphically below in the special case $D=3$. Indeed, we first remark that:
    \begin{equation}
        \Tr_{H \sqcup \bar H}(\ket{\psi}) = \Tr_H(\ket{\psi}) \cdot \Tr_{\bar H}(\ket{\psi}) = \Tr_H(\ket{\psi}) \cdot \bar{\Tr_{H}(\ket{\psi})} = \abs{\Tr_H(\ket{\psi})}^2 \geq 0\,.
    \end{equation}
    Then, the use of the Cauchy-Schwarz inequality yields, for any state $\ket\psi = \vcenter{\hbox{\includegraphics[scale=.5]{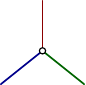}}}$:
    \begin{equation} \label{eq:ineq_R_H_barH}
        \Tr_{H\sqcup \bar H}(\ket{\psi}) \quad =\quad  \vcenter{\hbox{\includegraphics[scale=1.5]{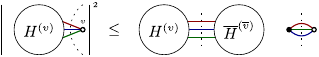}}} \quad = \quad \Tr_{H \op_{(v,\bar v)}\bar H}(\ket{\psi}) \cdot 1 \,.
    \end{equation}
    It follows that $\Tr_{H \op_{(v,\bar v)}\bar H}(\ket{\psi})$ is also positive (which is manifest given the axial symmetry of $H \op_{(v,\bar v)}\bar H$) and $H \sqcup \bar H \dom H \circ_{(v , \bar v)}\bar H$. In summary:
    \begin{equation}
        \Tr_{H \sqcup \bar{H}}(\cdot) \leq \Tr_{H \op_{(v,\bar v)}\bar H}(\cdot) \quad \rm{ and } \quad  \abs{\Tr_{H \sqcup \bar{H}}(\cdot)} \leq \vert\Tr_{H \op_{(v,\bar v)}\bar H}(\cdot)\vert\,. 
    \end{equation}
\end{proof}

Thanks to the previous lemma, one can associate a multipartite entanglement spread with any given graph $H$ together with a marked vertex $v$. The next proposition formalizes this construction.

\begin{prop}\label{prop:Spread_contractionOp}
    Let $D \geq 2$ and $H \in \cG_D$ with a marked vertex $v$. For any finite-dimensional $D$-partite pure state $\ket{\psi}$, we define:
    \begin{equation}\label{def:vertex_cont_spread}
        C_H^{[v]}(\ket \psi) \eqdef R_{H \sqcup \bar H , H \circ_{(v , \bar v)}\bar H}(\ket{\psi})= \begin{cases} \ln\left(\tr_{H \circ_{(v , \bar v)}\bar H}(\ket{\psi})\right) - 2 \ln\abs{\tr_{H}(\ket{\psi})} &\; \mathrm{if}\; \tr_{H}(\ket{\psi}) \neq 0 \,, \\ 
        +\infty &\; \mathrm{otherwise.}\end{cases}
    \end{equation}
    The spread $C_H^{[v]}$ takes value in $[0,+\infty]$, is $\LU$-invariant, and defines a $\LO$-monotone on pure states.
\end{prop}

\begin{proof}
    Being clear from Def.~\ref{def:entanglement_spread}, Prop.~\ref{prop:LO_monotones_more} and Lem.~\ref{lem:dom_HopH}, the proof is omitted.
\end{proof}

Multipartite entanglement spreads constructed by the mirror vertex-contraction operation will be particularly relevant in Sec.~\ref{sec:SHTstates}. 

To illustrate the previous construction, let us first see what it entails in the bipartite case, \ie for $D=2$.
\begin{ex}\label{ex:toyEx-bipartite}
     Given an integer $n \geq 1$, a cyclic graph $C_n$ and a vertex $v$ of $C_n$, one can show that the mirror vertex-contraction operation on $C_n$ with respect to $v$, \ie $C_n \op_{(v, \bar v)} \bar{C_n}$, is independent of $v$ and is equal to $C_{2n-1}$. Remarkably, the entanglement spread constructed by the use of $C_n \sqcup \bar{C_n}$ and $C_n \op_{(v, \bar v)} \bar{C_n}$ can naturally be related to the bipartite entanglement spread introduced in Ex.~\ref{ex:entanglement_spread}, thanks to the following relation: for any finite-dimensional tripartite pure state $\ket\psi$,
    \begin{equation}
        C_{C_n}^{[v]}(\ket \psi) = 2(n-1) \mathsf{D}_{n,2n-1}(\ket \psi)\,.
    \end{equation}
    The same remark goes for any $D$-colored cyclic graph $C_n(B)$ associated with the bipartition of $\paa{1,\dots,D}$ of the form $B \sqcup \bar{B}$. In particular, one has: $C_n(B) \op_{(v, \bar v)} \bar{C_n}(B) = C_{2n-1}(B)$ for any vertex $v$ of $C_n(B)$. 
\end{ex}
Let us now move to the multipartite setting, with a toy example for $D=3$.
\begin{ex} \label{ex:toyEx}    
    We assume $D=3$ and consider the graph associated with the third moment of the partial transpose denoted by $\PT_3$ in Ref.~\cite{Carrozza:2026qcf} and sometimes referred to as the complete bipartite graph $\K_{3,3}$. By means of a recursive use of the vertex-contraction operation with $\PT_3$, we construct an infinite family of graphs $\cal{I} = \{I_n\}_{n\in\bb{N}^*}$. The recursive construction is defined by the relation:
    \begin{equation}
        \forall n \in \bb{N}^*, \quad I_n = \begin{cases} 
        \PT_3 & \text{if } n=1 \\ 
        I_1 \op_{(v,\bar v)} I_{n-1} & \text{otherwise} 
        \end{cases} \,,
    \end{equation}
    where $v$ (resp.~$\bar{v}$) is the vertex of $I_1$ (resp.~$I_{n-1}$) highlighted in Fig.~\ref{fig:I_n} by the use of a gray disc.\footnote{The index $n$ counts the genus of the underlying orientable surface into which the graph $I_n$ can be embedded (while preserving the orientation of the edges around each vertex). Indeed, $I_n$ is a connected graph with $V(I_n) = 2(2n + 1)$ vertices and $E(I_n) = 3V(I_n)/2$ edges. What is more, $I_1$ is a maximally single-trace graph and in Sec.~6.2 of Ref.~\cite{Carrozza:2026qcf}, it was shown that the vertex-contraction operation preserves this property. Since, by definition, any $G \in \cG_3^{(\rm{MST})}$ has $F(G) = 3$, we have: 
    \begin{equation}
        g(I_n) \eqdef 1  - \frac 1 2 \chi(I_n) = 1 - \frac 1 2 \pac{V(I_n) - E(I_n) + F(I_n)} = n \,,
    \end{equation}
    with $\chi$ the so-called Euler characteristic.}
    Since $I_n^{(c)} \op_{(v,\bar v)} \bar{I_n^{(c)}} = I_{2n}^{(c)}$, one can define the following multipartite entanglement spread: for any finite dimensional tripartite pure state $\ket\psi$, 
    \begin{equation}
        C_{I_n}^{[v]}(\ket \psi) = \ln \Tr_{I_{2n}}(\ket{\psi}) - 2 \ln \vert \Tr_{I_n}(\ket{\psi}) \vert
    \end{equation}
    if $\tr_{I_n}(\ket\psi ) \neq 0$, and is set to $\infty$ otherwise.
    This toy example illustrates the variety of multipartite entanglement spreads that can be constructed by using the mirror vertex-contraction operation.
    
    \begin{figure}
        \centering
        \includegraphics[width = .65\textwidth]{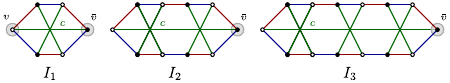}
        \caption{From left to right: the colored graphs associated with $I_1$, $I_2^{(c)}$ and $I_3^{(c)}$ where $c$ is the black color.}
        \label{fig:I_n}
    \end{figure}
\end{ex}

Interestingly $C_{H}^{[v]}$ evaluates to $0$ on any HT state. This is in analogy with the bipartite setting, in the context of which any entanglement spread $\mathsf{D}_{m,n}$ (defined in Ex.~\ref{ex:entanglement_spread}) does indeed vanish for any HT state: indeed, bipartite HT states are nothing but states with flat entanglement spectrum, and $\mathsf{D}_{m,n} (\ket \psi) = 0$ if and only if $\ket\psi$ has flat entanglement spectrum. 
\begin{prop}\label{prop:vertex_spread_HT}
    For any $H \in \cG_D$ and any HT state $\ket{\psi_\alpha}$ (with weight function $\alpha$), we have:
    \begin{equation}\label{eq:vanishing_alpha_1}
        C_{H}^{[v]}(\ket{\psi_\alpha}) = 0\,.
    \end{equation}
\end{prop}
\begin{proof}
    For any $B \subset\{1, \ldots , D\}$ with $\abs{B} \geq 2$, we note that: $\kappa\left((H \sqcup \bar H)\vert_B\right)= \kappa\left((H \circ_{(v , \bar v)}\bar H)\vert_B \right) + 1$ and $k(H \sqcup \bar H)= 2k(H) =k \left( H \circ_{(v , \bar v)}\bar H \right)+1$. As a result,
    \begin{align}
        \tr_{H \sqcup \bar H} \left( \ket{\psi_\alpha} \right) &= \prod_{B \subset\{1, \ldots, D\}\, |\, |B|\geq 2} \alpha(B)^{\kappa\left((H \sqcup \bar H)\vert_B\right)- 2k(H)} \\ &= \prod_{B \subset\{1, \ldots, D\}\, |\, |B|\geq 2} \alpha(B)^{\kappa\left((H \circ_{(v , \bar v)}\bar H)\vert_B\right)- k(H \circ_{(v , \bar v)}\bar H)} \nonumber = \tr_{H \circ_{(v , \bar v)}\bar H} \left( \ket{\psi_\alpha} \right)\,, 
    \end{align}
from which Eq.~\eqref{eq:vanishing_alpha_1} follows.   
\end{proof}

More broadly the quantum states that minimize $R_{H \sqcup \bar{H}, H \op_{(v,\bar{v})} \bar H}$ make it vanish and can be characterized by a graphical equation. Indeed, the minimizers must saturate the inequality Eq.~\eqref{eq:ineq_R_H_barH}, and this occurs precisely when the vectors $H^{(v)}$ and $\ket{\psi}$ are linearly dependent. Since $\ket{\psi}$ is a non-zero vector, we can find a complex number $\rho\in\bb{C}$ such that:
\begin{equation} \label{eq:minimize_eq}
    \pa{\tr_{H^{(v)}}(\ket\psi)}_{i_1 \dots i_D} = \rho \, \psi_{i_1 \dots i_D} \,.  
\end{equation}
The state $\ket{\psi}$ being normalized, contracting the previous equation with $\bar{\psi}_{j_1 \dots j_D}$ leads to: $\rho = \Tr_{H}(\ket{\psi})$. In particular, one must have $\rho \neq 0$, otherwise one would have $R_{H \sqcup \bar H , H \circ_{(v , \bar v)}\bar H}(\ket{\psi})= + \infty$. Graphically, Eq.~\eqref{eq:minimize_eq} thus translates (for $D=3$) into 
\begin{equation} \label{eq:CS_minimizers}
    \vcenter{\hbox{\includegraphics[scale=1.5]{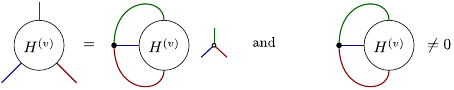}}}\,.
\end{equation}

We start with a concrete example for $D=2$ and then focus on the previous multipartite toy examples.
\begin{ex}\label{ex:vertex_contrac_bipartite}
    From Ex.~\ref{ex:toyEx-bipartite}, we observe that assuming, for example, $C_{C_4}^{[v]}(\ket \psi) = R_{C_4 \sqcup \bar{C_4}, C_7}(\ket{\psi}) = 0$ leads to the following graphical equation:
    \begin{equation}
        \vcenter{\hbox{\includegraphics[scale=1.8]{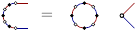}}}\,,
    \end{equation}
    Thus, starting from $C_n$, one may iteratively replace portions of the cycle to obtain a collection of cycles with $k=4$ (\ie $C_4$), and a remainder with $k\in \{1, 2, 3\}$. Precisely, we have:
    \begin{equation}
        \forall n \geq 4, \, \quad \Tr_{C_n}(\ket{\psi}) = \Tr_{C_4}(\ket{\psi})^{\lfloor\frac{n-1}{3} \rfloor}\begin{cases}
            1 & \rm{ if } \, n \equiv 1 \pmod{3} \,, \\
            \Tr_{C_2}(\ket{\psi}) & \rm{ if }\,  n \equiv  2 \pmod{3} \,, \\
            \Tr_{C_3}(\ket{\psi}) & \rm{ if }\,  n \equiv  0 \pmod{3} \,.
        \end{cases}
    \end{equation}
    Although a complete analysis of the minimizers' entanglement spectrum could be performed as presented later in Sec.~\ref{sec:bipartite_charac}, it is beyond the scope of this example.
\end{ex}
\begin{ex}
When $D=3$, the simplest example of non-trivial pair $(H,v)$ one might consider is the marked quartic melonic graph
    \begin{equation}
     H= \vcenter{\hbox{\includegraphics[scale=.5]{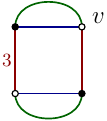}}}    
    \end{equation}
A minimizer $\ket\psi = \vcenter{\hbox{\includegraphics[scale=.5]{images/psi.pdf}}}$ of the associated entanglement spread then obeys the graphical equation
\begin{equation}
 \vcenter{\hbox{\includegraphics[scale=.5]{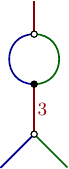}}} =      \vcenter{\hbox{\includegraphics[scale=.5]{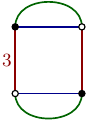}}} \; \vcenter{\hbox{\includegraphics[scale=.5]{images/psi.pdf}}} \,.
\end{equation}
If we reinterpret this equation in the coarse-grained bipartite Hilbert space $\left(\H_1 \otimes \H_2\right)\otimes \H_3$, we can equivalently represent it as:
\begin{equation}
 \vcenter{\hbox{\includegraphics[scale=.5]{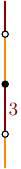}}} =      \vcenter{\hbox{\includegraphics[scale=.5]{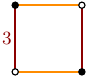}}} \; \vcenter{\hbox{\includegraphics[scale=.5]{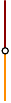}}}
 \,,
\end{equation}
where the color orange is now associated with the coarse-grained Hilbert space $\H_1 \otimes \H_2$. Thus, this example falls in the same category as Ex.~\ref{ex:vertex_contrac_bipartite}, namely: it only captures information about bipartite entanglement.  
\end{ex}
To generate examples of our construction that do not reduce to properties of bipartite entanglement, one needs to start from a graph $H$ that is \emph{genuinely multipartite} (see Ref.~\cite{Carrozza:2026qcf} for a precise definition of that notion). An infinite family of such examples is provided next. 
\begin{ex}
    From the previous multipartite toy example introduced in Ex.~\ref{ex:toyEx}, one observes that if $\ket{\psi}$ is such that $C^{[v]}_{I_1}(\ket \psi)=0$, \ie $\ket{\psi}$ is a minimizer of $C^{[v]}_{I_1}$, then, one has that:
    \begin{equation}
        \vcenter{\hbox{\includegraphics[scale=1.8]{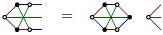}}}\,,
    \end{equation}
    thus leading, by recursion, to:
    \begin{equation}
        \forall n \in \bb{N}^*, \quad \Tr_{I_n}(\ket{\psi}) = \pac{\Tr_{\PT_3}(\ket{\psi})}^n \,.
    \end{equation}
    As a direct consequence, one has: $C^{[v]}_{I_n}(\ket{\psi}) = 0$ for any $n\geq 1$.
\end{ex}

The next subsection will focus on finding graphical equations satisfied by minimizers of entanglement spreads, typically in the same spirit of Eq.~\eqref{eq:minimize_eq}. Prior to that, we will be interested in computing -- both generally and for specific examples -- the expectation value $\langle C^{[v]}_{H}(\ket{\varphi_N})\rangle$ for random pure states $\ket{\varphi_N}$ drawn according to the Haar measure on $\left(\mathbb{C}^N\right)^{\otimes D}$. 

We first recall that for any graphs $G,H$, the quantity $R_{G,H}$ involves logarithms of the trace-invariants associated with the graphs $G$ and $H$. However, as highlighted in Sec.~6.3.1 of Ref.~\cite{Carrozza:2026qcf}, computing the expectation value of $R_{G,H}$ becomes easier whenever both $G$ and $H$ satisfy the conditions presented in Prop.~\ref{prop:mean_log}. Namely, the lemma is requiring that $G$ and $H$: 1) satisfy the large-$N$ factorization criterion, and 2) follow a power law. 

In Cor.~6.20 of Ref.~\cite{Carrozza:2026qcf}, several examples of trace-invariants satisfying both conditions of Prop.~\ref{prop:mean_log} were presented. These include, among others, the cyclic graphs $C_k(B)$ associated with arbitrary bipartitions $B \cup \bar B$, as well as the graphs corresponding to the realignment moments and those associated with the even moments of the partial transpose.

Regarding the study of the multipartite entanglement spreads constructed using the mirror vertex-contraction operation, the following lemma will come in handy. 

\begin{lem} \label{lem:powerlawHoH}
    Let $D \geq 2$, $H \in \cG_D$ be a $D$-colored graph with a marked vertex $v$. If $H$ follows a power law, then, so does $H \op_{(v,\bar{v})} \bar H$.
\end{lem}

\begin{proof}
    Assume $H$ follows a power law, \ie ther exists $(K,r) \in \bb{R}_+^* \times \bb{R_+}$ such that $\abs{\Tr_H(\ket{\psi})} \geq K N^{-r}$. Using Lem.~\ref{lem:dom_HopH}, one directly arrives at the relation: 
    \begin{equation}
        \vert \Tr_{H \op_{(v,\bar{v})} \bar H}(\ket{\psi}) \vert \geq \abs{\Tr_H(\ket{\psi})}^2 \geq K^2 N^{-2r} \,,
    \end{equation}
    which proves that $H \op_{(v,\bar{v})} \bar H$ follows a power law.
\end{proof}

Once it is known that $H$ follows a power law, the computation of the expectation value of $C^{[v]}_{H}$ in the Haar-random state $\ket{\varphi_N}$, requires a focus on the large-$N$ factorization criterion. Clearly, if, in addition to the power law condition, both $H$ and $H \op_{(v,\bar{v})} \bar H$ satisfy the large-$N$ factorization criterion, then the expectation and the logarithms can be exchanged in the large dimension limit $N\to \infty$, which simplifies the computation of $\langle C^{[v]}_{H}(\ket{\varphi_N}) \rangle$ (at leading order). Interestingly, as presented in the next proposition, using mainly Thm.~4.1 of Ref.~\cite{Factorization2026}, one can extract a condition that depends only on the graph $H$.

\begin{prop} \label{prop:Haar_HoH}
    Let $D \geq 2$, $H\in \cG_D$ with a marked vertex $v$, and let $(\ket{\varphi_N})_{N\in\bb{N}^*}$ be a sequence of $D$-partite Haar-random states of local dimension $N$. If $H$ follows a power law and $\Delta(H) < \frac{D(D-1)}{8}$, then: 
    \begin{equation} \label{eq:meanR_HopH}
        \langle C^{[v]}_{H}(\ket{\varphi_N}) \rangle \underset{N \to \infty}{=} \frac{2}{D-1}\pa{2 \Delta(H) - \Delta(H \op_{(v,\bar{v})} \bar H)} \ln N + \ln\left( \mu_{H \circ_{(v,\bar v)} \bar H}\right) - 2 \ln\left( \mu_H \right) + o(1) \,.
    \end{equation}
\end{prop}

\begin{proof}
    Let $G \in \cG_D$ be a $D$-colored graph. It was shown in Ref.~\cite{Factorization2026} that if $\Delta(G) < \frac{D(D-1)}{4}$, then, the large-$N$ factorization criterion is satisfied, \ie 
    \begin{equation}
        \mean{\Tr_G(\ket{\varphi_N})\Tr_{\bar G}(\ket{\varphi_N})} \underset{N \to \infty}{\sim}\mean{\Tr_G(\ket{\varphi_N})}\mean{\Tr_{\bar G}(\ket{\varphi_N})} \,.
    \end{equation}
    Hence, the condition $\Delta(H) < \frac{D(D-1)}{8}$ ensures that $H$ satisfies the large-$N$ factorization criterion.
    
    What is more, from a graph $\widehat H \in \cG_{D+1}(H)$ satisfying $\Delta_0(\widehat H) = \Delta(H)$, one can construct the graph $\widehat H \op_{(v,\bar v)} \widehat{\bar H} \in \cG_{D+1}(\widehat{H \op_{(v,\bar v)} \bar H})$, which obeys: 
    \begin{align}
        \Delta_0(\widehat H \op_{(v,\bar v)} \widehat{\bar H}) &=  \frac{D(D-1)}{4} k (H \op_{(v,\bar v)} \bar H) + \frac{1}{2} F(H \op_{(v,\bar v)} \bar H) - \frac{D-1}{2}F_0 (\widehat H \op_{(v,\bar v)} \widehat{\bar H}) \\
        &= \frac{D(D-1)}{4} \left( 2k (H) - 1\right)+ \frac{1}{2} \left( 2F(H) - \frac{D(D-1)}{2}\right)- \frac{D-1}{2} \left( 2F_0 (\widehat H) - D \right) \\
        &=2 \Delta(H)\,. 
    \end{align}
    It follows that    
    \begin{equation}
        \Delta(H \op_{(v,\bar v)} \bar H) \leq \Delta_0(\widehat H \op_{(v,\bar v)} \widehat{\bar H}) = 2 \Delta(H) \,,
    \end{equation}
    and, together with the condition $\Delta(H) < \frac{D(D-1)}{8}$, we obtain
    \begin{equation}
        \Delta(H \op_{(v,\bar v)} \bar H) < \frac{D(D-1)}{4}\,.
    \end{equation}
    As a consequence, $H \op_{(v,\bar v)} \bar H$ also satisfies the large-$N$ factorization criterion. 

    Finally, for any $G \in \cG_D$, the asymptotic expectation value of $\tr_G$ in the Haar-random state $\ket{\varphi_N}$ obeys:
    \begin{equation}
        \mean{\Tr_G(\ket{\varphi_N})} \underset{N \to \infty}{\sim} \mu_G N^{\kappa(G) - k(G) - \frac{\omega(G)}{D-1} - \frac{2\Delta(G)}{D-1}} \,,
    \end{equation}
    where $\mu_G \in \mathbb{N}^*$ is a degeneracy factor. Furthermore, it was shown in Sec.~6.2 of Ref.~\cite{Carrozza:2026qcf} that $\kappa(H \op_{(v,\bar v)} \bar H) = 2\kappa(H) -1$, $k(H \op_{(v,\bar v)} \bar H) = 2k(H) - 1$ and $\omega(H \op_{(v,\bar v)} \bar H) = 2\omega(H)$. Using these relations, together with the fact that the logarithm and expectation values commute at large $N$ for these graphs, leads to Eq.~\eqref{eq:meanR_HopH}.
\end{proof}

\begin{ex} 
    For instance, the even moments of the partial transpose $\PT_{2n}^{(c)}$ satisfy the power law condition and are compatible, \ie $\Delta(\PT_{2n}^{(c)}) = 0$. Hence, the last proposition yields:
    \begin{equation}
        \forall n \geq 1, \, \forall c \in \paa{1,2,3}, \quad \langle C^{[v]}_{\PT_{2n}^{(c)}}(\ket{\varphi_N}) \rangle \underset{N \to \infty}{=} O(1) \,.
    \end{equation}
    In the case $n=2$, point 5. of Thm.~6.4 of Ref.~\cite{Carrozza:2026qcf}  allows us to refine the last result to: 
    \begin{equation}
        \langle C^{[v]}_{\PT_{4}^{(c)}}(\ket{\varphi_N}) \rangle \underset{N \to \infty}{=} o(1)\,.
    \end{equation}
    Indeed, the graph $\PT_{4}^{(c)}$ is compatible (see Ref.~\cite{Dong:2024gud}) and has only faces of size four or less,\footnote{The size of a face is defined as the number of edges along it.} and this result implies $\mu_{\PT_{4}^{(c)} \op_{(v,\bar v)} \bar{\PT_{4}^{(c)}}} = (\mu_{\PT_4^{(c)}})^2$. 
\end{ex}

\subsubsection{Flip operation}
Albeit not relevant for Secs.~\ref{sec:HTstates}--\ref{sec:SHTstates}, we remark that an appropriate use of the Cauchy-Schwarz inequality reveals another well-known binary operation on graphs: the \emph{flip} (see \eg Ref.~\cite{Carrozza:2026qcf}). First, we formally define the flip and mirror flip operations in a manner analogous to the vertex-contraction operation.

\begin{defi}\label{def:mirror_flip}
    Let $D \geq 2$, and let $G_1,G_2 \in \cG_D$ be two $D$-colored graphs.
    
    Let $c\in \paa{1,\dots,D}$ be a color and $e_{G_1}$ (resp.~$e_{G_2}$) be an edge of color $c$ of $G_1$ (resp.~$G_2$). The \emph{flip} of $G_1$ and $G_2$ along the edges $(e_{G_1},e_{G_2})$, denoted by $G_1 \flip_{(e_{G_1},e_{G_2})} G_2$, is obtained by cutting the edges $e_{G_1}$ and $e_{G_2}$ in half and gluing them back together in the only other way that preserves the $\LU$-invariance. This construction is illustrated in Fig.~\ref{fig:flip}. For any subset $B \subset \paa{1,\dots,D}$ with $\abs{B} \geq 2$, we observe that:
    \begin{equation} \label{eq:flipOnKappa}
        \kappa(G_1 \flip_{(e_{G_1},e_{G_2})} G_2 \vert_B) =  \begin{cases} 
        \kappa(G_1\vert_B) + \kappa(G_2\vert_B) - 1 & \text{if } c \in B \\ 
        \kappa(G_1\vert_B) + \kappa(G_2\vert_B)  & \text{otherwise} 
        \end{cases} \quad \rm{ and } \quad k(G_1 \flip_{(e_{G_1},e_{G_2})} G_2) = k(G_1) + k(G_2)\,.
    \end{equation}
    
    \begin{figure}[ht]
		\centering
		\includegraphics[height = 1.7cm]{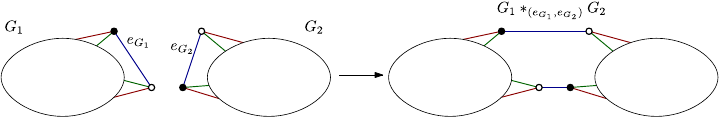}
		\caption{Left: the graphs $G_1$ and $G_2$ with the edges $e_{G_1}$ and $e_{G_2}$ highlighted. Right: the result of the operation $G_1 \flip_{(e_{G_1},e_{G_2})} G_2$.}
		\label{fig:flip}
	\end{figure}
    
    Given $G \in \cG_D$ a colored graph with a marked edge $e$ of color $c$, we let $\bar{G}$ be the graph obtained by exchanging white vertices of $G$ with black vertices and vice-versa. We also denote by $\bar{e}$ the image of $e$ in $\bar G$ under this operation (in other words, $\bar e$ is the ``mirror'' of $e$ in $\bar G$). The \emph{mirror flip operation} on $G$ is the graph operation producing the graph $G \flip_{(e,\bar{e})} \bar G$. This operation is exemplified in Fig.~\ref{fig:mirror_flip}.
    
    \begin{figure}
        \centering
        \includegraphics[height = 3.5cm]{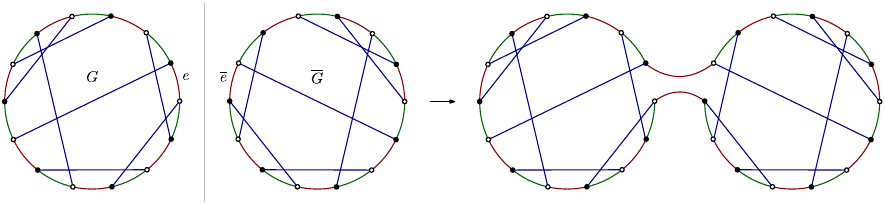}
        \caption{Left: for a given graph $G$ with a marked edge $e$ we represented the associated graph $\bar G$ with the edge $\bar e$ and we highlighted the mirror construction with the vertical gray line. Right: example of mirror flip $G \flip_{(e ,\bar e)} \bar G$ obtained from the previous choice of $G$ and $e$.}
        \label{fig:mirror_flip}
    \end{figure}
\end{defi}

From the mirror flip operation, one can construct multipartite entanglement spreads that are $\LU$-invariant, positive and $\LO$-monotone as exemplified by the following proposition.
\begin{prop} \label{prop:ES_flip}
    Let $D \geq 2$, and $G \in \cG_D$ with a marked edge $e$. For any finite-dimensional $D$-partite pure state $\ket\psi$, we define:
    \begin{equation} \label{eq:ES_flip}
        F_{G}^{[e]}(\ket{\psi}) \eqdef \begin{cases} \ln\left(\tr_{G \flip_{(e , \bar e)}\bar G}(\ket{\psi})\right) - 2 \ln\abs{\tr_{G}(\ket{\psi})} + \ln \mathrm{rank}\left( \tr_{G^{(e)}}(\ket\psi)\right) &\; \mathrm{if}\; \tr_{G}(\ket{\psi}) \neq 0 \,, \\ 
    +\infty &\; \mathrm{otherwise.}\end{cases}
    \end{equation}
    $F_{G}^{[e]}$ takes value in $[0,+\infty]$, is $\LU$-invariant, and defines a $\LO$-monotone on pure states.
\end{prop}

\begin{proof} 
    Let $G\in \cG_D$ be a $D$-colored graph with a marked edge $e$, say of color $c\in \{1, \ldots , D\}$. Consider a $D$-partite pure state $\ket\psi \in \H_1 \otimes \cdots \otimes \H_D$ (with $\H_{c'}$ finite-dimensional for any $c'\in \{1, \ldots , D\}$). Since $\Tr_{G^{(e)}}(\ket{\psi})$ belongs to $\H_c \otimes \H_c^*$, we can view it as a linear map on $\H_c$. By application of the singular value decomposition theorem, we can find an integer $N_\psi \in \bb{N}^*$, two orthonormal bases $(\vert u_i^{\psi} \rangle)_{1 \leq i \leq N_\psi}$ and $(\vert v_i^{\psi} \rangle)_{1 \leq i \leq N_\psi}$ of $\H_c$, and positive numbers $(\lambda_i^{\psi})_{1 \leq i \leq N_\psi}$ such that:
    \begin{equation}
        \Tr_{G^{(e)}}(\ket{\psi}) = \sum_{i=1}^{N_\psi} \lambda_i^{\psi} \vert u_i^{\psi} \rangle \langle v_i^{\psi} \vert\,.
    \end{equation}
    In particular, we observe that $N_{\psi} = \rm{rank}(\Tr_{G^{(e)}}(\ket{\psi}))$. Let us then introduce the orthogonal projector
    \begin{equation}
        P_\psi \eqdef \sum_{i=1}^{N_\psi} \vert u_i^{\psi} \rangle \langle u_i^{\psi} \vert\,,\quad \rm{obeying} \quad \Tr P_\psi = N_\psi\,.
    \end{equation}
Employing a graphical representation of this linear operator as $P_\psi  = \vcenter{\hbox{\includegraphics[scale=2]{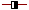}}}$, we can derive the following Cauchy-Schwarz inequality:
    \begin{equation} \label{eq:CS_flip}
        \abs{\tr_G(\ket\psi)}^2 = \vcenter{\hbox{\includegraphics[scale=1.3]{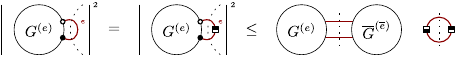}}} = \tr_{G \flip_{(e , \bar e)}\bar G}(\ket{\psi}) \times N_\psi \,,
    \end{equation}
    where we have made use of the relations $P_\psi \cdot \Tr_{G^{(e)}}(\ket{\psi}) = \Tr_{G^{(e)}}(\ket{\psi})$ and  $\tr P_\psi^2 = \tr P_\psi = N_\psi$. 
    
    Now, for any finite-dimensional $D$-partite pure state $\ket\psi$, let us define\footnote{This quantity is well-defined since the trace of the linear map $\Tr_{G^{(e)}}(\ket{\psi})$ is $\tr_G (\ket\psi) \neq0$, hence $\mathrm{rank}(\Tr_{G^{(e)}}(\ket{\psi})) >0$.}
    \begin{equation}
        W_{G^{(e)}}(\ket{\psi}) \eqdef \frac{\abs{\Tr_G(\ket{\psi})}^2}{ \mathrm{rank}(\Tr_{G^{(e)}}(\ket{\psi})) }\,.
    \end{equation}
$W_{G^{(e)}}$ is a (non-polynomial) $\LU$-invariant: indeed, the numerator is clearly invariant, and so is the denominator since $\Tr_{G^{(e)}}(\ket{\psi})$ is mapped to the conjugate element $U_c \Tr_{G^{(e)}}(\ket{\psi}) U_c^{-1}$ under a local unitary transformation $U_1 \otimes \cdots \otimes U_D$. Furthermore, for any finite-dimensional $D$-partite pure state $\ket\psi$ such that $\tr_G(\ket\psi)\neq 0$, equation \eqref{eq:CS_flip} implies:
\begin{equation}
 0 \leq W_{G^{(e)}}(\ket\psi) \leq \tr_{G \flip_{(e , \bar e)}\bar G}(\ket{\psi}) \leq 1\,.
\end{equation}
We can thus apply the same reasoning as in Prop.~\ref{prop:LO_monotones_more} to the pair of invariants $(W_{G^{(e)}}, \tr_{G \flip_{(e , \bar e)}\bar G})$ to conclude that the function $F_G^{[e]}\eqdef \ln \circ \tr_{G \flip_{(e , \bar e)}\bar G} - \ln \circ W_{G^{(e)}}$ defines a $\LO$-monotone valued in $[0, +\infty[$ on the set of pure states $\ket\psi$ obeying $\tr_G (\ket\psi ) \neq 0$ (more precisely, this function is non-increasing under $\LO$ transformations on that set). Furthermore, extending $F_G^{[e]}$ to the full state space by attributing the value $+ \infty$ to any state $\ket\psi$ such that $\tr_G (\ket\psi ) = 0$ preserves the monotonicity.\footnote{Indeed, by Cor.~\ref{cor:LO_monotones}, it is not possible to have a $\LO$-relation of the form $\ket\psi \toLO \ket\phi$ with $\tr_G(\ket\psi)\neq 0$ (thus $F_G^{[e]} (\ket\psi)\in [0, +\infty[$) and $\tr_G(\ket\phi)= 0$ (thus $F_G^{[e]} (\ket\phi)=+\infty$).}
\end{proof}

Without going into too much detail, we will review a few results about this flip operation, in analogy with what we did for the vertex-contraction operation. We will firstly be interested in the minimizers of the entanglement spreads defined in Prop.~\ref{prop:ES_flip}. The minimizers satisfy a graphical equation presented below
\begin{equation} \label{eq:CS_minimizers_flip}
    \vcenter{\hbox{\includegraphics[scale=1.5]{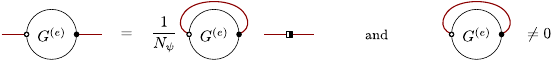}}}\,,
\end{equation}
and -- as presented in the following proposition -- it is not difficult to use Sec.~6.2 of Ref.~\cite{Carrozza:2026qcf} to prove that the HT states are among the minimizers.

\begin{prop}\label{prop:flip_spread_HT}
    For any $G \in \cG_D$ with a marked edge $e$, and any HT state $\ket{\psi_\alpha}$ (with weight function $\alpha$), we have:
    \begin{equation}
        F_{G}^{[e]}(\ket{\psi_\alpha}) = 0\,.
    \end{equation}
\end{prop}
\begin{proof}
    Let $c \in \paa{1,\dots,D}$ be the color of the marked edge $e$ of $G$ and $N_c = \dim \H_c$. From its definition, the matrix  $\Tr_{G^{(e)}}(\ket{\psi_{\alpha}}) \in \cal{M}_{N_c}(\bb{C})$ can be decomposed as
    \begin{equation}
        \Tr_{G^{(e)}}(\ket{\psi_{\alpha}}) = \bigotimes_{\substack{B \in \paa{1,\dots,D}\\\abs{B} \geq 2}}  \Tr_{G^{(e)}}(\vert\GHZ\rangle_{B, \alpha(B)} )\,,
    \end{equation}
    or more explicitly,
    \begin{equation}
        \Tr_{G^{(e)}}(\ket{\psi_{\alpha}}) = \left( \prod_{\substack{B \in \paa{1,\dots,D}\\\abs{B} \geq 2, \, c \not\in B}} \alpha(B)^{\kappa(G|_B) - k(G)} \right) \left( \bigotimes_{\substack{B \in \paa{1,\dots,D}\\\abs{B} \geq 2, \, c \in B}}  \Tr_{G^{(e)}}(\vert\GHZ\rangle_{B, \alpha(B)} ) \right)\,,
    \end{equation}
    where, the first term being a non-vanishing number, we can focus exclusively on the second term. The rank being multiplicative under tensor products, we can focus on one block $B_0 \in \paa{1,\dots,D}$ with $\abs{B_0} \geq 2$ and $c \in B_0$. Using the definition of trace-invariants, the matrix $\Tr_{G^{(e)}}(\vert\GHZ\rangle_{B_0 , \alpha(B_0)})$ is a product of contraction of copies of tensor components $\pa{\psi_{B_0 , \alpha(B_0)}}_{a_1 \dots a_D}$ and $\pa{\bar{\psi_{B_0 , \alpha(B_0)}}}_{a_1 \dots a_D}$, that one can define as follows: we can assume that $a_1, \ldots , a_D \in \{1, \ldots , \alpha(B_0)\}$ and let
    \begin{equation}\label{eq:tensor_comp_flip}
        \pa{\psi_{B_0, \alpha(B_0)}}_{a_1 \dots a_D} \eqdef \frac{1}{\sqrt{\alpha(B_0)}} \sum_{s=1}^{\alpha(B_0)} \prod_{c' \in B_0} \delta_{s a_{c'}}\prod_{c'' \notin B_0} \varphi^{(c'')}_{a_{c''}}\,,
    \end{equation}
    where $\{\varphi^{(c'')}\}_{c''\not\in B_0}$ are normalized vectors in $\mathbb{C}^{\alpha(B_0)}$. Now, note that the two vertices adjacent to $e$ in $G$ are necessarily in the same connected component of $G^{(e)}$, because the existence of a bridge in $G$ would be inconsistent with its $D$-colored structure. As a result, by explicit contraction of the tensor components from \eqref{eq:tensor_comp_flip}, we can check that 
    \begin{equation}\label{eq:proof_flip-1}
        \Tr_{G^{(e)}}(\vert\GHZ\rangle_{B_0 , \alpha(B_0)} ) = \alpha(B_0)^{\kappa(G\vert_{B_0})- k(G)-1}\, \mathrm{id}_{\mathbb{C}^{\alpha(B_0)}} 
    \end{equation}
    The rank of this map is $\alpha(B_0)$, thus proving that:
    \begin{equation}
       \mathrm{rank}\left( \Tr_{G^{(e)}}(\ket{\psi_{\alpha}})\right) = \prod_{\substack{B \in \paa{1,\dots,D}\\\abs{B} \geq 2, \, c \in B}} \alpha(B) \,.
    \end{equation}
    Moreover, using Eq.~\eqref{eq:flipOnKappa}, we observe that:
    \begin{align}
        \abs{\Tr_G(\ket{\psi_\alpha})}^2 &= \prod_{B \subset\{1, \ldots, D\},\, |B|\geq 2} \alpha(B)^{2\kappa(G\vert_B)- 2k(G)} \,, \\ 
        &= \left( \prod_{B \subset\{1, \ldots, D\},\, |B|\geq 2} \alpha(B)^{\kappa(G \flip_{(e,\bar e)} \bar G\vert_B) - k(G \flip_{(e,\bar e)} \bar G)} \right) \left(  \prod_{\substack{B \in \paa{1,\dots,D}\\\abs{B} \geq 2, \, c \in B}} \alpha(B) \right)\,, \\
        &=\Tr_{G \flip_{(e,\bar e)} \bar G}(\ket{\psi_{\alpha}}) \cdot \mathrm{rank}\left( \Tr_{G^{(e)}}(\ket{\psi_{\alpha}})\right) \,,
    \end{align}
    thus showing that $F_{G}^{[e]}(\ket{\psi_\alpha}) = 0$.
\end{proof}

Regarding the computation of the expectation value, we can first focus on the power law condition of Lem.~\ref{prop:mean_log} by the use of the next lemma.

\begin{lem} \label{lem:powerlawGflipG}
    Let $D \geq 2$, $G \in \cG_D$, and $e$ a marked edge in $G$. If $G $ follows a power law, then, so does $G \flip_{(e,\bar e)} \bar G$.
\end{lem}
\begin{proof}
    We omit the proof, which closely follows that of Lem.~\ref{lem:powerlawHoH}.
\end{proof}

If, in addition to the fact that $G$ follows a power law, we know that $G$ and $G \flip_{(e,\bar e)} \bar G$ factorize at large $N$, we could exchange the logarithm with the average for the first two terms of $F_{G}^{[e]}$. However, as highlighted by the following proposition, one can extract a condition uniquely based on the graph $G$ for which the large-$N$ factorization criterion is satisfied. 

\begin{prop} \label{prop:exchange_flip}
    Let $D \geq 3$, and let $G \in \cG_D$ be a $D$-colored graph with a marked edge $e$. If $\Delta(G) < \frac{D(D-1)}{8}$, then both $G$ and $G \flip_{(e,\bar e)} \bar G$ satisfy the large-$N$ factorization criterion.
\end{prop}

\begin{proof}
    The proof follows closely that of Prop.~\ref{prop:Haar_HoH}. Indeed, from a graph $\widehat G \in \cG_{D+1}(G)$ satisfying $\Delta_0(\widehat G) = \Delta(G)$, one can construct the graph $\widehat G \flip_{(e,\bar e)} \widehat{\bar G} \in \cG_{D+1}(G \flip_{(e,\bar e)} \bar G)$, which satisfies the relation:
    \begin{align}
        \Delta_0(\widehat G \flip_{(e,\bar e)} \widehat{\bar G}) &= \frac{D(D-1)}{4} k( G \flip_{(e,\bar e)} {\bar G}) + \frac 1 2 F( G \flip_{(e,\bar e)} {\bar G}) - \frac{D-1}{2} F_0(\widehat G \flip_{(e,\bar e)} \widehat{\bar G})  \\
        &= 2 \pa{\frac{D(D-1)}{4}  k( G) + \frac 1 2 F( G) - \frac{D-1}{2} F_0(\widehat G)} \\
        &= 2 \Delta(G) \,.
    \end{align}
    Then, this implies that
    \begin{equation}
        \Delta(G \flip_{(e,\bar e)}\bar G) \leq \Delta_0(\widehat G \flip_{(e,\bar e)} \widehat{\bar G}) = 2 \Delta(G) \,.
    \end{equation}
    Making the hypothesis that $\Delta(G) < \frac{D(D-1)}{8}$ -- from which one directly has that $G$ satisfies the large-$N$ factorization criterion (see Thm.~4.1 of Ref.~\cite{Factorization2026})  -- leads to
    \begin{equation}
        \Delta(G \flip_{(e,\bar e)}\bar G) < \frac{D(D-1)}{4} \,.
    \end{equation}
    Thus, $G \flip_{(e,\bar e)}\bar G$ also satisfies the large-$N$ factorization criterion.
\end{proof}

Under a different condition, the averages of the trace-invariants can be related to each other as presented in the following proposition.

\begin{prop} \label{prop:strong_cond_flip}
    Let $D \geq 3$, $G \in \cG_D$ a $D$-colored graph with a marked edge $e$, and let $(\ket{\varphi_N})_{N\in\bb{N}^*}$ be a sequence of $D$-partite Haar-random states of local dimension $N$. If $\Delta(G) < \frac{(D-1)(D-2)}{4}$, then:
    \begin{equation}
        N \cdot \mean{\Tr_{G \flip_{(e,\bar e)} \bar G}(\ket{\varphi_N})} \underset{N \to \infty}{=} \mean{\abs{\Tr_G(\ket{\varphi_N})}^2}\pa{1 + O\pa{\frac{1}{N}}}\,.
    \end{equation}
    In particular, we have $\Delta(G \flip_{(e,\bar e)} \bar G) = 2 \Delta(G)$.
\end{prop}

\begin{proof}
    Prop.~\ref{prop:strong_cond_flip} is a special case of item 4 of Thm.~6.4 in Ref.~\cite{Carrozza:2026qcf} with $G_1 = G$ and $G_2 = \bar G$. The proof of Thm.~6.4 can be found in App.~D of that reference.
\end{proof}

\begin{rem} \label{rem:regimes}
    We emphasize that Prop.~\ref{prop:exchange_flip} and Prop.~\ref{prop:strong_cond_flip} are not always simultaneously valid. \\
    Indeed, for $D=3$ and $D=4$, we observe three distinct regimes because Prop.~\ref{prop:strong_cond_flip} implies Prop.~\ref{prop:exchange_flip} in these cases (see Fig.~\ref{fig:Summary}).
    \begin{enumerate}
        \item If $0 \leq \Delta(G) < \frac{(D-1)(D-2)}{4}$, then both Prop.~\ref{prop:exchange_flip} and Prop.~\ref{prop:strong_cond_flip} hold. 
        \item If $\frac{(D-1)(D-2)}{4} \leq \Delta(G) < \frac{D(D-1)}{8}$, then only Prop.~\ref{prop:exchange_flip} holds. 
        \item If $ \frac{D(D-1)}{8} \leq \Delta(G)$, then neither Prop.~\ref{prop:exchange_flip} nor Prop.~\ref{prop:strong_cond_flip} holds. 
    \end{enumerate}
    Interestingly, for $D \geq 5$, the direction of the implication is reversed, as Prop.~\ref{prop:exchange_flip} now implies Prop.~\ref{prop:strong_cond_flip}. Consequently, we observe the following regimes (see Fig.~\ref{fig:Summary}).
    \begin{enumerate}
        \item If $0 \leq \Delta(G) < \frac{D(D-1)}{8}$, then both Prop.~\ref{prop:exchange_flip} and Prop.~\ref{prop:strong_cond_flip} hold. 
        \item If $\frac{D(D-1)}{8} \leq \Delta(G) < \frac{(D-1)(D-2)}{4}$, then only Prop.~\ref{prop:strong_cond_flip} holds. 
        \item If $\frac{(D-1)(D-2)}{4} \leq \Delta(G)$, then neither Prop.~\ref{prop:exchange_flip} nor Prop.~\ref{prop:strong_cond_flip} holds. 
    \end{enumerate}
    
    \begin{figure}
        \centering
        \includegraphics[width = .9 \textwidth]{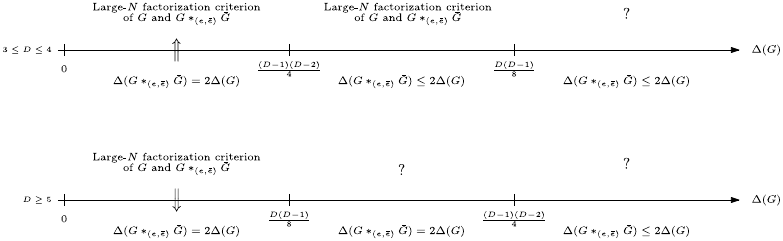}
        \caption{Top: Illustration of the three regimes presented in Rem.~\ref{rem:regimes} that appears for $3 \leq D \leq 4$. Bottom: The three regimes presented in the same remark and appearing for $D \geq 5$.}
        \label{fig:Summary}
    \end{figure}
\end{rem}

To determine the large-$N$ expectation value of $F_G^{[e]}(\ket{\varphi_N})$, we need to get a handle on the expectation value of $\mathrm{rank}\left( \tr_{G^{(e)}}(\ket{\varphi_N})\right)$. Since this quantity defines a \emph{non-polynomial} $\LU$-invariant, evaluating its expectation value might in principle require methods that go significantly beyond the framework developed in \cite{Carrozza:2026qcf} and \cite{Factorization2026}. Luckily, this potential issue is easily circumvented: we can indeed argue that the rank of the matrix $\tr_{G^{(e)}}(\ket{\varphi_N})$ is almost everywhere maximal with respect to any ``reasonable'' measure (including the Haar measure), and thus the looked-for expectation value is easily evaluated. This observation is made precise in the next lemma.
\begin{lem}\label{lem:larg_N_rank}
    Let $D\geq 2$, $G \in \cG_D$ with marked edge $e$, and let $\ket{\varphi_N}$ be a $D$-partite Haar-random state of local dimension $N$. The random variable $\mathrm{rank}\left( \tr_{G^{(e)}}(\ket{\varphi_N})\right)$ takes value $N$ with probability $1$. In particular, we have
    \begin{equation}
        \langle \ln \left( \mathrm{rank}\left( \tr_{G^{(e)}}(\ket{\varphi_N})\right) \right) \rangle = \ln(N)\,.
    \end{equation}

\end{lem}
\begin{proof}
    Let us define $\ket{\varphi_N}$ as a random variable taking value in the projective Hilbert space\footnote{That is, $\mathcal{P}(\H_1 \otimes \cdots \otimes \H_D)$ is the unit sphere of $\H_1 \otimes \cdots \otimes \H_D$ modulo phases.} $\mathcal{P}(\H_1 \otimes \cdots \otimes \H_D)$, with $\H_c \eqdef \mathbb{C}^N$ for any $c\in \{1, \ldots , D\}$. Let us also denote by $\{\ket{i}\}_{1\leq i \leq N}$ the canonical basis of $\mathbb{C}^N$. Consider the map $f: \mathcal{P}(\H_1 \otimes \cdots \otimes \H_D) \to M_N (\mathbb{C}), \ket\psi \mapsto \left( \tr_{G^{(e)}}(\ket\psi )\right)$, which sends a pure state to a matrix $f(\ket\psi)$. The matrix elements of this map are homogeneous polynomials in $\ket\psi$ (or, equivalently, in its tensor coordinates $\{\psi_{a_1 \cdots a_D}\}$ and $\{\bar\psi_{a_1 \cdots a_D}\}$). Accordingly, the map $\det\circ f :\mathcal{P}(\H_1 \otimes \cdots \otimes \H_D)\to \mathbb{C}$ is a homogeneous polynomial in $\ket\psi$, and it is moreover non-zero. To justify this last claim, notice first that  the vertices adjacent to $e$ in $G$ are in the same connected component of $G^{(e)}$ (as there can be no bridge in a $D$-colored graph). As a result, we have $f (\ket{\GHZ}_{\{1, \ldots , D\}, N})= N^{\kappa(G)-k(G)-1} \sum_{i=1}^N \ket{i}\bra{i}$ (remember Eq.~\eqref{eq:proof_flip-1}), which is clearly invertible in $M_N (\mathbb{C})$, thus implying $\det\circ f(\ket{\GHZ}_{\{1, \ldots , D\}, N})\neq 0$. It follows that the zero-locus of $\det \circ f$, namely
    \begin{equation}
        V= \{\ket\psi \in \mathcal{P}(\H_1 \otimes \cdots \otimes \H_D) \, \vert \, \det \circ f (\ket\psi) = 0 \} =  \{\ket\psi \in \mathcal{P}(\H_1 \otimes \cdots \otimes \H_D) \, \vert \, \mathrm{rank} \left( f (\ket\psi)\right) < N \} \,, 
    \end{equation}
is a measure zero set with respect to the Haar measure on $\mathcal{P}(\H_1 \otimes \cdots \otimes \H_D)$ (it is an example of Zariski-closed set in $\mathcal{P}(\H_1 \otimes \cdots \otimes \H_D)$). In other words, the complement $\mathcal{P}(\H_1 \otimes \cdots \otimes \H_D) \setminus V$ has full measure, which means that $\mathrm{rank}\left(\tr_{G^{(e)}}(\ket{\varphi_N} )\right)=N$ with probability $1$. Taking a logarithm, it follows that $\ln\left(\mathrm{rank}\left(\tr_{G^{(e)}}(\ket{\varphi_N} )\right)\right)=\ln(N)$ with probability $1$. 
\end{proof}

The following two corollaries summarize the last propositions and lemmas. We make, in particular, a case distinction between $3 \leq D \leq 4$ and $D \geq 5$.

\begin{cor}\label{cor:Haar_flip-1}
    Let $D \in\{3,4\}$, $G \in \cG_D$ a $D$-colored graph with a marked edge $e$, and let $(\ket{\varphi_N})_{N\in\bb{N}^*}$ be a sequence of $D$-partite Haar-random states of local dimension $N$.
    \begin{enumerate}
        \item If $G$ follows a power-law and $\frac{(D-1)(D-2)}{4} < \Delta(G) < \frac{D(D-1)}{8}$, then:
        \begin{equation}
            \mean{F_G^{[e]}(\ket{\varphi_N})} \underset{N \to \infty}{=} \frac{2}{D-1} \pa{2\Delta(G) - \Delta(G \flip_{(e,\bar e)} \bar G)} \ln N + \ln \frac{\mu_{G \flip_{(e,\bar e)} \bar G)}}{\mu_G^2} + o(1)\,.
        \end{equation}
        \item If $G$ follows a power-law and $\Delta(G) < \frac{(D-1)(D-2)}{4}$, then:
        \begin{equation}
            \mean{F_G^{[e]}(\ket{\varphi_N})} \underset{N \to \infty}{=} o(1)\,.
        \end{equation}
    \end{enumerate}
\end{cor}

\begin{cor}\label{cor:Haar_flip-2}
    Let $D \geq 5$, $G \in \cG_D$ a $D$-colored graph with a marked edge $e$, and let $(\ket{\varphi_N})_{N\in\bb{N}^*}$ be a sequence of $D$-partite Haar-random states of local dimension $N$. If $G$ follows a power-law and $\Delta(G) < \frac{D(D-1)}{8}$, then:
    \begin{equation}
        \mean{F_G^{[e]}(\ket{\varphi_N})} \underset{N \to \infty}{=} o(1)\,.
    \end{equation}
\end{cor}

\subsection{Families of tripartite entanglement spreads whose minimizers obey useful graphical equations}\label{subse:tripartite_spreads}

In this section, we fix $D=3$. By convention (and by slight abuse of notation), for any $c \in \{ 1,2,3\}$ we will denote by $c'$ and $c''$ the unique elements of $\{1,2,3\}$ such that $c' \equiv c+1 \mod 3$ and $c'' \equiv c+2 \mod 3$.

\subsubsection{First family: minimizers obey separating \texorpdfstring{$1$}{1}-dipole contraction moves}

For any $c \in \{1,2,3\}$ and any $n \in \mathbb{N}^*$, we define
\begin{itemize}
    \item the realignment moment $\RM_{2n}^{(c)}$, defined as
    \begin{equation}
\RM_{2n}^{(c)} = \vcenter{\hbox{\includegraphics[scale=.5]{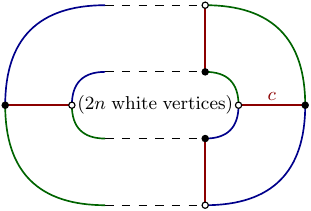}}}\,,
    \end{equation}
obeying  
\begin{equation}\label{eq:comb_RM}
    \begin{aligned}
        &k(\RM_{2n}^{(c)})=2n\,, &&\kappa\left( \RM_{2n}^{(c)}\right) = 1\,, \\
        &\kappa\left( \RM_{2n}^{(c)}\vert_{c,c'}\right) = \kappa\left( \RM_{2n}^{(c)}\vert_{c,c''}\right) = n \,, \qquad \quad &&\kappa\left( \RM_{2n}^{(c)}\vert_{c',c''}\right)=2\,;
    \end{aligned}
\end{equation}
    \item the realignment moment decorated by one melonic insertion on one of the edges of color $c$, noted $\widetilde{\RM}_{2n+1}^{(c)}$, defined as \begin{equation}
\widetilde{\RM}_{2n+1}^{(c)} = \vcenter{\hbox{\includegraphics[scale=.5]{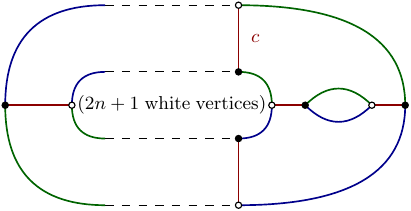}}}\,, 
    \end{equation}
obeying
\begin{equation}\label{eq:comb_RMtilde}
    \begin{aligned}
        &k(\widetilde{\RM}_{2n+1}^{(c)})=2n+1\,, &&\kappa\left( \widetilde{\RM}_{2n+1}^{(c)} \right) = 1\,, \\
        &\kappa\left( \widetilde{\RM}_{2n+1}^{(c)}\vert_{c,c'}\right) = \kappa\left( \widetilde{\RM}_{2n+1}^{(c)}\vert_{c,c''}\right) = n \,, \qquad \quad &&\kappa\left( \widetilde{\RM}_{2n+1}^{(c)}\vert_{c',c''}\right)=3\,;
    \end{aligned}
\end{equation}    
    \item the melonic diagram\footnote{$\Mel_{2n-1}^{(c)}$ can be obtained from $\RM_{2n}^{(c)}$ by removing one edge of color $c$ together with its adjacent vertices, and reconnecting the resulting open half-edges in the only way that respects the color structure.}
            \begin{equation}
\Mel_{2n-1}^{(c)} = \vcenter{\hbox{\includegraphics[scale=.5]{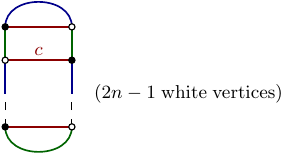}}}\,,
    \end{equation}
obeying
\begin{equation}\label{eq:comb_Mel}
    \begin{aligned}
        &k(\Mel_{2n-1}^{(c)})=2n-1\,,  &&\kappa\left( \Mel_{2n-1}^{(c)} \right) = 1\,, \\
        &\kappa\left( \Mel_{2n-1}^{(c)}\vert_{c,c'}\right) = \kappa\left( \Mel_{2n-1}^{(c)}\vert_{c,c''}\right) = n \,, \qquad\quad  &&\kappa\left( \Mel_{2n-1}^{(c)}\vert_{c',c''}\right)=1\,.
    \end{aligned}
\end{equation}
\end{itemize}

\begin{lem}\label{lem:SpreadOne}
    For any $n \in \mathbb{N}^*$ and $c\in \{1,2,3\}$, we have:
    \begin{equation}
        \RM_{2n}^{(c)}\sqcup \overline{\RM_{2n}^{(c)}} \dom \Mel_{2n-1}^{(c)}\sqcup \widetilde{\RM}_{2n+1}^{(c)}
    \end{equation}
Furthermore, the invariants $\tr_{\RM_{2n}^{(c)}}$, $\tr_{\Mel_{2n-1}^{(c)}}$ and $\tr_{\widetilde{\RM}_{2n+1}^{(c)}}$ are positive and never vanish. 
\end{lem}
\begin{proof}
Given a $3$-partite pure state $\ket\psi = \vcenter{\hbox{\includegraphics[scale=.5]{images/psi.pdf}}}$, we have the following Cauchy-Schwarz inequality:
    \begin{align}
        \abs{\tr_{\RM_{6}^{(c)}}(\ket\psi)}^2= \abs{ \vcenter{\hbox{\includegraphics[scale=.5]{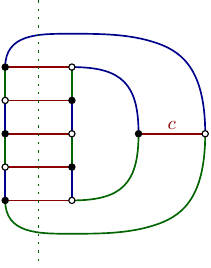}}} }^2 &\leq \vcenter{\hbox{\includegraphics[scale=.5]{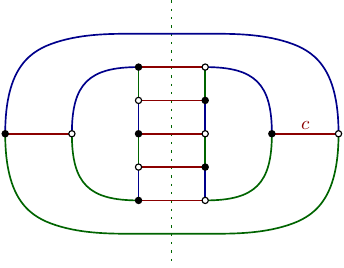}}} \times \vcenter{\hbox{\includegraphics[scale=.5]{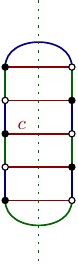}}} \label{eq:CS_move1_n3} \\&= \vcenter{\hbox{\includegraphics[scale=.5]{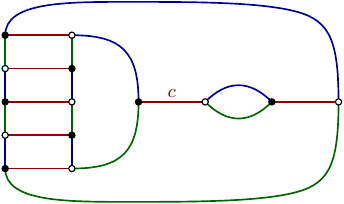}}} \times \vcenter{\hbox{\includegraphics[scale=.5]{images/CS_move1_c.pdf}}} \\   
        &= \tr_{\widetilde{\RM}_{7}^{(c)}}(\ket\psi) \times \tr_{\Mel_{5}^{(c)}}(\ket\psi) \,,
    \end{align}
which yields the looked-for inequality for $n=3$. This graphical proof generalizes straightforwardly to an arbitrary $n \in \mathbb{N}^*$, leading to:
\begin{equation}\label{eq:CS_RM_RMtilde_M}
        \abs{\tr_{\RM_{2n}^{(c)}}(\ket\psi)}^2\leq \tr_{\widetilde{\RM}_{2n+1}^{(c)}}(\ket\psi) \times \tr_{\Mel_{2n-1}^{(c)}}(\ket\psi)\,. 
\end{equation}
Furthermore, it was proven in \cite{Carrozza:2026qcf} (see Prop.~6.16) that the invariant $\tr_{\RM_{2n}^{(c)}}$ is positive and never vanishes. According to the Cauchy-Schwarz inequality we just established, $\tr_{\Mel_{2n-1}^{(c)}}$ and $\tr_{\widetilde{\RM}_{2n+1}^{(c)}}$ cannot vanish either. Finally, as we have seen while establishing the Cauchy-Schwarz inequality, we can interpret  $\tr_{\Mel_{2n-1}^{(c)}}(\ket\psi)$ (resp.~$\tr_{\widetilde{\RM}_{2n+1}^{(c)}}(\ket\psi)$) as the norm squared of some vector; it is therefore positive. 
\end{proof}
\begin{defi}\label{def:Pn}
    For any $n \in \mathbb{N}^*$ and $c\in \{1,2,3\}$, we define
    \begin{equation}
    \SpreadOne_{n}^{(c)} \eqdef R_{\RM_{2n}^{(c)}\sqcup \overline{\RM_{2n}^{(c)}}, \Mel_{2n-1}^{(c)}\sqcup \widetilde{\RM}_{2n+1}^{(c)}}\,,   
    \end{equation}
which, according to Lem.~\ref{lem:SpreadOne}, is a well-defined entanglement spread. Lem.~\ref{lem:SpreadOne} also implies that: for any finite-dimensional $D$-partite pure state $\ket\psi$,
\begin{equation}
    \SpreadOne_{n}^{(c)} \left( \ket\psi \right) = \ln \left( \tr_{\Mel_{2n-1}^{(c)}} (\ket\psi)  \right) + \ln \left( \tr_{\widetilde{\RM}_{2n+1}^{(c)}} (\ket\psi)  \right) - 2 \ln \left( \tr_{\RM_{2n}^{(c)}} (\ket\psi)  \right) \,.  
    \end{equation}
\end{defi}
\begin{rem}\label{rem:P1}
    Specializing to $n=1$, we observe that: for any $c \in \{1,2,3\}$ and any tripartite pure state $\ket\psi \in \H_1 \otimes \H_2 \otimes \H_3$, 
    \begin{equation}
        \SpreadOne_{1}^{(c)}\left( \ket\psi \right) = \ln\left(\tr_{C_3(\{c\})}\left( \ket\psi\right)\right) - 2 \ln\left(\tr_{C_2(\{c\})}\left( \ket\psi\right)\right)= 2 \mathsf{D}_{2,3} \left( \ket\psi \right)\,,
    \end{equation}
    where $\ket\psi$ in the last equality is understood as a bipartite state on $\H_c \otimes \left( \H_{c'} \otimes \H_{c''}\right)$.
\end{rem}    
\begin{prop}\label{prop:SpreadOne_alpha_states}
    For any $c \in \{1,2,3\}$, any $n \in \mathbb{N}^*$ and any HT state $\ket{\psi_\alpha}$ (with weight function $\alpha$), we have:
    \begin{equation}
        \SpreadOne_{n}^{(c)} \left( \ket{\psi_\alpha} \right) = 0\,.
    \end{equation}
\end{prop}    
\begin{proof}
    Using the combinatorial relations from \eqref{eq:comb_RM}, \eqref{eq:comb_RMtilde} and \eqref{eq:comb_Mel}, we infer that
    \begin{align}
\tr_{\RM_{2n}^{(c)}}\left(\ket{\psi_\alpha}\right) &= \alpha\left( \{1,2,3\}\right)^{1-2n} \times [ \alpha\left( \{c,c'\}\right) \alpha\left( \{c,c''\}\right)]^{-n} \times \alpha\left( \{c',c''\}\right)^{2-2n}
\,,\\
\tr_{\widetilde{\RM}_{2n+1}^{(c)}}\left(\ket{\psi_\alpha}\right) &= \alpha\left( \{1,2,3\}\right)^{-2n} \times [ \alpha\left( \{c,c'\}\right) \alpha\left( \{c,c''\}\right)]^{-n-1} \times \alpha\left( \{c',c''\}\right)^{2-2n} \,,\\
\tr_{\Mel_{2n-1}^{(c)}}\left(\ket{\psi_\alpha}\right) &= \alpha\left( \{1,2,3\}\right)^{2-2n} \times [ \alpha\left( \{c,c'\}\right) \alpha\left( \{c,c''\}\right)]^{-n+1} \times \alpha\left( \{c',c''\}\right)^{2-2n} \,,
    \end{align}
from which it is clear that 
\begin{equation}
\tr_{\RM_{2n}^{(c)}}\left(\ket{\psi_\alpha}\right)^2 =  \tr_{\widetilde{\RM}_{2n+1}^{(c)}}\left(\ket{\psi_\alpha}\right) \tr_{\Mel_{2n-1}^{(c)}}\left(\ket{\psi_\alpha}\right)\,,    
\end{equation}
and thus $\SpreadOne_{n}^{(c)} \left( \ket{\psi_\alpha} \right) = 0$.
\end{proof}
     
\begin{prop}\label{prop:graph_eq_move1}
 Let $n\in \mathbb{N}^*$, $c\in \{1,2,3\}$ and $\ket\psi$ a finite-dimensional $3$-partite pure state. We have $\SpreadOne_{n}^{(c)} \left( \ket\psi \right) = 0$ if and only if there exists $\rho_{c}^{(n)} \in \mathbb{C}\setminus \{0\}$ such that $\ket\psi = \vcenter{\hbox{\includegraphics[scale=.5]{images/psi.pdf}}}$ obeys the following graphical equation:
 \begin{equation}\label{eq:graph_eq_move1}
 \vcenter{\hbox{\includegraphics[scale=.5]{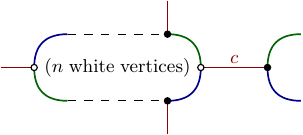}}}= \rho_{c}^{(n)} \vcenter{\hbox{\includegraphics[scale=.5]{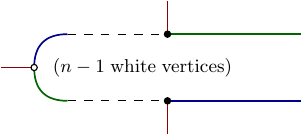}}}
 \end{equation}
In particular, for $n=1$, $2$ and $3$, we have, respectively:
\begin{align}
\vcenter{\hbox{\includegraphics[scale=.5]{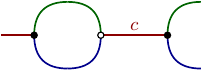}}} &= \rho_{c}^{(1)} \vcenter{\hbox{\includegraphics[scale=.5]{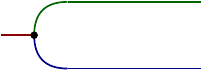}}} \,,\\
    \vcenter{\hbox{\includegraphics[scale=.5]{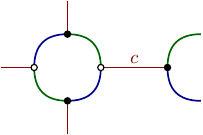}}} &= \rho_{c}^{(2)} \vcenter{\hbox{\includegraphics[scale=.5]{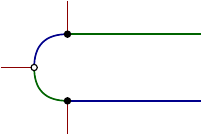}}} \,,\\
    \vcenter{\hbox{\includegraphics[scale=.5]{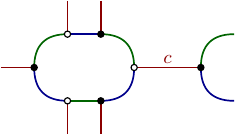}}} &= \rho_{c}^{(3)} \vcenter{\hbox{\includegraphics[scale=.5]{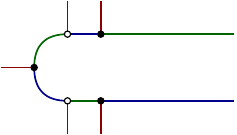}}}\,. \label{eq:equality_move1_3}\\
\end{align}
\end{prop}
\begin{proof}
Let us fix $n=3$ for concreteness. Given a $3$-partite pure state $\ket\psi = \vcenter{\hbox{\includegraphics[scale=.5]{images/psi.pdf}}}$, we have equality in \eqref{eq:CS_move1_n3} if and only if the vectors
    \begin{equation}
    \vcenter{\hbox{\includegraphics[scale=.5]{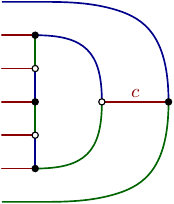}}} \quad    \mathrm{and} \quad     \vcenter{\hbox{\includegraphics[scale=.5]{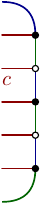}}}
    \end{equation}
are linearly-dependent. The inner product between these two vectors is $\tr_{\RM_{6}^{(c)}}(\ket\psi)$ and this quantity cannot vanish (according to Prop.~6.16 of \cite{Carrozza:2026qcf}), hence neither of these two vectors is the zero vector. We can thus find a complex number $\rho_c^{(3)} \neq 0$ such that:
\begin{equation}
        \vcenter{\hbox{\includegraphics[scale=.5]{images/CS_move1_vector1.pdf}}}= \rho_c^{(3)}\times  \vcenter{\hbox{\includegraphics[scale=.5]{images/CS_move1_vector2.pdf}}}\,,
\end{equation}
which we recognize as \eqref{eq:equality_move1_3}. The case of an arbitrary $n \in \mathbb{N}^*$ can be dealt with in a similar way. 
\end{proof}

In the existing literature on $D$-colored graphs and tensor models \cite{Crystallisation,Ferri1986,Gurau:2011xp, gurau_random_2017}, a single edge of color $c$ connecting two vertices is known as a \emph{$1$-dipole} of color $c$. When those two vertices belong to distinct faces of color $(c',c'')$, as in the left-hand-side of Eq.~\eqref{eq:graph_eq_move1}, we will more precisely refer to such a pattern as a \emph{separating $1$-dipole} (of color $c$), following the nomenclature of \eg Ref.~\cite{LINS1985261}.\footnote{In the literature on tensor models (see \eg Refs.~\cite{Gurau:2010ba,Carrozza:2012uv}), such a $1$-dipole is sometimes referred to as a \emph{non-degenerate} $1$-dipole. However, since the dichotomy between  degenerate and non-degenerate $1$-dipoles sometimes refers to something else entirely in the early literature on $D$-colored graphs \cite{Crystallisation,Ferri1986}, we will rely here on the nomenclature non-separating/separating from \cite{LINS1985261}.} Furthermore, the \emph{contraction} of such a $1$-dipole is defined as the graph operation that consists in: 1) removing the dipole edge and its adjacent vertices; and 2) reconnecting the open half-edges resulting from the previous operation in the unique way compatible with the colored structure of the graph. This leads to the structure shown on the right-hand-side of Eq.~\eqref{eq:graph_eq_move1}. Hence, the minimizers of the family of entanglement spreads $\{\SpreadOne_{n}^{(c)}\}$ obey equations that enable the recursive application of separating $1$-dipole contractions. In the interpretation of $3$-colored graphs as triangulated orientable surfaces, such separating $1$-dipole contractions preserve the topology of the underlying surface (\ie they do not affect its genus). Hence, there is a limit to how much we can simplify the structure of the initial graph by means of such moves. This motivates the introduction of a second family of spreads, which will in turn enable us to implement graphical moves that decrease the value of the genus. Before turning to this question in Sec.~\ref{sec:second_family}, we observe that it is straightforward to determine the large-$N$ expectation values of the spreads $\{\SpreadOne_{n}^{(c)}\}$ in the Haar-random state of local dimension $N$.
\begin{prop}\label{prop:SpreadOne_Haar}
    Let $(\ket{\varphi_N})_{N \in \mathbb{N}^*}$ denote a sequence of $3$-partite Haar-random states of local dimension $N$. For any $c \in \{1,2,3\}$, we have:
    \begin{equation}
        \langle \SpreadOne_{1}^{(c)}(\ket{\varphi_N}) \rangle \underset{N\to +\infty}{=} o(1)\,,
    \end{equation}
    and
    \begin{equation}
\forall n \geq 2\,,     \qquad \langle \SpreadOne_{n}^{(c)}(\ket{\varphi_N}) \rangle \underset{N\to +\infty}{=} \ln(N) + o(1)\,.
    \end{equation}
\end{prop}
\begin{proof}
    It was proven in our earlier work \cite{Carrozza:2026qcf} that $\RM_{2n}^{(c)}$ obeys a power law for every $n$ and $c$; by Lem.~\ref{lem:power-law} and the Cauchy-Shwarz inequality \eqref{eq:CS_RM_RMtilde_M}, it follows that $\widetilde{\RM}_{2n+1}^{(c)}$ and $\Mel_{2n-1}^{(c)}$ also obey power laws. Moreover, these three invariants are known to obey the large-$N$ factorization criterion (see Def.~\ref{def:fact_power} and Ref.~\cite{Factorization2026}). Henceforth, Prop.~6.18 of \cite{Carrozza:2026qcf} enables us to commute the logarithm with the expectation value of any such invariant at leading and subleading orders in the asymptotic large-$N$ regime. Finally, it can be shown that \cite{Bonzom:2012hw,Aubrun:2012ydp,lionni_colored_2018,Carrozza:2026qcf}, for any $n\geq2$ and any $c \in \{1,2,3\}$,
    \begin{align}
        \langle \tr_{   \Mel_{2n-1}^{(c)}}(\ket{\varphi_N})\rangle   &\underset{N\to+\infty}{=} N^{2-2n} \left( 1+ O(1/N)\right) \,,\\
        \langle \tr_{   \RM_{2n}^{(c)}}(\ket{\varphi_N})\rangle   &\underset{N\to+\infty}{=} N^{-2n} \left( 1+ O(1/N)\right) \,,\\
        \langle \tr_{   \widetilde{\RM}_{2n+1}^{(c)}}(\ket{\varphi_N})\rangle   &\underset{N\to+\infty}{=} N^{-1-2n} \left( 1+ O(1/N)\right) \,,
    \end{align}
from which we can conclude that
\begin{align}
        \langle \SpreadOne_n^{(c)}(\ket{\varphi_N})\rangle  &\underset{N\to+\infty}{=} \ln\left(\langle \tr_{   \Mel_{2n-1}^{(c)}}(\ket{\varphi_N})\rangle \right) + \ln\left(\langle \tr_{   \widetilde{\RM}_{2n+1}^{(c)}}(\ket{\varphi_N})\rangle \right)  - 2 \ln\left(\langle \tr_{   \RM_{2n}^{(c)}}(\ket{\varphi_N})\rangle \right)+ o(1)\\
        &\underset{N\to+\infty}{=} \ln(N) + o(1) \nonumber
\end{align}        
The remaining case $n=1$ is special because $\RM_2^{(c)}$ and $\tilde{\RM}_3^{(c)}$ are both cyclic and melonic, while $\Mel_1^{(c)}$ reduces to the trivial two-vertex graph (see Rem.~\ref{rem:P1}). We then find
\begin{align}
        \langle \tr_{   \RM_{2}^{(c)}}(\ket{\varphi_N})\rangle   &\underset{N\to+\infty}{=} N^{-1} \left( 1+ O(1/N)\right) \,,\\
        \langle \tr_{   \widetilde{\RM}_{3}^{(c)}}(\ket{\varphi_N})\rangle   &\underset{N\to+\infty}{=} N^{-2} \left( 1+ O(1/N)\right) \,,
    \end{align}
and therefore
\begin{align}
        \langle \SpreadOne_1^{(c)}(\ket{\varphi_N})\rangle  &\underset{N\to+\infty}{=}  \ln\left(\langle \tr_{   \widetilde{\RM}_{3}^{(c)}}(\ket{\varphi_N})\rangle \right)  - 2 \ln\left(\langle \tr_{   \RM_{2}^{(c)}}(\ket{\varphi_N})\rangle \right)+ o(1)\\ &\underset{N\to+\infty}{=} o(1) \nonumber
\end{align} 
\end{proof}

\subsubsection{Second family: minimizers obey equations enabling the elimination of non-separating \texorpdfstring{$1$}{1}-dipoles}\label{sec:second_family}

 For any $c\in \{1,2,3\}$ and any $m,n,p,q \in \mathbb{N}$, we define the $3$-colored graphs
\begin{equation}
    G_{m,n,p,q}^{(c)}= \vcenter{\hbox{\includegraphics[scale=.6]{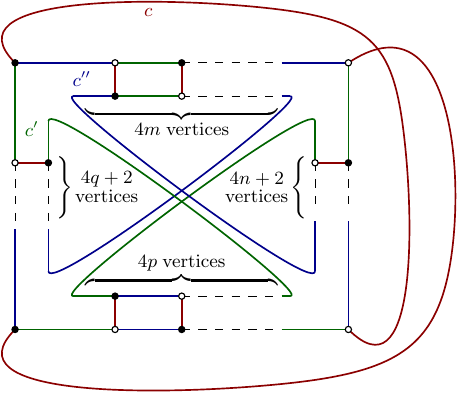}}}
\end{equation}
and
\begin{equation}
    H_{m,n,p,q}^{(c)}= \vcenter{\hbox{\includegraphics[scale=.6]{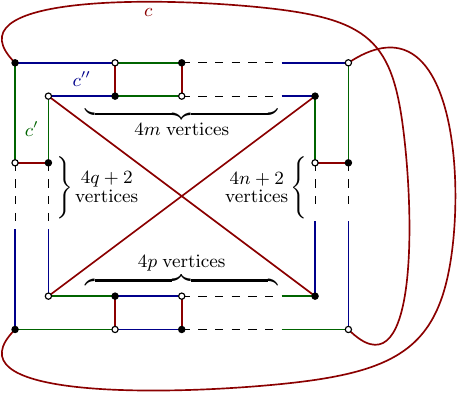}}} \,.
\end{equation}
For example, one has:
\begin{equation}
    \vcenter{\hbox{\includegraphics[scale=.6]{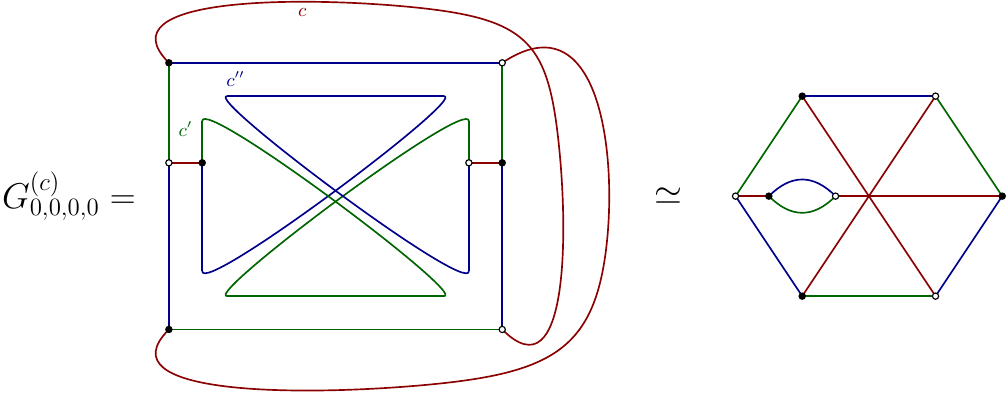}}}
\end{equation}
and
\begin{equation}
        \vcenter{\hbox{\includegraphics[scale=.6]{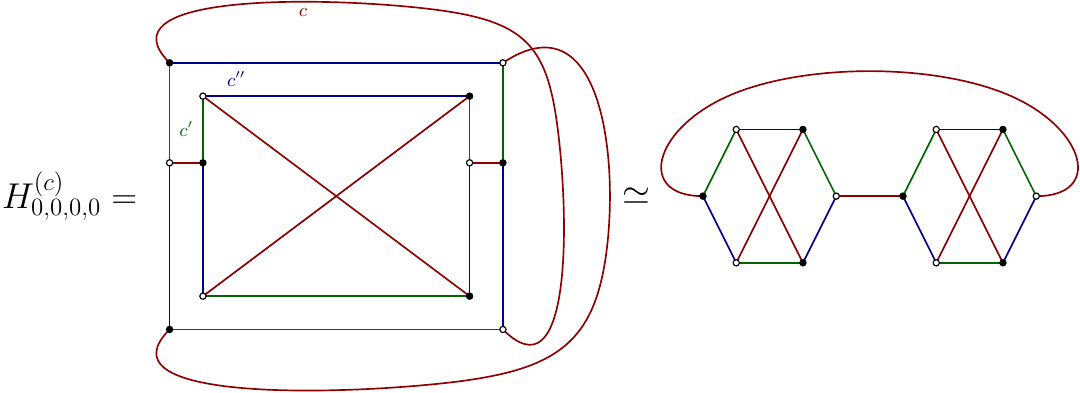}}}
\end{equation}
We also note that 
\begin{align}\label{eq:comb_H-G_mnpq}
    k(G_{m,n,p,q}^{(c)})&= 2(m+n+p+q)+4\,, \qquad k(H_{m,n,p,q}^{(c)})= 2(m+n+p+q)+6\,, \nonumber \\
    \kappa\left( G_{m,n,p,q}^{(c)} \right) &=  \kappa\left( H_{m,n,p,q}^{(c)} \right) = 1 \,, \nonumber \\
    \kappa\left( G_{m,n,p,q}^{(c)}\vert_{c,c'} \right) &= \kappa\left( G_{m,n,p,q}^{(c)}\vert_{c,c''} \right) = \kappa\left( H_{m,n,p,q}^{(c)}\vert_{c,c'} \right) = \kappa\left( H_{m,n,p,q}^{(c)}\vert_{c,c''} \right) = m+n+p+q+1\,, \nonumber \\
    \kappa\left( G_{m,n,p,q}^{(c)}\vert_{c',c''} \right) &= \kappa\left( H_{m,n,p,q}^{(c)}\vert_{c',c''} \right) = 2\,.
\end{align}

\begin{lem}\label{lem:SpreadTwo}
    For any $m,n,p,q \in \mathbb{N}$ and $c\in \{1,2,3\}$, we have:
    \begin{equation}
        G_{m,n,p,q}^{(c)}\sqcup \overline{G_{m,n,p,q}^{(c)}} \dom \RM_{2(m+n+p+q)}^{(c)}\sqcup H_{m,n,p,q}^{(c)}\,.
    \end{equation}
\end{lem}
\begin{proof}
    For any finite-dimensional pure state $\ket\psi = \vcenter{\hbox{\includegraphics[scale=.5]{images/psi.pdf}}}$, we can invoke the Cauchy-Schwarz inequality to write:
\begin{align}
        \abs{\tr_{G_{m,n,p,q}^{(c)}}(\ket\psi)}^2&= \abs{\, \vcenter{\hbox{\includegraphics[scale=.5]{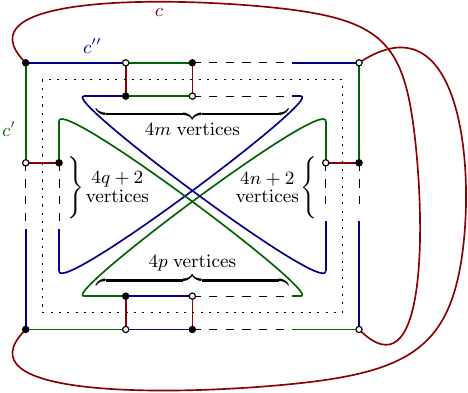}}} \,}^2 \label{eq:CS_move2} \\&\leq \vcenter{\hbox{\includegraphics[scale=.5]{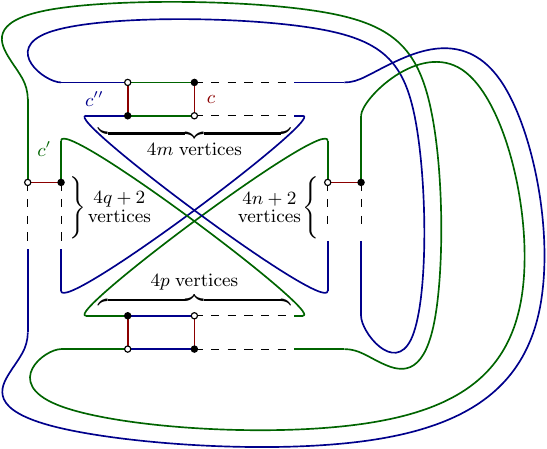}}} \times \vcenter{\hbox{\includegraphics[scale=.5]{images/Hmnpq.pdf}}} \label{eq:CS_move1_n3} \nonumber\\   
        &= \tr_{\RM_{2(m+n+p+q+1)}^{(c)}}(\ket\psi) \times \tr_{H_{m,n,p,q}^{(c)}}(\ket\psi) \nonumber \,,
    \end{align} 
where the dotted square in the first line indicates how we interpreted this tensor contraction as an inner product and, in going from the second to the third line, we recognized that the first graph of the second line is in fact isomorphic to $\RM_{2(m+n+p+q+1)}$ (to see this, one can first verify that this graph has exactly two faces of color $(c',c'')$, of equal lengths, and that they ``run parallel'' to one-another).    
\end{proof}

\begin{defi}\label{def:Qmnpq}
    For any $m,n,p,q \in \mathbb{N}$ and $c\in \{1,2,3\}$, we define
    \begin{equation}
    \SpreadTwo_{m,n,p,q}^{(c)} \eqdef R_{G_{m,n,p,q}^{(c)}\sqcup \overline{G_{m,n,p,q}^{(c)}}, \RM_{2(m+n+p+q+1)}^{(c)}\sqcup H_{m,n,p,q}^{(c)}}\,, \end{equation}
which, according to Lem.~\ref{lem:SpreadTwo}, is a well-defined entanglement spread. More concretely, for any finite-dimensional $D$-partite pure state $\ket\psi$, we have
\begin{equation}
    \SpreadTwo_{m,n,p,q}^{(c)} \left( \ket\psi \right) = \begin{cases} \ln \abs{\tr_{H_{m,n,p,q}^{(c)}} (\ket\psi)  } + \ln \left( \tr_{\RM_{2(m+n+p+q+1)}^{(c)}} (\ket\psi)  \right)
    - 2 \ln \abs{\tr_{G_{m,n,p,q}^{(c)}} (\ket\psi) } \; &\mathrm{if}\;\tr_{G_{m,n,p,q}^{(c)}} (\ket\psi)\neq 0\\
    + \infty \; &\mathrm{otherwise}
    \end{cases}
    \end{equation}
\end{defi}

\begin{prop}\label{prop:SpreadTwo_alpha_states}
    For any $c \in \{1,2,3\}$, any $m,n,p,q \in \mathbb{N}$ such that $m+n+p+q\neq 0$, and any HT state $\ket{\psi_\alpha}$ (with weight function $\alpha$), we have:
    \begin{equation}
        \SpreadTwo_{m,n,p,q}^{(c)} \left( \ket{\psi_\alpha} \right) = 0\,.
    \end{equation}
\end{prop}    
\begin{proof}
    According to \eqref{eq:comb_H-G_mnpq} and \eqref{eq:comb_RM}, we have
\begin{align}
\tr_{G_{m,n,p,q}^{(c)}}\left(\ket{\psi_\alpha}\right) &= \alpha\left( \{1,2,3\}\right)^{-3-2(m+n+p+q)}  [ \alpha\left( \{c,c'\}\right) \alpha\left( \{c,c''\}\right)]^{-3-(m+n+p+q)}  \alpha\left( \{c',c''\}\right)^{-2-2(m+n+p+q)} \,,\\
\tr_{\RM_{2(m+n+p+q+1)}^{(c)}}\left(\ket{\psi_\alpha}\right) &= \alpha\left( \{1,2,3\}\right)^{-1-2(m+n+p+q)}  [ \alpha\left( \{c,c'\}\right) \alpha\left( \{c,c''\}\right)]^{-1-(m+n+p+q)}  \alpha\left( \{c',c''\}\right)^{-2(m+n+p+q)}
\,,\\
\tr_{H_{m,n,p,q}^{(c)}}\left(\ket{\psi_\alpha}\right) &= \alpha\left( \{1,2,3\}\right)^{-5-2(m+n+p+q)}  [ \alpha\left( \{c,c'\}\right) \alpha\left( \{c,c''\}\right)]^{-5-(m+n+p+q)}  \alpha\left( \{c',c''\}\right)^{-4-2(m+n+p+q)} \,,
\end{align}
from which we conclude that
\begin{equation}
    \abs{\tr_{G_{m,n,p,q}^{(c)}}\left(\ket{\psi_\alpha}\right)}^2 =\abs{\tr_{H_{m,n,p,q}^{(c)}}\left(\ket{\psi_\alpha}\right)} \cdot \tr_{\RM_{2(m+n+p+q+1)}^{(c)}}\left(\ket{\psi_\alpha}\right)\,, 
\end{equation}
and thus $\SpreadTwo_{m,n,p,q}^{(c)} \left( \ket{\psi_\alpha} \right) = 0$.
\end{proof}

\begin{prop}\label{prop:move2}
 Let $m,n,p,q\in \mathbb{N}$, $c\in \{1,2,3\}$ and $\ket\psi$ a finite-dimensional $3$-partite pure state. We have $\SpreadTwo_{m,n,p,q}^{(c)} \left( \ket\psi \right) = 0$ if and only if there exists $\lambda^{(m,n,p,q)}_c \in \mathbb{C}\setminus \{0\}$ such that $\ket\psi = \vcenter{\hbox{\includegraphics[scale=.5]{images/psi.pdf}}}$ obeys the following graphical equation:
 \begin{equation}\label{eq:move2_graphical_eq}
       \vcenter{\hbox{\includegraphics[scale=.5]{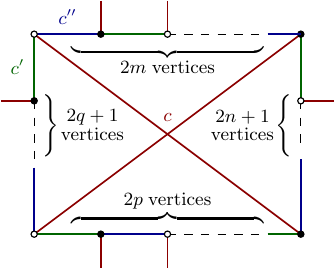}}} = \lambda^{(m,n,p,q)}_c    \vcenter{\hbox{\includegraphics[scale=.5]{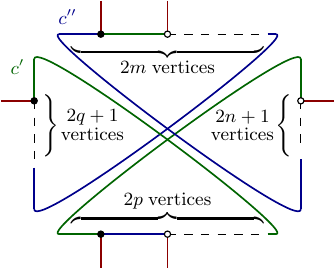}}}\,,
 \end{equation}
 which is associated with the contraction of two non-separating $1$-dipoles. 
 \end{prop}
 \begin{proof}
     The Cauchy-Schwarz inequality \eqref{eq:CS_move2} is saturated if and only if the vectors 
     \begin{equation}
         \vcenter{\hbox{\includegraphics[scale=.5]{images/CS_move2_vector1.pdf}}} \quad    \mathrm{and} \quad     \vcenter{\hbox{\includegraphics[scale=.5]{images/CS_move2_vector2.pdf}}}
     \end{equation}
are linearly-dependent. The norm of the second vector is $\tr_{\RM_{m,n,p,q}}(\ket\psi)\neq 0$, hence this vector is non-zero and we can find $\lambda^{(m,n,p,q)}_c \in \mathbb{C}$ such that \eqref{eq:move2_graphical_eq} holds. Finally, in order for $\SpreadTwo_{m,n,p,q}^{(c)} \left( \ket\psi \right)$ to vanish, one must have $\tr_{G_{m,n,p,q}^{(c)}} \left( \ket\psi \right) \neq 0$, which implies $\lambda^{(m,n,p,q)}_c \neq 0$. 
 \end{proof}
Evaluating the asymptotic expectation values of the family $\{\SpreadTwo^{(c)}_{m,n,p,q}\}$ of spreads in Haar-random states is not straightforward. Indeed, we strongly expect that any invariant $\tr_{G_{m,n,p,q}^{(c)}}$ or $\tr_{H_{m,n,p,q}^{(c)}}$ admits zeroes, and thus cannot follow a power law (and this could be checked explicitly, at least for small values of $m,n,p,q$). Moreover, we do not know whether they obey the large-$N$ factorization criterion in general. What we know is that: $G_{0,0,0,0}^{(c)}$ consists in the graph $\PT_3$ on which an elementary melon graph has been inserted, from which we can deduce that $\Delta(G_{0,0,0,0}^{(c)})= 0$. Therefore, $G_{0,0,0,0}^{(c)}$ does obey the large-$N$ factorization criterion. Similarly, $H_{0,0,0,0}^{(c)}$ can be generated by a mirror flip operation on $\PT_3$ (see Def.~\ref{def:mirror_flip}), \ie $H_{0,0,0,0}^{(c)} = \PT_3\flip_{(e,\bar e)} \overline{\PT_3}$ with $c$ any edge of color $c$ in $\PT_3$. Given that $\Delta(\PT_3)=0$, we conclude (by application of Prop.~\ref{prop:exchange_flip}) that $H_{0,0,0,0}^{(c)}$ also obeys the large-$N$ factorization criterion. As a result, one can show that, at large $N$, the \emph{typical value} of $\SpreadTwo_{0,0,0,0}^{(c)}$ in the Haar-random state of local dimension $N$ is\footnote{We informally mean by this that the random variable $\SpreadTwo_{0,0,0,0}^{(c)}(\ket{\varphi_N})$ takes a value of order $1$ with high probability, and a more precise way of quantifying this can be found \eg in Prop.~6.11 of \cite{Carrozza:2026qcf}.} 
\begin{equation}\label{eq:typical_value_Q0000}
\ln \langle\tr_{H_{0,0,0,0}^{(c)}} (\ket{\varphi_N})  \rangle + \ln \langle \tr_{\RM_{2}^{(c)}} (\ket{\varphi_N})  \rangle - 2 \ln \langle \tr_{G_{0,0,0,0}^{(c)}} (\ket{\varphi_N}) \rangle   \underset{N \to \infty}{=} o(1)\,.
\end{equation}
Indeed, one can check that 
\begin{equation}
    \langle\tr_{G_{0,0,0,0}^{(c)}} (\ket{\varphi_N})  \rangle \underset{N\to +\infty}{\sim} 3 N^{-4}\,,\quad  \langle \tr_{\RM_{2}^{(c)}} (\ket{\varphi_N})  \rangle \underset{N\to +\infty}{\sim} N^{-1}\,, \quad\mathrm{and} \quad\langle\tr_{H_{0,0,0,0}^{(c)}} (\ket{\varphi_N})  \rangle \underset{N\to +\infty}{\sim} 9 N^{-7}\,.
\end{equation}

For arbitrary values of $m,n,p,q$, we were able to establish that:
\begin{equation}
  \forall m, n , p , q \in \mathbb{N}^*\,, \qquad    \begin{cases}
     \Delta( G_{m,n,p,q}^{(c)} ) &\leq 6
     \\
     \Delta( H_{m,n,p,q}^{(c)} ) &\leq 7
  \end{cases}
\end{equation} 
Indeed, one can check that $\Delta_0( \widehat{G_{m,n,p,q}^{(c)}} )=6$ and $\Delta_0( \widehat{H_{m,n,p,q}^{(c)}} )=7$ for the two $(D+1)$-colored graphs $\widehat{G_{m,n,p,q}^{(c)}}$ and $\widehat{H_{m,n,p,q}^{(c)}}$ illustrated in Fig.~\ref{fig:Gmnpq_hat}. This observation is not by itself sufficient to determine the large-$N$ behaviour of $\SpreadTwo_{m,n,p,q}^{(c)}(\ket{\varphi_N})$ for arbitrary values of $m,n,p,q$, we therefore leave this question open for the future.
 
   \begin{figure}
        \centering
        \includegraphics[width = .6\textwidth]{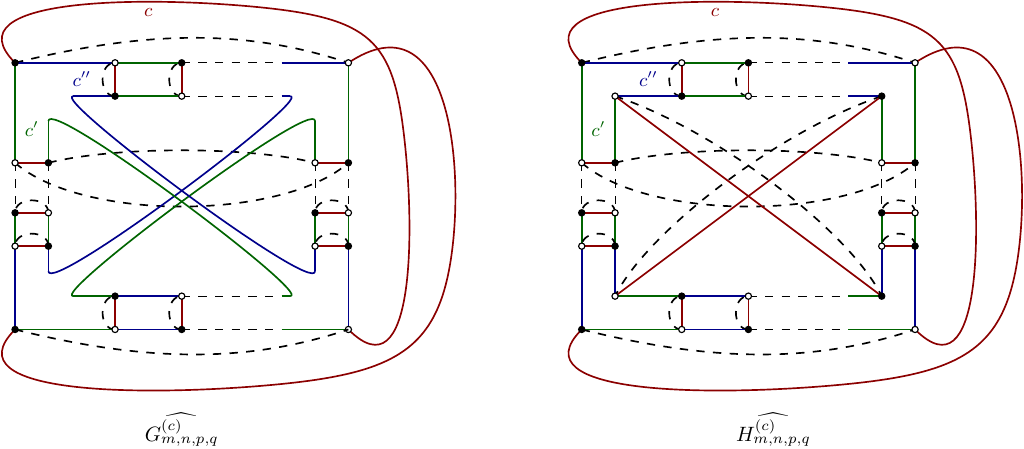}
        \caption{Configurations $\widehat{G_{m,n,p,q}^{(c)}}$ and $\widehat{H_{m,n,p,q}^{(c)}}$ allowing us to prove that $\Delta( G_{m,n,p,q}^{(c)} )\leq 6$ and
     $\Delta( H_{m,n,p,q}^{(c)} ) \leq 7$ for $m,n,p,q >1$.}
        \label{fig:Gmnpq_hat}
    \end{figure} 
 
\section{Characterization of tripartite hypergraph-tensor states as minimizers of entanglement spreads}\label{sec:HTstates}

\subsection{Preliminaries: bipartite setting}\label{sec:bipartite_charac}

It is well-known that the bipartite entanglement spread of a bipartite pure state vanishes if and only if that state has flat entanglement spectrum.
\begin{prop}\label{prop:flat_spectrum_2d}
    Let $m, n \in \mathbb{N}$ such that $2 \leq m < n$, and $\ket\psi$ a finite-dimensional $2$-partite pure state. We have: 
    \begin{align}
        \mathsf{D}_{m , n} (\ket\psi) = 0 &\quad \Leftrightarrow \quad \ket\psi \; \mathrm{has} \; \mathrm{flat} \; \mathrm{entanglement}\; \mathrm{spectrum} \\  
        &\quad \Leftrightarrow \quad \exists \alpha \in W_2\,, \quad \ket\psi \simLU \ket{\psi_\alpha} 
    \end{align}
\end{prop}
We are interested in exploring generalizations of this simple observation to $D = 3$. As a warm-up, let us rederive  the following subcase of Prop.~\ref{prop:flat_spectrum_2d}, by means of largely combinatorial methods that will generalize well to the tripartite setting: for any finite-dimensional bipartite pure state $\ket\psi$
    \begin{align}
        \mathsf{D}_{2 , 3} (\ket\psi) = 0
        &\quad \Leftrightarrow \quad \exists \alpha \in W_2\,, \quad \ket\psi \simLU \ket{\psi_\alpha}\,.
        \end{align}
Let $\ket\psi = \vcenter{\hbox{\includegraphics[scale=.5]{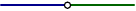}}}$ be a finite-dimensional $2$-partite pure state. We have the following Cauchy-Schwarz inequality:
\begin{equation}
    \abs{\tr_{C_2}(\ket\psi)}^2= \abs{ \vcenter{\hbox{\includegraphics[scale=.5]{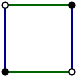}}} }^2 \leq \vcenter{\hbox{\includegraphics[scale=.5]{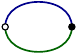}}}  \times \vcenter{\hbox{\includegraphics[scale=.5]{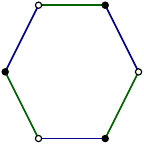}}}  = \vcenter{\hbox{\includegraphics[scale=.5]{images/C6.pdf}}} = \tr_{C_3}(\ket\psi)\,,
\end{equation}
with equality if and only if the vectors 
\begin{equation}
    \vcenter{\hbox{\includegraphics[scale=.5]{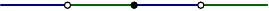}}} \quad \mathrm{and} \quad \vcenter{\hbox{\includegraphics[scale=.5]{images/psi_2d.pdf}}} 
\end{equation}
are linearly dependent. The second vector is $\ket\psi$ itself, hence it is a non-zero vector. We can thus find $\rho \in \mathbb{C}$ such that
\begin{equation}\label{eq:graphical_eq_2d}
    \vcenter{\hbox{\includegraphics[scale=.5]{images/vector1_2d.pdf}}} = \rho\; \vcenter{\hbox{\includegraphics[scale=.5]{images/psi_2d.pdf}}} \,. 
\end{equation}
By taking the inner product of the previous equation with $\ket\psi$, it is clear that we must have $\rho = \tr_{C_3}(\ket\psi)$, so that
\begin{equation}
    \vcenter{\hbox{\includegraphics[scale=.5]{images/vector1_2d.pdf}}} = \vcenter{\hbox{\includegraphics[scale=.5]{images/C6.pdf}}}\; \vcenter{\hbox{\includegraphics[scale=.5]{images/psi_2d.pdf}}} \,. 
\end{equation}
This type of graphical equation can be interpreted as a special case of Eq.~\eqref{eq:CS_minimizers}, and can also be obtained from Eq.~\eqref{eq:equality_move1_3} by suppressing all the edges of color $c$ (and taking $\rho_c^{(3)}=1$). Making repeated use of this equation, we infer that: for any integer $k \geq 2$,
\begin{equation}
    \tr_{C_k} (\ket\psi) = \rho \tr_{C_{k-1}} (\ket\psi) = \cdots = \rho^{k-1} \tr_{C_1} (\ket\psi) = \rho^{k-1}\,.  
\end{equation}
In addition, we notice that $\tr_{C_k} (\ket\psi)= \tr\left( M^k \right)$ with $M= \vcenter{\hbox{\includegraphics[scale=.5]{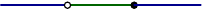}}}$ the reduced density matrix of $\ket\psi$ on one of the subsystems. $M$ being positive, Hermitian and of unit trace, it is in particular diagonalizable. Denoting by $n_\lambda \in \mathbb{N}^*$ the multiplicity of $\lambda \in \spec(M)$, we observe that: for any $z \in \mathbb{C}$ such that $\abs{z} > \underset{\lambda \in \spec(M)}{\max} \abs{\lambda}$, the resolvent $\left( z - M\right)^{-1}$ can be computed by a convergent series, leading to
\begin{equation}
    \sum_{\lambda \in \spec(M) \setminus\{0\}} \frac{\lambda n_\lambda}{z- \lambda} = \tr\left( M \left( z - M\right)^{-1} \right)= \sum_{k=1}^{+\infty} \frac{\Tr\left( M^k\right)}{z^k} = \sum_{k=1}^{+\infty} \frac{\rho^{k-1}}{z^k} =  \frac{1}{z- \rho}\,.
\end{equation}
It follows that $\spec(M) \setminus \{0\} = \{ \rho \}$ and $\rho n_\rho = 1$. Defining $N\eqdef n_\rho \in \mathbb{N}^*$, it results that:
\begin{equation}
    \rho = \frac{1}{N}\,,
\end{equation}
and therefore:
\begin{equation}
    \forall k \in \mathbb{N}^*\,, \quad \tr_{C_k} ( \ket\psi ) = N^{1-k}\,.
\end{equation}
Now, let $\ket{\psi_\alpha}$ denote the HT state of weight function $\alpha \in W_2$ defined by: $ \alpha(\{1,2\})\eqdef N$. It is easy to check that:
\begin{equation}
    \forall k \in \mathbb{N}^*\,, \quad \tr_{C_k} ( \ket{\psi_\alpha} ) = N^{1-k}\,.
\end{equation}
Since the family of $2$-colored graphs $\{C_k\}_{k \in \mathbb{N}^*}$ separates the $\LU$-orbits of bipartite pure states, we must therefore conclude that $\ket\psi \simLU \ket{\psi_\alpha}$. 

\subsection{Tripartite setting}\label{sec:tripartite_HT}
In this subsection, $D=3$.

\subsubsection{First characterization}

\begin{theo}\label{thm:first_charac}
    Let $\ket\psi$ be a finite-dimensional $3$-partite pure state. There exists $\alpha \in W_3$ such that $\ket\psi \simLU \ket{\psi_\alpha}$ if and only if the following conditions are obeyed:
    \begin{enumerate}
        \item for any $n \in \mathbb{N}^*$ and $c\in \{1,2,3\}$, $\SpreadOne_n^{(c)}(\ket\psi ) = 0$;
        \item there exists $N\in \mathbb{N}^*$ such that, for any maximally single-trace graph $G \in \cG_3$, $\tr_G (\ket\psi ) = N^{1-k(G)}$.
    \end{enumerate}
\end{theo}
\begin{proof}
    Suppose that $\ket\psi \simLU \ket{\psi_\alpha}$ for some weight function $\alpha$. By Prop.~\ref{prop:SpreadOne_alpha_states}, condition 1 is verified. Furthermore, for any $G$ a maximally single-trace $3$-colored graph, we have $\kappa(G)=\kappa(G\vert_{1,2})= \kappa(G\vert_{1,3})=\kappa(G\vert_{2,3})=1$, leading to $\tr_G (\ket\psi)= N^{1-k(G)}$ with $N\eqdef \alpha(\{1,2,3\})\alpha(\{1,2\})\alpha(\{1,3\}) \alpha(\{2,3\})$. Hence, condition 2 is also verified.
    
    Reciprocally, let us assume that $\ket\psi = \vcenter{\hbox{\includegraphics[scale=.5]{images/psi.pdf}}}$ obeys condition 1. From Prop.~\ref{prop:graph_eq_move1}, the graphical equation \eqref{eq:graph_eq_move1} is obeyed for any $c$ and any $n$. As a result, for any $m,n \in \mathbb{N}^*$ and $c \in \{1,2,3\}$, we have:
    \begin{equation}
        \rho_c^{(m)}  \vcenter{\hbox{\includegraphics[scale=.5]{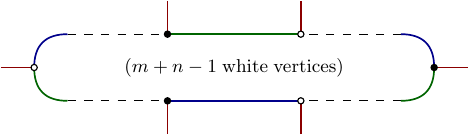}}} = \vcenter{\hbox{\includegraphics[scale=.5]{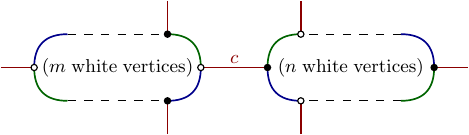}}} = \rho_c^{(n)}  \vcenter{\hbox{\includegraphics[scale=.5]{images/move1-mn-bis.pdf}}} 
    \end{equation}
Given that the vector
\begin{equation}
    \vcenter{\hbox{\includegraphics[scale=.5]{images/move1-mn-bis.pdf}}}
\end{equation}
is non-zero (because its norm squared is $\tr_{\RM_{2(m+n-1)}}(\ket\psi) \neq 0$), we conclude that $\rho_c^{(m)}=\rho_c^{(n)}$. Hence, for any $c \in \{1,2,3\}$, we have $\rho_c \in \mathbb{C}\setminus \{0\}$ such that:
 \begin{equation}\label{eq:graph_eq_move1_update1}
 \vcenter{\hbox{\includegraphics[scale=.5]{images/move1-1.pdf}}}= \rho_c \vcenter{\hbox{\includegraphics[scale=.5]{images/move1-2.pdf}}}
 \end{equation}
Specializing to $n=1$ gives:
\begin{equation}
 \vcenter{\hbox{\includegraphics[scale=.5]{images/move1_1a.pdf}}} = \rho_c \vcenter{\hbox{\includegraphics[scale=.5]{images/move1_1b.pdf}}} \,.  
\end{equation}
Viewing $\ket\psi$ as a bipartite state relative to the partition $\{1,2,3\}= \{c\}\sqcup \{c', c''\}$, this equation is an instance of \eqref{eq:graphical_eq_2d}. Invoking the same argument as in Sec.~\ref{sec:bipartite_charac}, we can therefore find $N_c \in \mathbb{N}^*$ such that $\rho_c = \frac{1}{N_c}$. For any connected graph $G \in \cG_3^{\conn}$, we can use \eqref{eq:graph_eq_move1_update1} recursively to compute $\tr_G (\ket\psi)$. Indeed, for any $c\in \{1,2,3\}$, if $\kappa(G\vert_{c',c''})\geq 2$ we can apply the graphical move from \eqref{eq:graph_eq_move1_update1} to infer that
\begin{equation}
    \tr_G \left( \ket\psi \right) =\frac{1}{N_c} \tr_{G'} \left( \ket\psi \right)\,,
\end{equation}
where $G'\in \cG_3^{\conn}$ is a $3$-colored graph obeying
\begin{equation}
    \kappa(G'\vert_{c',c''}) = \kappa(G\vert_{c',c''})-1\,,\quad \kappa(G'\vert_{c,c'}) = \kappa(G\vert_{c,c'})\,, \quad \kappa(G'\vert_{c,c''}) = \kappa(G\vert_{c,c''})\,, \quad \mathrm{and}\quad k(G')=k(G)-1\,.
\end{equation}
Iterating, we can find a maximally single-trace graph $H\in \cG_3^{\conn}$ such that
\begin{equation}
    \tr_G \left( \ket\psi \right) =\left( \prod_{c\in \{1,2,3\}} \left(\frac{1}{N_c}\right)^{\kappa(G\vert_{c',c''})-1}\right) \tr_{H} \left( \ket\psi \right) 
\end{equation}
and
\begin{equation}
    k(H)=k(G)-\sum_{c\in \{1,2,3\}} \left( \kappa(G\vert_{c',c''})-1\right) = 3-2k(G)- \sum_{c\in \{1,2,3\}} \left( \kappa(G\vert_{c',c''})-k(G)\right)\,.  
\end{equation}
At this point, we have only made use of condition 1. Assuming that condition 2 is also obeyed by $\ket\psi$, we can find some $N \in \mathbb{N}^*$ (that is independent from $G$ or $H$) such that
\begin{equation}
    \tr_H(\ket\psi)= N^{1-k(H)}= N^{-2+2k(G)+ \sum_{c\in \{1,2,3\}} \left( \kappa(G\vert_{c',c''})-k(G)\right)}
\end{equation}
All in all, we conclude that
\begin{equation}
    \tr_G(\ket\psi) = \prod_{c=1}^3 \left( \frac{N}{N_c}\right)^{\kappa(G\vert_{c',c''})-k(G)} \times \left(N^{-2}\prod_{c=1}^3 N_c \right)^{1-k(G)} = \prod_{B \subset\{1,2,3\}\,, \,\abs{B}\geq 2} \alpha(B)^{\kappa(G\vert_B) - k(G)}\,,
\end{equation}
where we have defined
\begin{equation}\label{eq:alpha_1st_thm}
\forall c \in\{1,2,3\}\,, \quad \alpha(\{1,2,3\}\setminus \{c\})\eqdef \frac{N}{N_c}\,, \qquad \mathrm{and} \qquad \alpha(\{1,2,3\})\eqdef N^{-2}\prod_{c=1}^3 N_c\,.
\end{equation}
By Lem.~\ref{lem:alpha_integer_valued}, the function $\alpha$ is integer-valued and $\ket\psi \simLU \ket{\psi_\alpha}$, which concludes the proof.
\end{proof}

\subsubsection{Second characterization}

\begin{theo}\label{thm:second_charac}
    Let $\ket\psi$ be a finite-dimensional $3$-partite pure state. There exists $\alpha \in W_3$ such that $\ket\psi \simLU \ket{\psi_\alpha}$ if and only if the following conditions are obeyed:
    \begin{enumerate}
        \item for any $n \in \mathbb{N}^*$ and $c\in \{1,2,3\}$, $\SpreadOne_n^{(c)}(\ket\psi ) = 0$;
        \item for any $m,n,p,q \in \mathbb{N}$ and $c\in \{1,2,3\}$, $\SpreadTwo_{m,n,p,q}^{(c)}(\ket\psi ) = 0$.
    \end{enumerate}
\end{theo}
\begin{proof}
    Suppose that $\ket\psi \simLU \ket{\psi_\alpha}$ for some weight function $\alpha$. By Prop.~\ref{prop:SpreadOne_alpha_states} and Prop.~\ref{prop:SpreadTwo_alpha_states}, conditions 1 and 2 are both verified. 
    
    Reciprocally, let us assume that $\psi= \vcenter{\hbox{\includegraphics[scale=.5]{images/psi.pdf}}}$ obeys both conditions 1 and 2. Using the same argument as in the proof of Thm.~\ref{thm:first_charac} (where we analyzed the implications of condition 1 in detail), we can find $N_1, N_2 , N_3 \in \mathbb{N}^*$ such that: for any $n \in \mathbb{N}^*$ and any $c\in \{1,2,3\}$,
    \begin{equation}\label{eq:graph_eq_move1_update2}
 \vcenter{\hbox{\includegraphics[scale=.5]{images/move1-1.pdf}}}= \frac{1}{N_c} \vcenter{\hbox{\includegraphics[scale=.5]{images/move1-2.pdf}}}\,.
 \end{equation}
 According to Prop.~\ref{prop:move2}, we also have non-zero complex numbers $\{\lambda_{m,n,p,q}^{(c)}\}_{m,n,p,q \in \mathbb{N}, 1\leq c\leq 3}$ such that the family of graphical equations \eqref{eq:move2_graphical_eq} holds. Let us prove that $\lambda_{m,n,p,q}^{(c)}$ is independent from $m,n,p,q$ and $c$. We proceed in two steps.
 \begin{itemize}
     \item First, for any $m,n,p,q\in \mathbb{N}$ and $c\in \{1,2,3\}$, we have 
     \begin{equation}
        \vcenter{\hbox{\includegraphics[scale=.5]{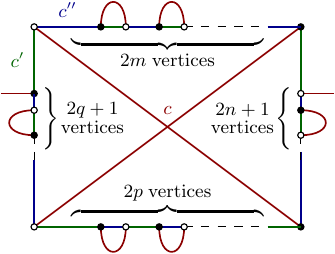}}} = \lambda_{c}^{(m,n,p,q)} \, \vcenter{\hbox{\includegraphics[scale=.5]{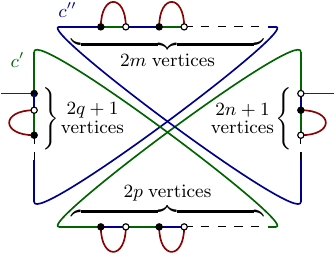}}} =  \lambda_{c}^{(m,n,p,q)} N_{c''}^{-m-n-q}N_{c'}^{-p} \, \vcenter{\hbox{\includegraphics[scale=.5]{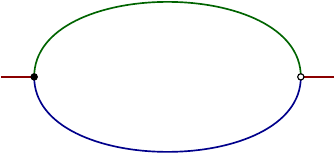}}} \,,
     \end{equation}
     where we first applied \eqref{eq:move2_graphical_eq} once, and then applied successively \eqref{eq:graph_eq_move1_update2} $m+n+p+q$ times. On the other hand, if first apply \eqref{eq:graph_eq_move1_update2} $m+n+p+q$ times and then \eqref{eq:move2_graphical_eq} once, we obtain instead
          \begin{equation}
        \vcenter{\hbox{\includegraphics[scale=.5]{images/move2_vector1_contracted.pdf}}} = N_{c''}^{-m-n-q}N_{c'}^{-p} \, \vcenter{\hbox{\includegraphics[scale=.5]{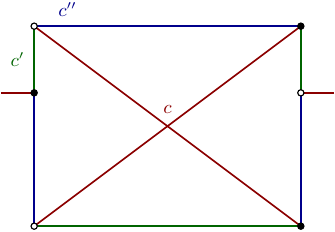}}}= N_{c''}^{-m-n-q}N_{c'}^{-p}  \lambda_c^{(0,0,0,0)} \, \vcenter{\hbox{\includegraphics[scale=.5]{images/melon_2pt.pdf}}}\,.
     \end{equation}
     Combining these two equations, it follows that $\lambda_c^{(m,n,p,q)}= \lambda_c^{(0,0,0,0)}$, which we denote from now on by $\lambda_c$.
    \item Second, we can prove that $\lambda_c$ does not depend on the value of $c$. Indeed, we have
    \begin{equation}
        \tr_{\PT_3}(\ket\psi)= \vcenter{\hbox{\includegraphics[scale=.5]{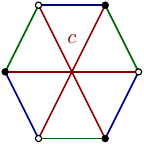}}} = \vcenter{\hbox{\includegraphics[scale=.5]{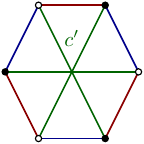}}} = \vcenter{\hbox{\includegraphics[scale=.5]{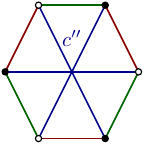}}}\,, 
    \end{equation}
which we can compute in three distinct way by means of \eqref{eq:move2_graphical_eq}, leading to
\begin{equation}
\lambda_1 \, \vcenter{\hbox{\includegraphics[scale=.5]{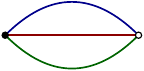}}} = \lambda_2 \, \vcenter{\hbox{\includegraphics[scale=.5]{images/elem_melon.pdf}}}    = \lambda_3 \, \vcenter{\hbox{\includegraphics[scale=.5]{images/elem_melon.pdf}}}\,.
\end{equation}
We thus conclude that $\lambda_1 = \lambda_2 = \lambda_3$, which we denote simply as $\lambda$ from now on.
 \end{itemize}
At this point, we can use the combinatorial relations \eqref{eq:graph_eq_move1_update2} and \eqref{eq:move2_graphical_eq} to determine the value of any trace-invariant in the state $\ket\psi$. Let $G\in \cG_3^{\conn}$. Proceeding as in the proof of Thm.~\ref{thm:first_charac}, we can find a maximally single-trace $H$ with $\displaystyle k(H)=3 - 2k(G)- \sum_{c=1}^3 \left( \kappa\left( G_{\vert_{c',c''}}\right)-k(G)\right)$ and such that
\begin{equation}\label{eq:tr_G-1}
    \tr_G\left( \ket\psi \right) = \left( \prod_{c=1}^3 N_c^{1- \kappa(G\vert_{c',c''})}\right) \Tr_H\left( \ket\psi \right)\,.
\end{equation}
Now, let us prove that $k(H)$ is odd and that:
\begin{equation}\label{eq:eval_MST_lambda}
    \Tr_H\left( \ket\psi \right) = \lambda^{(k(H) -1)/2}\,.
\end{equation}
We proceed by induction on the size of $H$ (as measured by $k(H)$). If $k(H)=1$, $k(H)$ is odd, $H$ is the two-vertex $3$-colored graph, and \eqref{eq:eval_MST_lambda} trivially holds. Suppose now that $H$ is a maximally single-trace graph with $k(H)>1$. Then $H$ necessarily admits a subgraph of the form\footnote{Indeed, there is a single face of color $(c',c'')$ visiting all the vertices of the graph, which we can use to define a preferred cyclic ordering of the vertices (as represented in the picture). Picking up an edge of color $c$ (\ie one of the two diagonals represented in the picture), which connects two nodes along this cycle, we can then follow the face of color $(c,c')$: it contains at least four edges (otherwise $H$ would reduce to the two-vertex graph) and must also visit all the vertices in the cycle. But this can only happen if there is a second red edge that ``crosses'' the first one (as represented by the second diagonal in the picture). The rest of the structure of the subgraph is fixed by the coloring constraints. Alternatively, one may justify this structure by noticing that the graph $H$ cannot be planar: indeed, according to footnote \ref{foot:genus}, $H$ is a maximally single-trace graph and its genus $g(H)$ obeys the relation $2-2g(H)=3 - k(H)\leq 1$. As a result, $g(H)>0$, which implies that the graph is non-planar (hence we necessarily have crossings of color-$c$ edges along the cycle of color $(c',c'')$ previously considered). }  
\begin{equation}
    \vcenter{\hbox{\includegraphics[scale=.5]{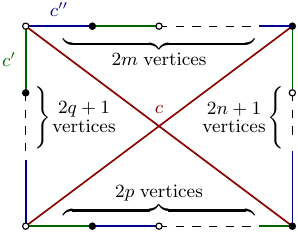}}}\,,
\end{equation}
for some values of $m,n,p,q\in \mathbb{N}$. We can therefore apply the relation \eqref{eq:move2_graphical_eq} to find a new graph $H'$ such that $k(H') = k(H)-2$ and
\begin{equation}
    \tr_H\left( \ket\psi \right) = \lambda \tr_{H'}\left( \ket\psi \right)\,. 
\end{equation}
More precisely, $H'$ is obtained from $H$ after implementing the following subgraph substitution:
\begin{equation}
        \vcenter{\hbox{\includegraphics[scale=.5]{images/CS_move2_vector1-bis.pdf}}} \longrightarrow     \vcenter{\hbox{\includegraphics[scale=.5]{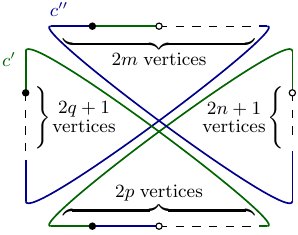}}}\,,
\end{equation}
which preserves the number of faces of each color. As a result, $H'$ is a maximally single-trace graph. Since it has $k(H') = k(H)-2< k(H)$, we 
can apply the induction hypothesis to deduce that: $k(H')$ is odd and
\begin{equation}
        \Tr_{H'}\left( \ket\psi \right) = \lambda^{(k(H') -1)/2}= \lambda^{(k(H) -1)/2 - 1} \,.
\end{equation}
As a result, $k(H)$ is also odd and 
\begin{equation}
        \Tr_{H}\left( \ket\psi \right) = \lambda \Tr_{H'}\left( \ket\psi \right) =  \lambda^{(k(H) -1)/2} \,.
\end{equation}
Combining the intermediate results \eqref{eq:tr_G-1} and \eqref{eq:eval_MST_lambda}, we finally obtain: for any $G \in \cG_3^{\conn}$,
\begin{equation}\label{eq:tr_G_interm}
    \tr_G\left( \ket\psi \right) = \lambda^{g(G)} \prod_{c=1}^3 N_c^{1 - \kappa\left( G\vert_{\{1,2,3\}\setminus \{c\}}\right)}\,,
\end{equation}
where $g(G) \in \mathbb{N}$ denotes the \emph{genus} of $G$ and is defined by the Euler relation\footnote{\label{foot:genus}This is indeed the standard Euler relation $2-2g(G)= F(G) - E(G) + V(G)$ where: the number of faces of $G$ is $F(G)\eqdef \sum_{c=1}^3 \kappa\left( G\vert_{\{1,2,3\}\setminus \{c\}}\right)$, its number of edges is $E(G)= 3k(G)$, and its number of vertices is $V(G)=2k(G)$. The relation $E(G)=3 k(G)$ is itself a consequence of the $3$-regular nature of $G$, which implies that $3V(G)=2E(G)$.}
\begin{equation}
    2-2g(G) \eqdef  \sum_{c=1}^3 \kappa\left( G\vert_{\{1,2,3\}\setminus \{c\}}\right) - 3k(G) + 2k(G) = \sum_{c=1}^3 \kappa\left( G\vert_{\{1,2,3\}\setminus \{c\}}\right) - k(G)\,.
\end{equation}
Next, let us prove that $\lambda >0$. To do so, we can rely on the properties of the partially-transposed reduced density matrix $M \eqdef \rho_{12}^{T_2}$ (see Eq.~\eqref{eq:PartialTransp_Realignm}). $M$ is Hermitian (hence diagonalizable) and, thanks to \eqref{eq:tr_G_interm}, its moments can be computed explicitly:
\begin{equation}
    \forall k \in \mathbb{N}^*\,, \qquad \tr\left( M^k \right) = \tr_{\PT_k^{(3)}} \left( \ket\psi\right)= \begin{cases}
    \lambda^{\frac{k-1}{2}} \quad &\mathrm{if}\quad k\;\mathrm{is}\;\mathrm{odd}\\
    \lambda^{\frac{k-2}{2}} N_3^{-1} \quad &\mathrm{if}\quad k\;\mathrm{is}\;\mathrm{even}
    \end{cases}
\end{equation}
As a result, we have:
for any $z \in \mathbb{C}$ such that $\abs{z} > \underset{x \in \spec(M)}{\max} \abs{x}$
\begin{align}\label{eq:spec_PTk}
    \sum_{x \in \spec(M) \setminus\{0\}} \frac{x m_x}{z- x} &= \tr\left( M \left( z - M\right)^{-1} \right)= \sum_{k=1}^{+\infty} \frac{\Tr\left( M^k\right)}{z^k} = \sum_{n\in \mathbb{N}}^{+\infty} \frac{\lambda^{n}}{z^{2n+1}} + N_3^{-1}\sum_{n \in \mathbb{N}} \frac{\lambda^n}{z^{2n+2}} \nonumber \\
    &= \frac{z}{z^2- \lambda} + \frac{N_3^{-1}}{z^2- \lambda}\,.
\end{align}
 where $m_x$ denotes the mutliplicity of the eigenvalue $x \in \spec(M)$. 
Since $\lambda \neq 0$,\footnote{This already known fact can be recovered from the present calculation: if we had $\lambda =0$, the expression in the second line of \eqref{eq:spec_PTk} would reduce to
 \begin{equation}
     \frac{z+N_3^{-1}}{z^2} \,,
 \end{equation}
 which has a double-pole at $z=0$, thus leading to an inconsistency with the left-hand-side of the same equation.} the previous expression has two simple poles, located at $\pm \sqrt{\lambda}$. It follows that $M$ has exactly two non-zero eigenvalues: 
\begin{equation}
    \spec(M)\setminus \{0\} = \{+\sqrt{\lambda}, - \sqrt{\lambda}\}\,.
\end{equation}
Given that $M$ is Hermitian, both $\sqrt{\lambda}$ and $-\sqrt{\lambda}$ must be real, which can only happen if $\lambda >0$.

The strict positivity of $\lambda$ allows one to reorganize equation \eqref{eq:tr_G_interm} to bring it to the following form:
\begin{equation}\label{eq:trG_int_2}
    \forall G \in \cG_3^{\conn} \,, \quad \tr_G\left( \ket\psi \right) = \prod_{B \subset\{1,2,3\}\,, \,\abs{B}\geq 2} \alpha(B)^{\kappa(G\vert_B) - k(G)}\,,
\end{equation}
where we have defined
\begin{equation}\label{eq:alpha_2nd_thm}
\forall c \in\{1,2,3\}\,, \quad \alpha(\{1,2,3\}\setminus \{c\})\eqdef \frac{1}{\sqrt{\lambda} N_c}\,, \qquad \mathrm{and} \qquad \alpha(\{1,2,3\})\eqdef \lambda\prod_{c=1}^3 N_c\,.
\end{equation}
Since $\lambda >0$, the function $\alpha$ is valued in $\mathbb{R}_+^*$, so we can apply Lem.~\ref{lem:alpha_integer_valued} to conclude that it in fact takes value in $\mathbb{N}^*$. It thus defines a weight function $\alpha \in W_3$, and \eqref{eq:trG_int_2} establishes that $\ket\psi \simLU \ket{\psi_\alpha}$.
\end{proof}
\begin{rem}\label{rem:resolvent_PT}
    We can push the analysis of equation \eqref{eq:spec_PTk} a little further to bring it in the form:
    \begin{equation}
        \sum_{x \in \{\sqrt{\lambda}, - \sqrt{\lambda}\}} \frac{x m_x}{z- x} = \frac{\sqrt{\lambda}+N_c^{-1}}{2\sqrt{\lambda}} \frac{1}{z-\sqrt{\lambda}} + \frac{\sqrt{\lambda}-N_c^{-1}}{2\sqrt{\lambda}} \frac{1}{z+\sqrt{\lambda}}\,.
    \end{equation}
This relation implies that 
\begin{align}
    \frac{N_c^{-1}+\sqrt{\lambda}}{2\lambda} &= m_{\sqrt{\lambda}} \in \mathbb{N^*} \,,\\
    \frac{N_c^{-1}-\sqrt{\lambda}}{2\lambda} &= m_{-\sqrt{\lambda}} \in \mathbb{N^*} \,,
\end{align}
from which it follows that
\begin{equation}
\frac{1}{\sqrt{\lambda}} =    m_{\sqrt{\lambda}} - m_{-\sqrt{\lambda}} \in \mathbb{N}^*
\end{equation}
In other words, we can find some $N \in \mathbb{N^*}$ such that
\begin{equation}
    \lambda = N^{-2}\,.
\end{equation}
The weight function $\alpha$ defined in the proof of  Thm.~\ref{thm:second_charac} (see Eq.~\eqref{eq:alpha_2nd_thm}) therefore takes the same functional form as the one introduced in  Thm.~\ref{thm:first_charac} (see Eq.~\eqref{eq:alpha_1st_thm}). In addition, formula \eqref{eq:tr_G_interm} can now be expressed as
\begin{equation}
\forall G \in \cG_3^{\conn}\,, \quad    \tr_G\left( \ket\psi \right) = N^{-2g(G)} \prod_{c=1}^3 N_c^{1 - \kappa\left( G\vert_{\{1,2,3\}\setminus \{c\}}\right)}\,.
\end{equation}
The factor $N^{-2g(G)}$ is reminiscent of the way in which graphs are weighted in random matrix models admitting a topological large $N$ expansion: only graphs of trivial topology ($g(G)=0$) contribute at leading order in the asymptotic large $N$ limit.
\end{rem}
As a consequence of Thm.~\ref{thm:second_charac}, an HT state $\ket{\psi_\alpha}$ is \emph{characterized}, in the space of all finite-dimensional $3$-partite pure states, by the family of graphical equations:
  \begin{align}
\forall n \in \mathbb{N}^*\,,\forall c\in\{1,2,3\}\,, \quad \vcenter{\hbox{\includegraphics[scale=.5]{images/move1-1.pdf}}}&= \frac{1}{N_c} \vcenter{\hbox{\includegraphics[scale=.5]{images/move1-2.pdf}}}\,,
     \\ 
     \forall m,n,p,q\in \mathbb{N}\,,\forall c\in\{1,2,3\}\,, \quad  \vcenter{\hbox{\includegraphics[scale=.5]{images/CS_move2_vector1.pdf}}} &= \frac{1}{N^2}    \vcenter{\hbox{\includegraphics[scale=.5]{images/CS_move2_vector2.pdf}}}\,.
\end{align}
Furthermore, the one-to-one relation between the sets of  parameters $$\{\alpha(\{1,2,3\}), \alpha(\{1,2\}), \alpha(\{1,3\}) , \alpha(\{2,3\})\}\quad  \mathrm{and} \quad \{N, N_1, N_2 , N_3\}$$ is given by:
\begin{equation}
    \displaystyle\begin{cases} 
     \alpha(\{1,2,3\}\setminus \{c\})&= \frac{N}{N_c} \,,\quad \forall c \\\hfill \alpha(\{1,2,3\})&= N^{-2}\prod_{c} N_c
    \end{cases} \quad \Leftrightarrow \quad
    \begin{cases} 
     N_c&= \alpha(\{1,2,3\}) \prod_{c'\neq c} \alpha(\{1,2,3\}\setminus \{c'\})\,, \quad \forall c \\\hfill N&= \alpha(\{1,2,3\})\prod_{c} \alpha(\{1,2,3\}\setminus \{c\})
    \end{cases}
\end{equation}

\section{Multipartite spectral hypergraph-tensor states and their associated entanglement spreads}\label{sec:SHTstates}

For bipartite quantum systems, entanglement spreads offer ways of characterizing states with flat entanglement spectra. Indeed, as presented before, we recall that given two integers $m,n\in \bb{N}^*$ with $2 \leq m < n$ and the entanglement spread $\mathsf{D}_{m,n}$ introduced in Eq.~\eqref{eq:ent_spread}, we have, for $\ket{\psi} \in \H_1 \ot \H_2$:
\begin{equation} \label{eq:bip_ES_SHT}
    \mathsf{D}_{m,n}(\ket{\psi}) = 0 \qquad \Longleftrightarrow \qquad \exists N \in \bb{N}^*,\, \forall k \in \bb{N}^*\,, \quad \Tr_{C_k}(\ket{\psi}) = N^{1-k} \,.
\end{equation}
In other words, $\mathsf{D}_{m,n}(\ket{\psi})=0$ if and only if the entanglement spectrum is flat. For instance, the entanglement spread of a separable state or of a maximally entangled state is zero.

Be that as it may, no direct analogue of the notion of entanglement spectrum is available in the multipartite setting. More precisely, what is missing is a \emph{universal} concept of spectrum which could be associated with any state. To circumvent this limitation, we will introduce a particular subclass of multipartite states whose $\LU$-orbits can in fact be classified in terms of generalized entanglement spectra. In this restricted context, we will show that certain multipartite entanglement spreads can be used to characterize states with \emph{flat} entanglement spectra. Even though limited in scope, this will provide a natural and interesting generalization of the bipartite result we just recalled, and also justify our choice of nomenclature: just like their bipartite cousins, there is a context in which multipartite entanglement spreads provide quantitative measures of the ``spread'' of some underlying entanglement spectrum. More concretely, we will generalize hypergraph-tensor states (HT states) by assigning a spectrum to each block of the corresponding hypergraph. We refer to these as \textit{spectral hypergraph-tensor states} (SHT states), denoted by $\ket{\psi_{\pmb{p}}}$. Our goal is then to construct a multipartite entanglement monotone $\mathsf{D}$ satisfying:
\begin{equation}
    \mathsf{D}\pa{\ket{\psi_{\pmb{p}}}} = 0 \qquad \Longleftrightarrow \qquad \ket{\psi_{\pmb{p}}} \rm{ is an HT state}\,,
\end{equation}
\lat{i.e.} the spectrum of \textit{every} block of the hypergraph is flat.

\subsection{Definition and properties of spectral hypergraph-tensor states}

\paragraph{Spectra.} In the next definition, we will assign a non-trivial \emph{spectrum} to each collection of subsystems in a $D$-partite system. For that purpose, it will be convenient to work with a notion of spectrum that can accommodate an arbitrary number of non-zero spectral elements. To this effect, we introduce $\Seq$
 as the set of sequences $\pmb{p}= (p_i)_{n \in \mathbb{N}^*} \in [0,1]^{\mathbb{N}^*}$ obeying the following three conditions: (a) $\pmb{p}$ is decreasing \ie spectral elements $p_i$ are ordered from highest to lowest values, as is usually done in this context; (b) $p_i$ is zero for large enough values of $i$ \ie there exists some $r \in \mathbb{N}^*$ such that $p_i=0$ for any $i \geq r+1$; and (c) $\displaystyle\sum_{i \in \mathbb{N}^*} p_i = 1$ \ie $\pmb{p}$ defines a discrete probability measure. In addition, we will denote by $\mathrm{rank}\left(\pmb{p}\right)$ the minimum value of $r$ in conditions (b), which also equals the number of non-zero spectral elements $p_i$ appearing in $\pmb{p}$; we will refer to this number as the \emph{rank} of $\pmb{p}$.

\paragraph{Definition of SHT states.} An SHT state is defined through the use of a \textit{spectral function} $\pmb{p}$, which assigns a spectrum $\pmb{p}(B)$ to any block $B$ of $\{1, \ldots , D\}$ (with $\abs{B}\geq 2$). Associated with it, we have a rank function $r_{\pmb{p}}$, which assigns a rank $r_{\pmb{p}}(B)\eqdef \mathrm{rank}\left( \pmb{p}(B) \right)$ to any block $B$. In this way, a spectral function $\pmb{p}$ provides a map
\begin{equation}
    \begin{array}{ccccc}
        \pmb{p} & : & \paa{ B \subset \paa{1,\ldots,D} \vert \abs{B} \geq 2} & \longrightarrow & \Seq \\
            & & B & \longmapsto & \pmb{p}(B) = \bigl(p_i(B)\bigr)_{i \in \mathbb{N}^*}=  \bigl(p_1(B),\dots,p_{r_{\pmb{p}}(B) }(B), 0, 0 , 0 , \ldots \bigr)  \\
    \end{array}
\end{equation}
We will denote by $\SpecFun_D$ the set of spectral functions associated with $D$ parties.

\begin{defi}\label{def:SHT_states}
    Let $D\geq 2$ and $\pmb{p}\in\SpecFun_D$ a spectral function. The corresponding $D$\emph{-partite spectral hypergraph-tensor state} $\ket{\psi_{\pmb{p}}}$ is defined by:
    \begin{equation}
        \ket{\psi_{\pmb{p}}} = \bigotimes_{B \subset \paa{1,\dots,D},\, \abs{B}\geq 2} \pa{\sum_{i\in\mathbb{N}^*} \sqrt{p_i(B)} \ket{i}^{\ot \abs{B}}} =  \bigotimes_{B \subset \paa{1,\dots,D},\, \abs{B}\geq 2} \pa{\sum_{i=1}^{r_{\pmb{p}}(B)} \sqrt{p_i(B)} \ket{i}^{\ot \abs{B}}} \,.
    \end{equation}
    For each block $B$ of the hypergraph, $\pmb{p}(B)$ is called the \emph{entanglement spectrum of $B$}. By extension, we will say that any state $\ket\psi \simLU \ket{\psi_{\pmb{p}}}$ is a SHT state of \emph{spectral function} $\pmb{p}$.\footnote{Besides, we will prove in Thm.~\ref{th:LUSHT} that a $\LU$-equivalence class of SHT states is characterized by a unique spectral function.}
\end{defi}
\begin{rem}
    For any weight function $\alpha \in W_D$, we have $\ket{\psi_\alpha}= \ket{\psi_{\pmb{p}}}$ with $\pmb{p}\in \SpecFun_D$ obeying: for any $B\subset\{1, \ldots , D\}$ such that $\abs{B}\geq 2$, $r_{\pmb{p}}(B)= \alpha(B)$, and $p_i (B)=\frac{1}{\alpha(B)}$ for any $i\in \{1, \ldots , \alpha(B)\}$. In other words, an HT state is a SHT state with flat entanglement spectrum on any block $B$.
\end{rem}

\paragraph{Evaluation of trace-invariants on SHT states.} Since trace-invariants are multiplicative under tensor products of quantum states, we only need to consider one block $B_0 \subset \paa{1,\ldots,D}$ with $\abs{B_0} \geq 2$. For a given colored graph $G \in \cG_D$, only the colors appearing in $B_0$ will play a role. Thus, the colored graph under study for computing the SHT state associated with the block $B_0$ will be the graph $G$ restricted to the colors in $B_0$, \ie $G\vert_{B_0}$. Trace-invariants being multiplicative under the disjoint union of graphs, we restrict our attention to one connected component of $G\vert_{B_0}$: $H \in \Comp(G\vert_B)$.

Recall that $\Tr_H(\ket{\psi})$ is a sum over indices of products of $k(H)$ copies of the tensor components $\psi$ and $k(H)$ copies of the complex conjugate of the tensor components $\bar{\psi}$. For the block $B_0$, the product will be zero except when all the tensor indices are the same, thus leading to the following expression:
\begin{equation}
    \Tr_H\pa{\sum_{i=1}^{r_{\pmb{p}}(B_0)} \sqrt{p_i(B_0)} \ket{i}^{\ot \abs{B_0}}} = \sum_{i=1}^{r_{\pmb{p}}(B_0)} \sqrt{p_i(B_0)}^{k(H)} \cdot \sqrt{p_i(B_0)}^{k(H)} = \sum_{i=1}^{r_{\pmb{p}}(B_0)} p_i(B_0)^{k(H)}\,.
\end{equation}

Finally, adding the connected components of $G\vert_{B_0}$ and considering all blocks of the SHT state leads to the following expression:
\begin{equation} \label{eq:Tr_SHT}
    \Tr_G\pa{\ket{\psi_{\pmb{p}}}} = \prod_{B \subset \paa{1,\dots,D},\, \abs{B}\geq 2} \pac{ \prod_{H \in \Comp(G\vert_B)} \pa{\sum_{i\in \mathbb{N}^*} p_i(B)^{k(H)}}} \,,
\end{equation}
where we recall that $\Comp(G\vert_B)$ denotes the set of connected components of $G\vert_B$.

\subsection{\texorpdfstring{$\LU$}{LU}-classification of spectral hypergraph-tensor states}\label{sec:LU_SHT}

The purpose of this section is to demonstrate that an SHT state is completely characterized by its spectral function. We begin with a review of the well-understood bipartite setting before extending our discussion to the multipartite case.

\paragraph{Bipartite case.} In this setting, all quantum states are $\LU$-equivalent to a spectral hypergraph-tensor state as demonstrated by the following lemma.

\begin{lem} \label{lem:D2LU_eq_1}
    Let $\ket{\psi}$ a finite-dimensional bipartite pure state. There exists a spectral function $\pmb{p}\in \SpecFun_2$ such that: 
    \begin{equation}
        \ket{\psi} \simLU \ket{\psi_{\pmb{p}}}\,.
    \end{equation}
    In other words, any finite-dimensional bipartite state is a SHT state.
\end{lem}

\begin{proof}
    Let $\ket{\psi} \in \H_1 \ot \H_2$ a finite-dimensional bipartite pure state. Since the reduced density matrix $\rho_1 = \Tr_{\H_2}(\ket{\psi} \bra{\psi})\in \H_1 \otimes \H_1^*$ is a Hermitian positive semi-definite matrix, it has real non-negative eigenvalues. We let $r_{\pmb{p}}(\paa{1,2})$ be the number of non-zero eigenvalues, $p_1(\paa{1,2})\geq \cdots \geq p_{r_{\pmb{p}}(\paa{1,2})}(\paa{1,2})>0$ denote the non-zero eigenvalues themselves (in decreasing order), and we define $\pmb{p}(\paa{1,2}) = (p_1(\paa{1,2}),\dots,p_{r_{\pmb{p}}(\paa{1,2})}(\paa{1,2}), 0, 0, \ldots) \,.$
    Given that $\displaystyle\sum_{i=1}^{+\infty} p_i(\{1,2\})= \tr(\rho_1) = 1$, $\pmb{p}(\paa{1,2})$ is an element of $\Seq$, which we can equivalently understand as an element of  $\SpecFun_2 \simeq \Seq$. 
    
    We let $\ket{\psi_{\pmb{p}}}$ be the SHT state associated with the spectrum $\pmb{p}(\paa{1,2})$, \lat{i.e.}:
    \begin{equation}
        \ket{\psi_{\pmb{p}}} = \sum_{i=1}^{r_{\pmb{p}}(\paa{1,2})} \sqrt{p_i(\paa{1,2})} \ket{i} \ot \ket{i}\,.
    \end{equation}
    Recall that, for $k\geq 1$, $C_k$ is the connected $2$-colored graph with $2k$ vertices. Evaluating the trace-invariants associated with the graphs $\{C_k\}_{k \in \mathbb{N}^*}$ on $\ket{\psi}$ and $\ket{\psi_{\pmb{p}}}$ (by the use of Eq.~\eqref{eq:Tr_SHT}) leads to:
    \begin{equation}
        \Tr_{C_k}(\ket{\psi}) = \Tr(\rho_1^k) = \sum_{i=1}^{r_{\pmb{p}}(\paa{1,2})} p_i(\paa{1,2})^k \qquad \rm{and} \qquad \Tr_{C_k}\pa{\ket{\psi_{\pmb{p}}}} = \sum_{i=1}^{r_{\pmb{p}}(\paa{1,2})} p_i(\paa{1,2})^k \,.
    \end{equation}
    So, for all $k \geq 1$, $\Tr_{C_k}(\ket{\psi}) = \Tr_{C_k}\pa{\ket{\psi_{\pmb{p}}}}$, thus proving the $\LU$-equivalence (see \eg Prop.~2.8 of Ref.~\cite{Carrozza:2026qcf}). 
\end{proof}

\begin{lem} \label{lem:powersum}
    Let $\pmb{p},\pmb{q} \in \Seq$ and $r \in \mathbb{N}^*$ such that: $\mathrm{rank}\left(\pmb{p} \right) \leq r$ and $\mathrm{rank}\left(\pmb{q} \right) \leq r$. One has the equivalence: 
    \begin{equation}
        \forall k \in \{1, \ldots , r\}\,, \quad \sum_{i = 1}^{\mathrm{rank}(\pmb{p})} p_i^k = \sum_{i = 1}^{\mathrm{rank}(\pmb{q})} q_i^k \qquad \Longleftrightarrow \qquad \pmb{p} = \pmb{q}\,.
    \end{equation}
\end{lem}

\begin{proof}
    The implication ($\Leftarrow$) is direct. Regarding the implication ($\Rightarrow$), let us denote by $P_k(\pmb{t})$ the $k$-th power-sum of the $r$-tuple $\pmb{t} = (t_1,t_2,\dots,t_r) \in (\mathbb{R}_+)^{\times r}$, \ie:
    \begin{equation}
        P_k(\pmb{t}) = \sum_{i=1}^{r} t_i^k \,.
    \end{equation}
    By assumption, we have the equality $P_k(\pmb{p}) = P_k(\pmb{q})$, for any $k \in \{1, \ldots , r\}$. 
Moreover, given a $r$-tuple $\pmb{t}$, one can define $E_0(\pmb{t}) = 1$ and: for any $k \geq 1$, $E_k(\pmb{t}) = \displaystyle \frac 1 k \sum_{i=1}^k (-1)^{i-1} E_{k-i}(\pmb{t}) P_i(\pmb{t})$. By the assumed identities involving power sums, we observe that $E_k(\pmb{p}) = E_k(\pmb{q})$, for any $k \in \{0, \ldots , r\}$. Thus, by the Girard–Newton formula, one must have equality of the polynomials $Q_{\pmb{p}}$ and $Q_{\pmb{q}}$ defined as follows: 
    \begin{equation}
        \forall x \in \bb{R}, \quad Q_{\pmb{q}}(x) \eqdef \prod_{i = 1}^{r} (x - q_i) = \sum_{k = 0}^{r} (-1)^k E_k(\pmb{q}) x^{r - k} = \sum_{k = 0}^{r} (-1)^k E_k(\pmb{p}) x^{r - k} = \prod_{i = 1}^{r} (x - p_i) \eqdef Q_{\pmb{p}}(x)\,.
    \end{equation}
    Thus, since the roots and multiplicities must be the same for $Q_{\pmb{q}}$ and $Q_{\pmb{p}}$, one gets $p_i = q_i$ for any $i \in \{1,\ldots, r\}$. But given that $\mathrm{rank}(\pmb{p})$ and $\mathrm{rank}(\pmb{q})$ are both bounded by $r$, this implies $\pmb{p} = \pmb{q}$.
\end{proof}

In particular, the knowledge of the entanglement spectrum is equivalent to the knowledge of sufficiently many power sums. The next proposition then follows.

\begin{prop}
    Let $\pmb{p}, \pmb{q} \in \SpecFun_2$. Consider two finite-dimensional  bipartite pure states $\ket{\psi}$ and  $\ket{\varphi}$, with $\ket{\psi} \simLU \ket{\psi_{\pmb{p}}}$ and $\ket{\varphi} \simLU \ket{\psi_{\pmb{q}}}$. The following equivalence holds:
    \begin{equation}
        \ket{\psi} \simLU \ket{\varphi} \qquad \Longleftrightarrow \qquad \pmb{p}= \pmb{q}  \qquad \left(\,\Longleftrightarrow \, \pmb{p}(\paa{1,2}) = \pmb{q}(\paa{1,2})\,\right) \,,
    \end{equation}
    \ie $\ket\psi$ and $\ket\varphi$ are $\LU$-equivalent if and only if they have the same entanglement spectrum (resp. spectral function).
\end{prop}

\begin{proof}
    The equivalence is immediate given Lem.~\ref{lem:D2LU_eq_1}, Prop.~2.8 of Ref.~\cite{Carrozza:2026qcf}, and Lem.~\ref{lem:powersum}.
\end{proof}

\paragraph{Multipartite case.} While all finite-dimensional bipartite quantum states are SHT states, such a statement does not hold in the multipartite context. Indeed, let us look at a concrete example with the so-called $\ket{\rm{W}}$ state, defined as \cite{Dur:2000zz}:
\begin{equation}
    \ket{\rm{W}} = \frac{1}{\sqrt 3} (\ket{112} + \ket{121} + \ket{211} ) \,.
\end{equation}
Let $C_k(B)$ be the cyclic graphs with $2k$ vertices where the colored edges in $B$ (resp.~in $\bar{B}$) are all parallel. They are, in particular, the colored graphs involved in the Rényi entropies associated with the bipartite state space $\H_B \ot \H_{\bar{B}}$ (see Fig.~\ref{fig:Cyclic}). Furthermore, we consider the third moment of the partial transpose, which is represented by the $3$-colored graph $\PT_{3}$ (see on the left of Fig.~\ref{fig:PT_RM}). The evaluation of the cyclic trace-invariants $\{\Tr_{C_k(\paa{c})}\}_{1\leq c \leq 3, k\in \mathbb{N}^*}$ as well as $\Tr_{\PT_3}$ on the state $\ket{\rm{W}}$ yields:
\begin{equation}
    \forall c \in \paa{1,2,3} \,, \forall{k \in \mathbb{N}^*} \quad \Tr_{C_k(\paa{c})}(\ket{\rm{W}}) = \pa{\frac{1}{3}}^k + \pa{\frac{2}{3}}^k \,, \qquad \rm{and} \qquad \Tr_{\PT_3}(\ket{\rm{W}}) = \frac{2}{9} \,,
\end{equation}
and one observes that no SHT states can reproduce both conditions. Indeed, the condition regarding the colored graphs $C_k(\paa{c})$ is satisfied by an SHT state $\ket{\psi_{\pmb{p}}}$ if and only if $\pmb{p}(\paa{1,2,3}) = (2/3,1/3,0,0, \ldots)$ and $\pmb{p}(B) = (1,0,0, \ldots)$ for $B \neq \paa{1,2,3}$. However, assuming this, the evaluation of $\Tr_{\PT_3}$ on this SHT state yields: $\Tr_{\PT_3}(\ket{\psi_{\pmb{p}}}) = 1/3 \neq 2/9 = \Tr_{\PT_3}(\ket{\rm{W}})$. As a conclusion, we have the following result:
\begin{equation}
    \forall \pmb{p} \in \SpecFun_3\,, \quad \ket{\rm{W}} \underset{\LU}{\not\sim} \ket{\psi_{\pmb{p}}}\,.
\end{equation}
In other words, the W state is \emph{not} an SHT state.

Although SHT states do not cover all multipartite quantum states, they offer a natural extension of HT states. As a matter of fact, the $\LU$-classification of SHT states is similar to that in the bipartite case as provided by the following theorem.

\begin{theo} \label{th:LUSHT}
    Let $D \geq 3$, and $\pmb{p}, \pmb{q} \in \SpecFun_D$. The following equivalence holds:
    \begin{equation}
        \ket{\psi_{\pmb{p}}} \simLU \ket{\psi_{\pmb{q}}} \qquad \Longleftrightarrow \qquad \pmb{p} = \pmb{q} \,.
    \end{equation}
Consequently, two SHT states are $\LU$-equivalent if and only if they have the same spectral function.    
\end{theo}

\begin{proof}
    Since the evaluation of trace-invariants on SHT states (see Eq.~\eqref{eq:Tr_SHT}) only relies on the spectral function, the implication ($\Leftarrow$) is direct. Let us therefore focus on the implication ($\Rightarrow$). Suppose we are given $\pmb{p}, \pmb{q} \in \SpecFun_D$ such that $\ket{\psi_{\pmb{p}}} \simLU \ket{\psi_{\pmb{q}}}$. Equivalently, we have
    \begin{equation}\label{eq:proofLU_graphs-G}
        \tr_G \left( \ket{\psi_{\pmb{p}}} \right) =\tr_G \left( \ket{\psi_{\pmb{q}}} \right)
    \end{equation}
    for any graph $G \in \cG_D$. For each subset of parties $B \subset \{1, \ldots , D\}$ with $\abs{B} \geq 2$, our strategy will consist in applying the previous identity to specially designed families of graphs, which will enable us to prove that $\pmb{p}(B) = \pmb{q}(B)$.
    \begin{figure}
        \centering
        \includegraphics[width = \textwidth]{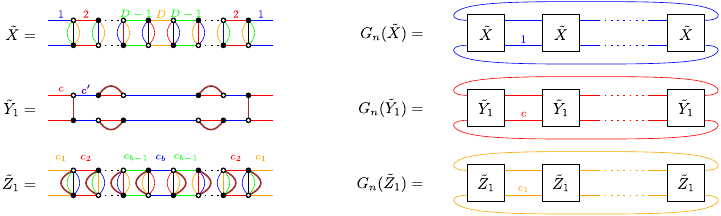}
        \caption{Matrices used in the proof of Thm.~\ref{th:LUSHT}, inspired from the matrices considered in Ref.~\cite{Carrozza:2026qcf}, see Fig.~34. Upper left corner: the open graph $\tilde{X}$, here drawn for $D=5$. Upper right corner: corresponding graph $G_n(\tilde X)$. Middle part: open graph $\tilde{Y}_1$ and its associated colored graph $G_n(\tilde Y_1)$; here, the thick line represents all colors in $\bar{\paa{c,c'}}$. Bottom left corner: open graph $\tilde{Z}_1$, with $B_0 = \paa{c_1,\dots,c_b}$, represented here for $D\geq 6$; the thick line represents the colors in $\bar B$. Bottom right corner: associated graph $G_n(\tilde Z_1)$.}
        \label{fig:proofLU_SHT}
    \end{figure}
    
    We start with the block $1_D = \paa{1,\dots,D}$. Let $\tilde{X}$ denote the open colored graph represented in the top left corner of Fig.~\ref{fig:proofLU_SHT}. For any $n \in \mathbb{N}^*$ we let $G_n(\tilde{X})$ denote the $D$-colored graph obtained by cyclically gluing $n$ copies of $\tilde{X}$, as illustrated in the top right corner of Fig.~\ref{fig:proofLU_SHT}. We can then introduce  
    \begin{equation}
        X(\pmb{p})\eqdef \frac{\tr_{\tilde{X}}\left( \ket{\psi_{\pmb{p}}} \right)}{\tr\left( \tr_{\tilde{X}}\left( \ket{\psi_{\pmb{p}}} \right) \right)} = \frac{\tr_{\tilde{X}}\left( \ket{\psi_{\pmb{p}}} \right)}{\tr_{G_1(\tilde{X})}\left( \ket{\psi_{\pmb{p}}} \right)}\,,
    \end{equation}
    which we view as a Hermitian positive operator of unit trace.\footnote{If $\ket{\psi_{\pmb{p}}} \in \H_1 \otimes \H_2 \otimes \H_3$, we can interpret $\tr_{\tilde{X}}\left( \ket{\psi_{\pmb{p}}} \right)$ as a Hermitian operator on $\H_1 \otimes \H_1^*$; after normalization, the statement follows.} For any $n \in \mathbb{N}^*$, we then have
    \begin{equation}
        \tr\left(X(\pmb{p})^n\right) = \frac{\tr_{G_n (\tilde X)}\left( \ket{\psi_{\pmb{p}}} \right)}{\tr_{G_1 (\tilde X)}\left( \ket{\psi_{\pmb{p}}} \right)^n}\,.
    \end{equation}
    We can evaluate this moment explicitly by means of Eq.~\eqref{eq:Tr_SHT}.  Observing that 1) $k(G_n(\tilde{X}))=2(D-1)n$, and 2) for any block $B \subsetneq 1_D$ (with $\abs{B}\geq 2$), $G_n(\tilde{X})|_B$ consists in the disjoint union of $n$ copies of $G_1(\tilde{X})|_B$, we are led to
    \begin{equation} \label{eq:tr_GtildeX}
        \Tr_{G_n(\tilde{X})}(\ket{\psi_{\pmb{p}}}) = \left( \sum_{i\in \mathbb{N}^*} p_i(1_D)^{2(D-1)n} \right) \cdot \pac{\prod_{B \subsetneq 1_D, \, \abs{B} \geq 2} \prod_{H \in \Comp\left( G_1(\tilde{X})\vert_B \right)} \sum_{j\in \mathbb{N}^*} p_j(B)^{k(H)}}^n \,.
    \end{equation}
    As a result, we have 
    \begin{equation}
        \tr\left(X(\pmb{p})^n\right) = \sum_{i \in \mathbb{N}^*} \pac{C_{1,\pmb{p}} p_i(1_D)}^{2(D-1)n}\,, \qquad \rm{with} \qquad C_{1,\pmb{p}}^{-2(D-1)} \eqdef \sum_{i = 1}^{r_{\pmb{p}} (1_D)} p_i(1_D)^{2(D-1)}\,.
    \end{equation}
    Similarly, we can introduce a Hermitian positive matrix $X(\pmb{q})$ of trace $1$, such that $\tr(X({\pmb{q}})^n) = \frac{\tr_{G_n (\tilde X)}\left( \ket{\psi_{\pmb{q}}} \right)}{\tr_{G_1 (\tilde X)}\left( \ket{\psi_{\pmb{q}}} \right)^n}$ for any $n \geq 1$. Applying \eqref{eq:proofLU_graphs-G} to the family of graphs $\{G_n(\tilde{X})\}_{n \geq 1}$, it follows that:
    \begin{equation}
        \sum_{i \in \mathbb{N}^*} \pac{C_{1,\pmb{p}} p_i(1_D)}^{2(D-1)n} = \tr\left(X(\pmb{p})^n\right) = \tr\left(X(\pmb{q})^n\right) = \sum_{i \in \mathbb{N}^*} \pac{C_{1,\pmb{q}} q_i(1_D)}^{2(D-1)n}\,, 
    \end{equation}
    for any $n \in \mathbb{N}^*$. Given that $X(\pmb{p})$ and $X(\pmb{q})$ are Hermitian, positive, of finite rank and of unit trace, we have $\left( \pac{C_{1,\pmb{p}} p_i(1_D)}^{2(D-1)} \right)_{i \in \mathbb{N}^*}\in \Seq$ and $\left( \pac{C_{1,\pmb{q}} q_i(1_D)}^{2(D-1)} \right)_{i \in \mathbb{N}^*}\in \Seq$. We can therefore apply Lem.~\ref{lem:powersum} to conclude that these two spectra are in fact equal. It follows that  
    \begin{equation}
        C_{1,\pmb{p}} p_i(1_D) = C_{1,\pmb{q}} q_i(1_D) 
    \end{equation}
    for any $i \in \mathbb{N}^*$. By summing this equation over $i$, we obtain in particular
    \begin{equation}
        C_{1,\pmb{p}}= \sum_{i \in \mathbb{N}^*} C_{1, \pmb{p}} p_i(1_D) =  \sum_{i \in \mathbb{N}^*} C_{1, \pmb{q}} q_i(1_D) = C_{1,\pmb{q}} \,.
    \end{equation}
    All in all, we conclude that $p_i(1_D)= q_i(1_D)$ for any $i \in \mathbb{N}^*$ or, in other words, that $\pmb{p}(1_D)= \pmb{q}(1_D)$.
    
    \
    
    For subsets $B \subset \{1,\dots,D\}$ with $2 < \abs{B} \leq D$, we proceed by downward induction on the size of $B$. The base case, \ie the equality between the entanglement spectra $\pmb{p}(1_D)$ and $\pmb{q}(1_D)$, follows directly from the previous paragraph.\\
    Let $B_0$ be a subset of $\{1,\dots,D\}$ such that $2 < \abs{B_0} < D$. By the induction hypothesis, we assume that the relation $\pmb{p}(B) = \pmb{q}(B)$ holds for all subsets $B$ satisfying $|B_0| < |B| \leq D$.
    We let $Z_1 (\pmb{p})$ be the Hermitian positive operator of unit trace defined by:
    \begin{equation}
        Z_1(\pmb{p}) \eqdef \frac{\tr_{\tilde{Z}_1}\left( \ket{\psi_{\pmb{p}}} \right)}{\tr\left( \tr_{\tilde{Z}_1}\left( \ket{\psi_{\pmb{p}}} \right) \right)} = \frac{\tr_{\tilde{Z}_1}\left( \ket{\psi_{\pmb{p}}} \right)}{\tr_{G_1(\tilde{Z}_1)}\left( \ket{\psi_{\pmb{p}}} \right)}\,,
    \end{equation}
    where the graphs $\tilde Z_1$ and $G_n (\tilde{Z}_1)$ are represented in the bottom part of Fig.~\ref{fig:proofLU_SHT}.
    A tedious but straightforward application of Eq.~\eqref{eq:Tr_SHT} shows that: for any $n \in \mathbb{N}^*$,
    \begin{align}
      \tr_{G_n (\tilde Z_1)}\left( \ket{\psi_{\pmb{p}}} \right) =& \pac{\sum_{i\in \mathbb{N}^*} p_i(B_0)^{2(\abs{B_0}-1)n}}  \cdot \prod_{B \supsetneq B_0} \sum_{j\in \mathbb{N}^*} p_j(B)^{2(\abs{B_0}-1)n} \nonumber \\ 
        &\cdot\pac{\prod_{\substack{B \subset \paa{1,\dots,D} \\ B \not\supset B_0, \, B\cap B_0 \neq \emptyset \\ \abs{B}\geq 2}} \,\, \prod_{H \in \Comp \left( G_1(\tilde Z_1)\vert_B \right) } \sum_{\ell\in \mathbb{N}^*} p_\ell(B)^{k(H)}}^n\,. 
    \end{align}
    Upon normalization, any term  appearing in the second line of the previous equation is cancelled, leading to
    \begin{equation}\label{eq:Z1p}
    \tr\left( Z_1 (\pmb p)^n \right) = \frac{\tr_{G_n (\tilde Z_1)}\left( \ket{\psi_{\pmb{p}}} \right)}{\tr_{G_1 (\tilde Z_1)}\left( \ket{\psi_{\pmb{p}}} \right)^n} =  \pac{\sum_{i \in \mathbb{N}^*} \pac{C_{3,\pmb{p}} p_i(B_0)}^{2(\abs{B_0}-1)n}} \cdot \prod_{B \supsetneq B_0}\pac{\sum_{j \in \mathbb{N}^*} \pa{ \frac{ p_j(B)^{2(\abs{B_0}-1)}}{\sum_{\ell \in \mathbb{N}^*} p_\ell (B)^{2(\abs{B_0}-1)}}}^n}\,,
    \end{equation}
    where $C_{3,\pmb{p}}^{2 - 2 \abs{B_0}} \eqdef \displaystyle \sum_{i \geq 1} p_i(B_0)^{2(\abs{B_0}-1)}$. Proceeding similarly for the spectral function $\pmb{q}$, we can invoke the identity $\tr\left( Z_1 (\pmb p)^n \right) = \tr\left( Z_1 (\pmb q)^n \right)$ together with the induction hypothesis to conclude that
    \begin{equation}
        \sum_{i \in \mathbb{N}^*} \pac{C_{3,\pmb{p}} p_i(B_0)}^{2(\abs{B_0}-1)n} = \sum_{i \in \mathbb{N}^*} \pac{C_{3,\pmb{q}} q_i(B_0)}^{2(\abs{B_0}-1)n}\,,
    \end{equation}
    for any $n \in \mathbb{N}^*$ (indeed, the last term of Eq.~\eqref{eq:Z1p} only involves blocks of size strictly greater than $\abs{B_0}$, and thus, by the induction hypothesis, it cancels out with a similar term in $\tr\left( Z_1 (\pmb q)^n \right)$). Finally, observing that $\left( \pac{C_{3,\pmb{p}} p_i(1_D)}^{2(\abs{B_0}-1)} \right)_{i \in \mathbb{N}^*}$ and $\left( \pac{C_{3,\pmb{q}} q_i(1_D)}^{2(\abs{B_0}-1)} \right)_{i \in \mathbb{N}^*}$ are elements of $\Seq$, we conclude by application of Lem.~\ref{lem:powersum} that they are equal. Arguing as in the first step of the proof (that was concerned with the block $1_D$), we conclude that $\pmb{p}(B_0) = \pmb{q}(B_0)$.   
    
    \
    
    Finally, we finish the proof with blocks of size $2$. We let $c,c' \in \paa{1,\dots,D}$ be two distinct integers, and focus on a block $\paa{c,c'}$. Introducing the graphs $\tilde Y_1$ and $G_n(\tilde{Y}_1)$ represented in the middle part of Fig.~\ref{fig:proofLU_SHT}, we let $Y_1 (\pmb{p})$ denote the Hermitian positive operator of  unit trace defined by
    \begin{equation} \label{eq:Y}
          Y_1(\pmb{p}) \eqdef \frac{\tr_{\tilde{Y}_1}\left( \ket{\psi_{\pmb{p}}} \right)}{\tr\left( \tr_{\tilde{Y}_1}\left( \ket{\psi_{\pmb{p}}} \right) \right)} = \frac{\tr_{\tilde{Y}_1}\left( \ket{\psi_{\pmb{p}}} \right)}{\tr_{G_1(\tilde{Y}_1)}\left( \ket{\psi_{\pmb{p}}} \right)}\,,
    \end{equation}
    and similarly for $\pmb{q}$. The moments of this matrix can be evaluated by means of the following relation, which is a consequence of \eqref{eq:Tr_SHT}: for any $n\geq 1$, 
    \begin{align} \label{eq:Tr_GYtilde}
        \Tr_{G_n(\tilde Y_1)}(\ket{\psi_{\pmb{p}}}) =& \pac{\sum_{k\geq1} p_k(\paa{i,j})^{3n}}^{2} \cdot \prod_{B \supsetneq \paa{i,j}} \sum_{k\geq 1} p_k(B)^{6n} \nonumber \\ 
        & \cdot \pac{\prod_{\substack{B \in \paa{1,\dots,D} \\
        i \in B, \, j \notin B \\ \abs{B} \geq 2}} \pa{\sum_{k\geq 1} p_k(B)^2} \prod_{\substack{B \in \paa{1,\dots,D} \\
        i \notin B, \, j \in B \\ \abs{B} \geq 2}} \pa{\sum_{k\geq 1} p_k(B)^6} }^n \,.
    \end{align}
    This, in turn, yields: for any $n \geq 1$
    \begin{equation} \label{eq:ij_LU}
        \Tr\left( Y_1 (\pmb{p})^n \right) = \pac{\sum_{l \geq 1} \pac{C_{2,\pmb{p}} p_l(\paa{i,j})}^{3n}}^2 \cdot \prod_{B \supsetneq \paa{i,j}} \pac{\sum_{l\geq 1} \pa{\frac{p_l(B)^{6}}{\sum_{k\geq 1} p_k(B)^{6}}}^n} \,,
    \end{equation}
    where $C_{2,\pmb{p}}^{-3} = \displaystyle \sum_{l = 1}^{r_{p}(\paa{i,j})} p_l(\paa{i,j})^{3}$, and similarly for $\pmb{p}$. Equating $\Tr\left( Y_1 (\pmb{p})^n \right)$ and $\Tr\left( Y_1 (\pmb{p})^n \right)$, the last term of \eqref{eq:ij_LU} (which only involves blocks $B$ with $\abs{B}\geq 2$) cancels out with an analogous term in $\Tr\left( Y_1 (\pmb{p})^n \right)$. It follows the power sums of $([C_{2,\pmb{p}} p_l(\paa{i,j})]^3)_{l \in \mathbb{N}^*}$ and $([C_{2,\pmb{q}} q_l(\paa{i,j})]^3)_{l \in \mathbb{N}^*}$ are identical. Since both sequences are in $\Seq$, we can apply Lem.~\ref{lem:powersum} to conclude that they are identical. Arguing as before, this finally implies $\pmb{p}(\{i,j\})= \pmb{q}(\{i,j\})$. 
\end{proof}

\begin{rem} 
    A comment on our choice of matrices $X$, $Y_1$, and $Z_1$ is in order. While many other trace-invariants could have been chosen for the preceding proof, our specific choice is motivated by their connection to the matrices used in Ref.~\cite{Carrozza:2026qcf} and App.~\ref{A:proofTechLem}  (see Fig.~\ref{fig:proof_integ}), namely $X$, $\{Y_m\}_{m\geq 1}$, and $\{Z_m\}_{m\geq 1}$. Moreover, the invariants $\tr_{G_n(\tilde{X})}$ and $\tr_{G_n(\tilde{Z}_1)}$ are closely related to so-called joint realignment moments, whereas $\tr_{G_n(\tilde{Y}_1)}$ can be obtained by insertion of melonic decorations on the graph of a realignment moment. Said differently, the family of realignment moments graphs could have been used instead of the family $\{G_n(\tilde{Y}_1)\}_{n\geq 1}$, with equal success.
\end{rem}

\subsection{Entanglement spread for spectral hypergraph-tensor states}

Let us recall that the colored graphs satisfying the property that $\kappa(G\vert_B) = 1$ for all subsets $B \subset \paa{1,\dots,D}$ with $\abs{B} \geq 2$ are known as \textit{maximally single-trace} (MST). Furthermore, the set of $D$-colored MST graphs was denoted in Sec.~\ref{sec:General} by $\cG_D^{\rm{(MST)}}$. Thanks to the following proposition, a characterization of the spectral function of an SHT state can be obtained by an entanglement spread constructed by the use of MST graphs. One could, furthermore, observe that this characterization is a natural generalization of the characterization of the entanglement spectrum of a bipartite quantum state provided by the entanglement spread $\mathsf{D}_{m,n}$. 

\begin{prop}\label{prop:def_Delta_G}
    Let $G \in \cG_D^{\mathrm{(MST)}}$ such that $k(G)>1$, and $v$ a marked vertex of $G$. For any finite-dimensional $D$-partite pure state $\ket{\psi}$, we define
    \begin{equation} \label{eq:def_spreadMST}
        \mathsf{D}_G^{\mathrm{(MST)}}(\ket{\psi}) \eqdef \begin{cases} \frac{1}{1-k(G)} \ln \abs{\Tr_G(\ket{\psi})} - \frac{1}{2} \cdot \frac{1}{1-k(G)} \ln \abs{\Tr_{G \op_{(v,\bar v)} \bar{G}}(\ket{\psi})} &\quad \mathrm{if}\; \tr_G(\ket\psi)\neq 0 \\
        +\infty &\quad \mathrm{otherwise.}
        \end{cases}
    \end{equation}
    $\mathsf{D}_G^{\mathrm{(MST)}}$ is a multipartite entanglement spread (in the sense of Def.~\ref{def:entanglement_spread}); in particular, it takes value in $[0,+\infty]$ and defines a $\LO$-monotone. 
\end{prop}
\begin{proof}
    The entanglement spread $\mathsf{D}_G^{\mathrm{(MST)}}$ can be written in a form that makes its positivity and $\LO$-monotonicity on pure states explicit. Indeed, one observes that
    \begin{equation}
        \mathsf{D}_G^{\mathrm{(MST)}} = \frac{1}{2} \frac{1}{k(G) - 1} C_G^{[v]} = \frac{1}{2} \frac{1}{k(G) - 1} R_{G \sqcup \bar{G}, G \op_{(v,\bar v)} \bar{G}}\,,
    \end{equation}
    with $C_G^{[v]}$ a $\LO$-monotone taking value in $[0,+\infty]$ thanks to Prop.~\ref{prop:Spread_contractionOp}. 
\end{proof}    

The next proposition shows that, among SHT states, HT states are exactly those that minimize the entanglement spread $\mathsf{D}_G^{\mathrm{(MST)}}$ associated with \emph{any given} (non-trivial) MST graph. Furthermore, the latter characterization is equivalent to the fact that $\mathsf{D}_H^{\mathrm{(MST)}}$ is minimized for \emph{every} (non-trivial) MST graph $H$, a condition we can express directly in terms of the evaluation of the invariant $\tr_H$.
\begin{prop} \label{prop:ES_SHT}
    Let $D \geq 3$, $\pmb{p}\in \SpecFun_D$, and $G\in \cG_D^{\mathrm{(MST)}}$ such that $k(G)>1$. The following properties are equivalent:
    \begin{enumerate}
        \item $\mathsf{D}_G^{\mathrm{(MST)}}\pa{\ket{\psi_{\pmb{p}}}} = 0$;
        \item $\ket{\psi_{\pmb{p}}}$ is a HT state;
        \item there exists $N \in \mathbb{N}^*$ such that: $\quad \forall H \in \cG_D^{\mathrm{(MST)}}\,, \; \tr_{H}\left( \ket{\psi_{\pmb{p}}} \right)= N^{1- k(H)}$\,.
    \end{enumerate}
In particular, $\mathsf{D}_G^{\mathrm{(MST)}}$ is minimized by an SHT state if and only if the  entanglement spectrum of the latter is flat on each block of the hypergraph.
\end{prop}

\begin{proof}
    Let $\pmb{p}\in \SpecFun_D$. We will show that $1 \Rightarrow 2 \Rightarrow 3 \Rightarrow 1$. The main difficulty of the proof consists in establishing the first of these implications.
    
    \medskip
    
    We first assume that $\mathsf{D}_G^{\mathrm{(MST)}}\pa{\ket{\psi_{\pmb{p}}}} = 0$. Let us prove that the second property holds. The graph $G$ being maximally single-trace, a short computation gives: 
    \begin{equation} \label{eq:Tr_G_MST}
        \Tr_G\pa{\ket{\psi_{\pmb{p}}}} = \prod_{\substack{B \in \paa{1,\dots,D}\\ \abs{B} \geq 2}} \,\sum_{i \in \mathbb{N}^*} p_i(B)^{k(G)}\,.
    \end{equation}
    For short-hand, let us denote by $\pos_D \eqdef \{B \subset \{1, \ldots , D\}\,\vert \, \abs{B}\geq 2\}$ the set of blocks $B$ appearing in this equation, which has cardinal $\abs{\pos_D}=2^D-D-1$. We can exchange the sum and product appearing in \eqref{eq:Tr_G_MST} to obtain a sum indexed by multi-indices (or maps) $i=(i_B)_{B \in \pos_D}\in (\mathbb{N}^*)^{\pos_D}$, namely:
    \begin{equation}
        \Tr_G\pa{\ket{\psi_{\pmb{p}}}} = \sum_{i \in (\mathbb{N}^*)^{\pos_D}} \prod_{B \in \pos_D} p_{i_B}(B)^{k(G)}\,.
    \end{equation}
    Defining 
    \begin{equation}
        P_i \eqdef \prod_{B \in \pos_D} p_{i_B}(B) 
    \end{equation}
    for any $i \in (\mathbb{N}^*)^{\pos_D}$, we note that: (a) $P_i \geq 0$, (b) $P_i > 0$ for finitely-many values of the multi-index $i$, and (c) the sequence $(P_i)_{i \in (\mathbb{N}^*)^{\pos_D}}$ is normalized in the sense that
    \begin{equation}
        \sum_{i \in (\mathbb{N}^*)^{\pos_D}} P_i = \prod_{B \in \pos_D} \sum_{i\in \mathbb{N}^*} p_i(B)= 1\,.
    \end{equation}
    As a result, there exists a one-to-one map between $(\mathbb{N}^*)^{\pos_D}$ and $\mathbb{N}^*$ which allows us to identify  $(P_i)_{i \in (\mathbb{N}^*)^{\pos_D}}$ to some \emph{spectrum} $\pmb{P}= (P_i)_{i\in \mathbb{N}^*}\in \Seq$. Let $\ket{\psi_{\pmb{P}}}$ denote the \emph{bipartite} HT state with entanglement spectrum $\pmb{P}$. We then note that 
    \begin{align}
            0 = \mathsf{D}_G^{(\rm{MST})} (\ket{\psi_{\pmb{p}}})&= \frac{1}{1 - k(G)} \ln \abs{\Tr_G(\ket{\psi_{\pmb{p}}})} - \frac{1}{2 - 2k(G)} \ln \abs{\Tr_{G \op_{(v , \bar v)} \bar G}(\ket{\psi_{\pmb{p}}})} \,, \label{eq:2}\\
             &= \frac{1}{1 - k(G)} \ln \left(\sum_{i\in \mathbb{N}^*} (P_i)^{k(G)}\right) - \frac{1}{2 - 2k(G)} \ln\left( \sum_{i\in \mathbb{N}^*} (P_i)^{2k(G)-1}\right) \,, \\
             &= \mathsf{D}_{k(G),2k(G)-1}\pa{\ket{\psi_{\pmb{P}}}} = \mathsf{D}_{m,n}\pa{\ket{\psi_{\pmb{P}}}} \,,
        \end{align}
    where we have used the fact that $k(G \op_{(v,\bar v)} \bar{G}) = 2 k(G) - 1$ and defined $m\eqdef k(G)$, $n\eqdef 2 k(G) - 1$. Since we assumed $k(G)>1$, we have $m>n$ and thus $\mathsf{D}_{m,n}$ defines a bipartite entanglement spread. As a result, it follows from $\mathsf{D}_{m,n}(\ket{\psi_{\pmb{P}}})=0$ that $\ket{\psi_{\pmb{P}}}$ has a flat entanglement spectrum $\pmb{P}$. 
    
    To conclude, let us prove that this last property in turn implies that $\pmb{p}(B)$ is flat for any $B\in \pos_D$. We first notice that, $\pmb{P}$ being flat, there exists some constant $\rho>0$ such that: for any $i \in (\mathbb{N}^*)^{\pos_D}$, either $P_i=\rho$ or $P_i = 0$. Moreover, denoting by $1\eqdef (1)_{B\in \pos_D}$ the constant multi-index (or map) equal to $1$, it is clear that $P_{1}\neq 0$ and thus $P_{1}=\rho$. Let us fix $B \in \pos_D$ and prove that, likewise, there exists $\rho_B>0$ such that: for any $j \in \mathbb{N}^*$, $p_{j}(B)= \rho_B$ or $p_{j}(B)=0$. To see this, let us fix $j\in \mathbb{N}^*$ and introduce the multi-index $i=(i_C)_{C\in \pos_D}$ such that: $i_B =j$, and $i_C =1$ for any $C\in \pos_D\setminus\{ B\}$. If $P_i = \rho$, this equation can be explicitly written as $p_j(B)\prod_{C\neq B}p_1(C)=p_1(B)\prod_{C\neq B}p_1(C)$, which leads to $p_j(B)=p_1(B)$. On the other hand, if $P_i = 0$, then $p_j(B)\prod_{C\neq B}p_1(C)=0$, which implies $p_j(B)=0$. As a result, $\pmb{p}(B)$ defines a flat entanglement spectrum, with non-zero eigenvalue $\rho_B \eqdef p_1(B)$. 

    This concludes the proof that $1\Rightarrow 2$.

    \medskip
    
    Next, let us prove that the second property implies the third one. If $\ket{\psi_{\pmb{p}}}$ is a HT state, we can find a weight function $\alpha$ such that $\ket{\psi_{\pmb{p}}}\simLU \ket{\psi_\alpha}$. Introducing 
     \begin{equation}
         N\eqdef \prod_{B \in \paa{1,\dots,D}\,, \, \abs{B} \geq 2} \alpha(B)\in \mathbb{N}^*\,,
     \end{equation}
     we then find that
        \begin{equation}
            \forall H \in \cG_D^{\rm{(MST)}}\,, \quad \Tr_H(\ket{\psi_{\pmb{p}}}) = \Tr_H(\ket{\psi_\alpha}) = \prod_{B \in \paa{1,\dots,D}\,, \, \abs{B} \geq 2} \alpha(B)^{1 - k(H)} = N^{1-k(H)}\,.
        \end{equation}
    
    \medskip
    
    Finally, the third property implies the first one. Indeed, suppose there exists $N \in \bb{N}^*$ such that, for any MST graph $H$, $\Tr_H(\ket{\psi_{\pmb{p}}}) = N^{1-k(H)}$. Then, since $G$ is itself a MST graph, one has:
        \begin{equation}
            \mathsf{D}_G^{\rm{(MST)}}(\ket{\psi_{\pmb{p}}}) = \frac{1}{1-k(G)}\ln N^{1-k(G)} - \frac{1}{2 - 2k(G)} \ln N^{1-(2k(G)-1)} = 0\,.
        \end{equation}
    
    \medskip
    
    This concludes the proof.
\end{proof}

\section{Outlook on random tensor networks}

In this paper, the random distribution of $D$-partite pure states we focused on when evaluating expectation values of multipartite entanglement spreads of the form $\langle R_{G, H}\rangle$ was the Haar-random state $\ket{\varphi_N}$ on $\left(\mathbb{C}^N\right)^{\otimes D}$. Going beyond this first natural example, it would be interesting to compute these expectation values for arbitrary \textit{random tensor networks}. We recall that a random tensor network is a random quantum state associated with a (fixed) graph or network, which is obtained by assigning independent random tensors (such as Haar or Gaussian states) with \textit{bond dimension $N$} to vertices and contracting their indices according to the geometry of the network. They are particularly relevant in the context of holographic dualities, quantum error-correcting codes, \lat{etc.} (see \eg Refs.~\cite{Akers:2024pgq,Akers:2022zxr,Akers:2021pvd,Hayden:2016cfa,Cheng:2022ori,Pastawski:2015qua,Hu:2026bhg}). Hence, studying random tensor networks could pave the way for constructing entanglement spreads whose expectation values are related to geometric constraints of the network, thereby leading to a characterization of the network itself. \\
As a concrete example, we present a tripartite entanglement spread constructed from the realignment moments and the even moments of the partial transpose. Let $n \in \bb{N}^*$ be an integer and $c\in \{1,2,3\}$, and consider the following Cauchy-Schwarz inequality (remember that, in these graphical equations, our convention is that $\ket\psi = \vcenter{\hbox{\includegraphics[scale=.5]{images/psi.pdf}}}$):
\begin{equation}
    \abs{\Tr_{\PT_{2n}^{(c)}}(\ket{\psi})}^2 = \abs{\vcenter{\hbox{\includegraphics[scale=1]{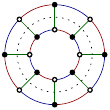}}}}^2 \leq \vcenter{\hbox{\includegraphics[scale=1]{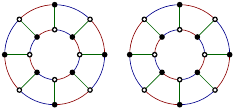}}} = \pa{\Tr_{\RM_{2n}^{(c)}}(\ket{\psi})}^2\,.
\end{equation}
So, we have that $\PT_{2n}^{(c)}$ is dominated by $\RM_{2n}^{(c)}$. These trace-invariants have been extensively studied in the context of random tensor networks, particularly for $n=2$ (see Refs.~\cite{Dong:2021clv,Penington:2022dhr}). Indeed, consider a tensor network whose boundary vertices are divided into three groups, one for each subsystem. The relevant geometric quantities for computing the expectation values are the so-called \textit{min-cuts} $\gamma_1$, $\gamma_2$ and $\gamma_3$ and the \textit{minimal tripartition} $\tau_{123}$ that we divide in three branches $(t_{12},t_{13},t_{23})$, \ie $t_{12}$ (resp.~$t_{13}$ and $t_{23}$) is the number of cuts in the minimal tripartition $\tau_{123}$ between the parts $1$ and $2$ (resp.~parts $1$ and $3$ and parts $2$ and $3$). A generic tripartite tensor network is presented on the left-hand side of Fig.~\ref{fig:RTN_trip}.

\begin{figure}
    \centering
    \includegraphics[width = .6\textwidth]{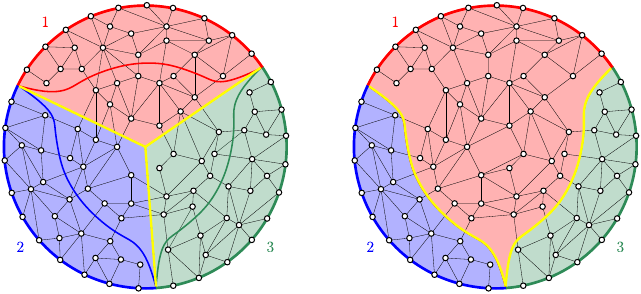}
    \caption{Left: example of a tensor network in which the min-cuts are represented by the blue, red, and green thick lines, while the minimal tripartition is represented by the yellow thick line. We count $\gamma_1 = 9$, $\gamma_2 = 7$, $\gamma_3 = 8$, $t_{12} = 5$, $t_{13} = 5$, $t_{23} = 4$ and $\tau_{123} = t_{12} + t_{13} + t_{23} = 14$. Right: Configuration of tensor networks for which $\mean{R_{\PT_4,\RM_4}^{(\rm{sym})}(\ket{\Psi_N})} \underset{N \to \infty}{=} O(1)$.}
    \label{fig:RTN_trip}
\end{figure}

Let $(\ket{\Psi_N})_{N \in \bb{N}^*}$ be a sequence of tripartite random tensor network states with bond dimension $N$ associated with the same fixed network, as well as the same tripartite subdivision of the boundary vertices of the network. It was shown in Refs.~\cite{Penington:2022dhr,Dong:2021clv} that both $\Tr_{\RM_4}$ and $\Tr_{\PT_4^{(c)}}$ obey the large-$N$ factorization criterion in the random pure state $\ket{\Psi_N}$ ,\footnote{The reader can consult Sec.~6.3.1 of \cite{Carrozza:2026qcf} for a precise definition of this criterion in this context, as well as a discussion of its implications.} and as a result, the distribution of the random variable $\ln \Tr_{\RM_4} \left( \ket{\Psi_N} \right)$ (resp.~$\ln \Tr_{\PT_4^{(c)}} \left( \ket{\Psi_N} \right)$) concentrate around a typical value in the asymptotic regime. Furthermore, it was recently shown in \cite{Carrozza:2026qcf} that those trace-invariants also follow a power-law, and thus, the expectation value $\mean{\ln \Tr_{\RM_4} \left( \ket{\Psi_N} \right)}$ (resp.~$\langle\ln \Tr_{\PT_4^{(c)}} \left( \ket{\Psi_N} \right)\rangle$) is asymptotically equal to the previously computed typical value (\ie no large deviations contribute significantly to the mean). Concretely, one finds in the large bond dimension limit that
\begin{equation}    
    \mean{\ln\Tr_{\RM_4}(\ket{\Psi_N})} \underset{N \to \infty}{=} - 2 \tau_{123} \ln (N) + O(1) 
\end{equation} 
and
\begin{equation}    
    \mean{\ln\Tr_{\PT_4^{(c)}}(\ket{\Psi_N})} \underset{N \to \infty}{=} - \pa{2\gamma_c + \sum_{i \neq c} \gamma_i} \ln (N) + O(1) \,.
\end{equation} 
Hence, one can define the ``symmetric'' multipartite entanglement spread 
\begin{equation}
    R_{\PT_4,\RM_4}^{(\rm{sym})} \eqdef \frac{1}{6} \cdot \sum_{c=1}^3 R_{\PT_4^{(c)},\RM_4}\,,
\end{equation}
which, in the large bond dimension limit, concentrates around the following expectation value:
\begin{equation}\label{eq:R_PT_RM_RTN}
    \mean{R_{\PT_4,\RM_4}^{(\rm{sym})}(\ket{\Psi_N})} \underset{N \to \infty}{=} \pa{ \frac{2}{3}(\gamma_1 + \gamma_2 + \gamma_3) - \tau_{123}} \ln (N) + O(1) \,.
\end{equation}

\begin{rem} \label{rem:RTN}
    Albeit not straightforward when expressed in this form, the leading order of Eq.~\eqref{eq:R_PT_RM_RTN} is non-negative (as required, since $R_{\PT_4,\RM_4}^{(\rm{sym})}$ is positive by construction). Indeed, the minimal tripartition being inside the region defined by the min-cuts (if not, the $\gamma_c$'s would not be minimal, thus they would not be min-cuts), we have, for all triplets of distinct integers $i,j,k \in \paa{1,2,3}$, the following inequalities:
    \begin{equation}
        \gamma_i + \gamma_j \geq \tau_{123} \,, \quad \rm{ and } \quad \gamma_i \leq t_{ij} + t_{ik} \,.
    \end{equation}
    Combining these inequalities leads to the bounds:
    \begin{equation} \label{eq:bound_RTN}
        \frac{1}{2} \pa{\gamma_1 + \gamma_2 + \gamma_3} \leq \tau_{123} \leq \frac{2}{3} \pa{\gamma_1 + \gamma_2 + \gamma_3} \,. 
    \end{equation}
    Thus proving the non-negativity of the dominant term of Eq.~\eqref{eq:R_PT_RM_RTN}.
\end{rem}

$R_{\PT_4,\RM_4}^{(\rm{sym})}$ being a sum of entanglement spreads, which are all individually positive, if $\langle R_{\PT_4,\RM_4}^{(\rm{sym})}(\ket{\Psi_N})\rangle \underset{N \to \infty}{=} O(1)$ then $\langle R_{\PT_4^{(c)},\RM_4}(\ket{\Psi_N})\rangle \underset{N \to \infty}{=} O(1)$ for any $c \in \paa{1,2,3}$. This implies, in particular, that:
\begin{equation}
    \gamma \eqdef \gamma_1 = \gamma_2 = \gamma_3 \qquad \rm{and} \qquad \tau_{123} = 2 \gamma \,.
\end{equation}
Said differently, the minimal tripartition glues to either one of three pairs of two min-cuts $(\gamma_1,\gamma_2)$, $(\gamma_1,\gamma_3)$ and $(\gamma_2, \gamma_3)$ as illustrated on the right-hand side of Fig.~\ref{fig:RTN_trip}.

\begin{rem}
    The condition $\langle R_{\PT_4^{(3)},\RM_4}(\ket{\Psi_N})\rangle \underset{N \to \infty}{=} O(1)$ partially constrains the geometry of the network (and thus, cannot be achieved on an arbitrary network). Indeed, this condition implies that: 
    \begin{equation}
        (\gamma_1 + \gamma_3) + (\gamma_2 + \gamma_3) = 2 \tau_{123}\,.
    \end{equation}
    However, from Rem.~\ref{rem:RTN}, we have $\gamma_1 + \gamma_3 \geq \tau_{123}$ and $\gamma_2 + \gamma_3 \geq \tau_{123}$, thus leading to $\gamma_1 + \gamma_3 = \gamma_2 + \gamma_3 = \tau_{123}$. The latter relation implies that the minimal tripartition glues to either one of the min-cuts $(\gamma_1,\gamma_3)$ and $(\gamma_2,\gamma_3)$. Said differently, we have 
    \begin{equation}
        \tau_{123} = t_{12} + t_{13} + t_{23} = \gamma_1 + 0 + \gamma_3 = \gamma_2 + \gamma_3 + 0  \,.
    \end{equation}
    Hence, since $\gamma_1 = \gamma_2 = \tau_{123} - \gamma_3$, we define $\gamma \eqdef \gamma_1 = \gamma_2$. Furthermore, the upper bound on $\tau_{123}$ presented in Eq.~\eqref{eq:bound_RTN} can be rewritten as: $\tau_{123} \leq 2 \gamma$. \\
    Assume, for definiteness, that the minimal tripartition is gluing to $(\gamma_2,\gamma_3)$, \ie $t_{12} = \gamma$, $t_{13} = \gamma_3$ and $t_{23} = 0$. Then, the inequalities relating the min-cuts and the branches of the minimal tripartition -- presented in Rem.~\ref{rem:RTN} -- lead to:
    \begin{equation}
        \gamma \leq \tau_{123} \leq 2 \gamma\,,
    \end{equation}
    where the upper bound was obtained earlier. In other words, $0 \leq \gamma_3 \leq \gamma$.
\end{rem}

\paragraph{Statement on AI usage.} This paper was written by humans for a human readership; no LLM was involved at any stage of this work.

\paragraph{Acknowledgements.} This work was supported by the  ANR JCJC project ``RTFPQuEnt" (ANR-25-CE40-5465). The IMB receives support from the EIPHI Graduate School (contract ANR-17-EURE-0002). L.L. also acknowledges support  from the  ANR PRC project  ``TAGADA" (ANR-25-CE40-5672).

\appendix

\section{Proof of Lem.~\ref{lem:alpha_integer_valued}} \label{A:proofTechLem}

We begin by restating the lemma.

\begin{lem}
    Let $\ket\psi$ a finite-dimensional $D$-partite pure state. Suppose there exists a function $\alpha : \{B \subset\{1,\dots,D\}\,\vert \, \abs{B}\geq 2\}\to \mathbb{R}_+^*$ such that:
    \begin{equation}
        \forall G \in \cG_D^{\conn}\,, \quad \tr_G(\ket{\psi})= \prod_{B \subset\{1,\dots,D\}\,, \,\abs{B}\geq 2} \alpha(B)^{\kappa(G\vert_B) - k(G)}\,. 
    \end{equation}
    Then, the function $\alpha$ necessarily takes value in $\mathbb{N}^*$, and identifying it to a weight function $\alpha  : \{B \subset\{1,\dots,D\}\,\vert \, \abs{B}\geq 2\}\to\mathbb{N}^*$, we have:
    \begin{equation}
        \ket\psi \simLU \ket{\psi_\alpha}\,.
    \end{equation}
\end{lem}

The proof of this result relies on the following property of finite-dimensional Hermitian matrices.

\begin{lem}\label{lem:Resolvent_pap}
    Let $A$ be a finite-dimensional Hermitian matrix and $a \in \mathbb{R}_+^*$ a non-negative real number. If
    \begin{equation}\label{eq:constraint_moments}
        \forall n \in \mathbb{N}^*\,, \qquad \tr\left( A^n\right) = a^{1-n}\,,
    \end{equation}
    then $a \in \mathbb{N}^*$.
\end{lem}

While the proof is provided in Ref.~\cite{Carrozza:2026qcf} (see Lem.~5.4), we briefly summarize it here for completeness.

\begin{proof}
    The resolvent $R: z \mapsto \pa{z\id - A}^{-1}$ can be defined by a convergent series for any $z \in \bb{C}$ with $\displaystyle\abs{z} > \max_{\lambda \in \spec(A)} \abs{\lambda}$ (thanks to the spectral theorem). For any such $z$, one can then prove, as in Sec.~\ref{sec:bipartite_charac}, that $\tr\left(AR(z)\right)= \frac{1}{z-a^{-1}}$. It follows that $a^{-1}$ is the unique non-zero eigenvalue of $A$, and that it has multiplicity $a$ (which is thus an integer).
\end{proof}

With these preliminaries in place, we are now ready to prove Lem.~\ref{lem:alpha_integer_valued}.

\begin{proof} 
    Let $\ket\psi$ a finite-dimensional $D$-partite pure state, and $\alpha : \{B \subset\{1,\dots,D\}\,\vert \, \abs{B}\geq 2\}\to \mathbb{R}_+^*$ satisfying:
    \begin{equation}
        \forall G \in \cG_D^{\conn}\,, \quad \tr_G(\ket{\psi})= \prod_{B \subset\{1,\dots,D\}\,, \,\abs{B}\geq 2} \alpha(B)^{\kappa(G\vert_B) - k(G)}\,. 
    \end{equation}
    Let $B_0 \subset \{1,\ldots, D\}$ with $\abs{B_0}\geq 2$. We will prove that $\alpha(B_0)\in \mathbb{N}^*$ by downward induction on $\abs{B_0}\in\{2, \ldots, D\}$.
    
    \medskip
    
    If $\abs{B_0}=D$, then $B_0 = \{1, \ldots , D\} \eqdef 1_D$; let us therefore prove that $\alpha(1_D) \in \bb{N}^*$. Let $\tilde X$ denote the open $D$-colored graph represented in Fig.~\ref{fig:proof_integ} (as well as Fig.~\ref{fig:proofLU_SHT}), and for any $n\in \mathbb{N}^*$, let $G_n(\tilde X)$ denote the (closed) $D$-colored graphs obtained by cyclic concatenation of $n$ copies of $\tilde{X}$ (see Fig.~\ref{fig:proofLU_SHT}). With this data at hand, we introduce a Hermitian positive operator $X$ of unit trace, defined by:
        \begin{equation}
    X\eqdef \frac{\tr_{\tilde{X}}\left( \ket{\psi} \right)}{\tr\left( \tr_{\tilde{X}}\left( \ket{\psi} \right) \right)} = \frac{\tr_{\tilde{X}}\left( \ket{\psi} \right)}{\tr_{G_1(\tilde{X})}\left( \ket{\psi} \right)}\,.
    \end{equation}
For any $n\in \mathbb{N}^*$, the $n^{\rm th}$ moment of $X$ can be expressed in terms of trace-invariants as     
    \begin{equation}
        \tr\left(X^n\right) = \frac{\tr_{G_n (\tilde X)}\left( \ket{\psi} \right)}{\tr_{G_1 (\tilde X)}\left( \ket{\psi} \right)^n}\,.
    \end{equation}
Noting that, for any $B \subset \paa{1,\dots,D} \backslash 1_D$,  $\kappa(G_n(\tilde X) \vert_B) \, \propto \, n$, most contributions cancel out in this quotient, leading to:
    \begin{equation}
       \tr\left(X^n\right) = \alpha(1_D)^{1-n} \,.
    \end{equation}
    The matrix $X$ being Hermitian, we use Lem.~\ref{lem:Resolvent_pap} to conclude that $\alpha(1_D) \in \bb{N}^*$.
    
    \begin{figure}
        \centering
        \includegraphics[width = \textwidth]{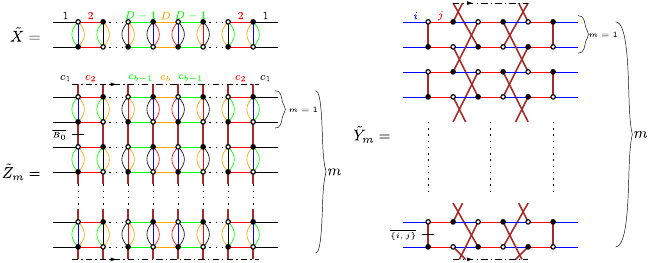}
        \caption{Matrices used in the proof of Lem.~\ref{lem:alpha_integer_valued} inspired from the matrices considered in Ref.~\cite{Carrozza:2026qcf}, see Fig.~34. In the upper left corner, the matrix $\tilde{X}$, here represented for $D=5$. On the right, the matrix $\tilde{Y}_m$, where the upper edges are connected to the bottom ones, as encoded by the dashed-dotted lines. The thick brown line represents all colors in $\bar{\paa{i,j}}$. In the bottom left corner, the matrix $\tilde{Z}_m$ is presented, with $B_0 = \paa{c_1,\dots,c_b}$. Again, the upper edges are connected to the bottom ones as encoded by the dashed-dotted lines. The thick brown line represents all colors in $\bar{B_0}$. The matrix $\tilde{Z}_m$ is here drawn for $D \geq 6$.}
        \label{fig:proof_integ}
    \end{figure}
    
    \ 
    
    Let us now consider the case $2\leq \abs{B_0}<D$, assuming by induction hypothesis that: for any $B \subset \paa{1,\dots,D}$ with $\abs{B_0}< \abs{B} \leq D$, $\alpha(B) \in \bb{N}^*$.
    
    \medskip
    
    Let us first focus on the subcase $2 < \abs{B_0} < D$. Let $m,n \in \bb{N}^*$ and consider the open $D$-colored graph $\tilde Z_m$. Cyclically gluing $n$ copies of this graph (in the horizontal direction), we obtain the closed $D$-colored graph $G_n(\tilde Z_m)$. We then introduce the Hermitian positive operator of unit trace  
    \begin{equation}
         Z_m\eqdef \frac{\tr_{\tilde{Z}_m}\left( \ket{\psi} \right)}{\tr\left( \tr_{\tilde{Z}_m}\left( \ket{\psi} \right) \right)} = \frac{\tr_{\tilde{Z}_m}\left( \ket{\psi} \right)}{\tr_{G_1(\tilde{Z}_m)}\left( \ket{\psi} \right)}\,,   
    \end{equation}
which obeys
\begin{equation}
     \tr\left((Z_m)^n\right) = \frac{\tr_{G_n (\tilde Z_m)}\left( \ket{\psi} \right)}{\tr_{G_1 (\tilde Z_m)}\left( \ket{\psi} \right)^n}\,.
\end{equation}
    To evaluate this moment, one observes that the graph $G_n(\tilde Z_m)$ was designed to obey the following combinatorial relations: for any $B\subset\{1,\ldots, D\}$ with $\abs{B}\geq 2$, 
    \begin{equation}
    \kappa(G_n(\tilde Z_m)\vert_B) = 
        \begin{cases}
            m \,& \rm{ if } B = B_0 \,,\\
            1 \, & \rm{ if } B \supsetneq B_0\,, \\
            \mu_{B,m} n \,& \rm{ otherwise,}
        \end{cases}
    \end{equation}
    where $\mu_{B,m} \in \bb{N}^*$ is a constant that does not depend on $n$. As a result, one finds that:
    \begin{equation} 
     \tr\left((Z_m)^n\right) = \pac{\alpha(B_0)^m \prod_{B \supsetneq B_0} \alpha(B)}^{1-n} \,.
    \end{equation}
   Introducing the short-hand notation $b(B_0) \eqdef \displaystyle \prod_{B \supsetneq B_0} \alpha(B)$, we remark on the one hand that $b(B_0) \in \bb{N}^*$ (by the induction hypothesis). On the other hand, from Lem.~\ref{lem:Resolvent_pap}, we also have
    \begin{equation}\label{eq:a^m_b}
        \forall m \in \bb{N}^* \,, \qquad \alpha(B_0)^m b(B_0) \in \bb{N}^*\,.
    \end{equation}
    Since both $\alpha(B_0)^2 b(B_0)$ and $\alpha(B_0) b(B_0)$ are positive integers, their quotient is a positive rational number, namely $\alpha(B_0) \in \bb{Q}_+^*$. Hence, we can find positive integers $r, s \in \bb{N}^*$ with $\gcd(r,s)=1$ such that     
    \begin{equation}
        \alpha(B_0) = \frac{r}{s} \,.
    \end{equation}
    Having established this property, we can rewrite Eq.~\eqref{eq:a^m_b} as     
    \begin{equation} \label{eq:sdivr}
        \forall m \in \bb{N}^* \,, \qquad s^m \divides r^m b(B_0)\,.
    \end{equation}    
    Since $\gcd(r,s)=1$, this implies that $s^m \divides b(B_0)$ for any $m\in \mathbb{N}^*$, which can only hold if $s=1$. Hence, $\alpha(B_0)=r\in \mathbb{N}^*$.
    
    \medskip 
    
Finally, suppose that $\abs{B_0}=2$, \ie we have distinct $i,j \in \paa{1,\dots,D}$ such that $B_0 =\paa{i,j}$. Reasoning similarly as before, for any $m,n \in \mathbb{N}^*$, we let $\tilde Y_m$ denote the open $D$-colored graph shown on the right panel of Fig.~\ref{fig:proof_integ}, and define $G_n(\tilde{Y})\in \cG_D^{\conn}$ as the cyclic contraction of $n$ copies of that graph (in the horizontal direction, relative to the depiction of $\tilde Y_m$ in Fig.~\ref{fig:proof_integ}). This allows us to introduce the Hermitian positive operator of unit trace
\begin{equation}
         Y_m\eqdef \frac{\tr_{\tilde{Y}_m}\left( \ket{\psi} \right)}{\tr\left( \tr_{\tilde{Y}_m}\left( \ket{\psi} \right) \right)} = \frac{\tr_{\tilde{Y}_m}\left( \ket{\psi} \right)}{\tr_{G_1(\tilde{Y}_m)}\left( \ket{\psi} \right)}\,,   
\end{equation}
whose moment of order $n$ is
\begin{equation}\label{eq:def_Ym}
     \tr\left((Y_m)^n\right) = \frac{\tr_{G_n (\tilde Y_m)}\left( \ket{\psi} \right)}{\tr_{G_1 (\tilde Y_m)}\left( \ket{\psi} \right)^n}\,.
\end{equation}
A combinatorial analysis shows that:
    \begin{equation}\label{eq:comb_G2nm}
        \kappa(G_n(\tilde Y_m)\vert_B) = 
        \begin{cases}
            2m  &\rm{ if } B = \paa{i,j} \,,\\
            1 &\rm{ if } B \supsetneq \paa{i,j}\,, \\
            \mu_{B,m} n &\rm{ otherwise}\,,
        \end{cases}
    \end{equation}
    with $\mu_{B,m} \in \bb{N}^*$ independent from $n$. The normalization in \eqref{eq:def_Ym} will eliminate all contributions from bubbles whose number of connected components is linear in $n$, thus leading to
    \begin{equation}
        \tr\left((Y_m)^n\right) = \pac{\alpha(\paa{i,j})^{2m} b(\paa{i,j})}^{1-n} \qquad \rm{with} \quad b(\paa{i,j}) \eqdef \prod_{B \supsetneq \paa{i,j}} \alpha(B)\,.
    \end{equation}
    By the induction hypothesis, we have that $b(\paa{i,j}) \in \bb{N}^*$. Reasoning as in the previous paragraph, it follows that $\alpha(\paa{i,j})^2 \in \bb{N}^*$. To prove that $\alpha(\paa{i,j})$ itself is an integer we will first prove that it is rational, and then conclude.    
    
    Let us prove that $\alpha(\paa{i,j}) \in \bb{Q}_+^*$, by invoking two additional families of $D$-colored graphs. For the first one, we start by partitioning the Hilbert space $\H = \H_1 \otimes \cdots \ot \H_D$ into three subsystems as
    \begin{equation}
        \H = \H_i \ot \H_j \ot \pa{\bigotimes_{c \neq i,j} \H_c}  = \H_i \ot \H_j \ot \H_{\bar{\paa{i,j}}} \,.
    \end{equation}
    We then consider the moments of the partial transpose of $\rho_{ij} \eqdef \Tr_{\H_{\bar{\paa{i,j}}}}(\ket{\psi}\bra{\psi})$ which is a Hermitian matrix. We recall that these moments were associated with the family of colored graphs $\{\PT_k^{(\bar{\paa{i,j}})}\}_{k\geq 1}$ in Sec.~\ref{sec:General} and Ref.~\cite{Carrozza:2026qcf}, \ie $\Tr((\rho_{ij}^{T_j})^k) = \Tr_{\PT_k^{(\bar{\paa{i,j}})}}(\ket{\psi})$ for any $k \geq 1$. To simplify our calculations, we introduce the notation:
    \begin{equation}
        n(i,j) \eqdef \prod_{\substack{B \subset\{1, \ldots, D\} \\  \abs{B}>2 \,,\, \{i,j\} \cap B\neq \emptyset}} \alpha(B) \,.
    \end{equation}
    We remark that this is a positive integer, as a consequence of our induction hypothesis. Furthermore, using the assumption of Lem.~\ref{lem:alpha_integer_valued}, the moments of $\rho_{ij}^{T_j}$ can be computed explicitly, leading to the following result: for any $k \in \mathbb{N}^*$, 
        \begin{equation}
        \Tr_{\PT_k^{(\bar{\paa{i,j}})}}(\ket{\psi}) = \begin{cases}
            \displaystyle\pac{\alpha(\paa{i,j})\cdot \left(  \prod_{c\in\{1, \ldots, D\} \setminus\{i,j\}} \alpha(\paa{c,i}) \alpha(\paa{c,j}) \right) \cdot 
            n(i,j)
            }^{1-k} & \rm{ if } k \rm{ is odd}, \\
            & \\
            \displaystyle\alpha(\paa{i,j}) \pac{\alpha(\paa{i,j}) \cdot \left(  \prod_{c\in\{1, \ldots, D\} \setminus\{i,j\}} \alpha(\paa{c,i}) \alpha(\paa{c,j}) \right) \cdot  
            n(i,j)
            }^{1-k} & \rm{ if } k \rm{ is even}.
        \end{cases}
    \end{equation}
    This result allows us to compute the resolvent $R(z) \eqdef \pa{z\id - \rho_{ij}^{T_j}}^{-1}$ (for $z\in \mathbb{C}$ taken in its convergent domain) explicitly. To simplify the presentation, let us use the short-hand 
    \begin{equation}
    N \eqdef \alpha(\paa{i,j})\cdot \left(  \prod_{c\in\{1, \ldots, D\} \setminus\{i,j\}} \alpha(\paa{c,i}) \alpha(\paa{c,j}) \right) \cdot n(i,j)\,.   
    \end{equation}
    A straightforward computation of the resolvent yields (see the proof of Thm.~\ref{thm:second_charac} as well as Rem.~\ref{rem:resolvent_PT} for a more detailed account of this calculation) 
    \begin{equation}
        \sum_{x \in \spec(\rho_{ij}^{T_j}) \setminus\{0\}} \frac{x m_x}{z- x} = \tr\left( \rho_{ij}^{T_j} R\left( z\right) \right)= \sum_{k=1}^{+\infty} \frac{\Tr\left( (\rho_{ij}^{T_j})^k\right)}{z^k} = \frac{1+ \alpha(\{i,j\})}{2} \frac{1}{z- \frac{1}{N}} + \frac{1- \alpha(\{i,j\})}{2} \frac{1}{z+ \frac{1}{N}} \,,
    \end{equation}
where $m_x \in \mathbb{N}^*$ denotes the multiplicity of the eigenvalue $x \in \spec(\rho_{ij}^{T_j})$. By identification, it follows that $\spec(\rho_{ij}^{T_j})\setminus \{0\}=\{\frac{1}{N}, \frac{-1}{N}\}$, and
\begin{equation}
    m_{\frac{1}{N}}= \frac{\left(\alpha(\{i,j\})+1\right) N}{2}\,, \qquad m_{-\frac{1}{N}}= \frac{\left(\alpha(\{i,j\})-1\right) N}{2}\,.
\end{equation}
This allows us to conclude that 
\begin{align}
    N &= m_{\frac{1}{N}}-m_{-\frac{1}{N}} \in \mathbb{N}^*\,, \\
    \alpha\left( \{i,j\}\right) N &= m_{\frac{1}{N}}+m_{-\frac{1}{N}} \in \mathbb{N}^*\,, 
\end{align}
and, therefore, that
\begin{equation}
    \alpha\left( \{i,j\}\right) = \frac{\alpha\left( \{i,j\}\right) N}{N}\in \mathbb{Q}_+^*\,.
\end{equation}

    Now that we know $\alpha(\paa{i,j})$ to be rational, we can find $r,s \in \mathbb{N}^*$ such that: $\alpha(\paa{i,j})= \frac{r}{s}$ and $\gcd(r,s)=1$. It follows that $\frac{r^2}{s^2} = \alpha(\paa{i,j})^2\in \mathbb{N}^*$, and given that $\gcd(r^2,s^2)$, this implies $s=1$ \ie $\alpha(\paa{i,j})=r \in \mathbb{N}^*$.
    
\end{proof}

\addcontentsline{toc}{section}{References}

\printbibliography

\end{document}